\documentclass[10pt]{paper} 

\usepackage{etex}

\usepackage[english]{babel}

\usepackage{authblk}

\usepackage{amsmath}
\usepackage{amsthm}
\usepackage{amssymb}
\usepackage{graphicx}
\usepackage{xcolor}
\usepackage{url}
\usepackage{soul}
\usepackage{acronym}
\usepackage{hyperref}
\usepackage{proof}
\usepackage{wasysym}
\usepackage{fancyvrb}
\usepackage{units}
\usepackage{fullpage}
\usepackage{enumitem}
\usepackage{listings}
\usepackage{float}
\floatstyle{plaintop}
\newfloat{program}{tbhp}{lop}
\floatname{program}{Listing}

\usepackage{pgf}
\usepackage{tikz}
\usetikzlibrary{arrows,automata,shapes}

\newtheorem{theorem}{Theorem}[section]
\newtheorem{lemma}[theorem]{Lemma}
\newtheorem{corollary}[theorem]{Corollary}

\title{A Rigorous Framework for Specification, Analysis and Enforcement of Access Control Policies}

\definecolor{grigiomoltochiaro}{gray}{0.97}
\definecolor{verde}{rgb}{0,1,0}
\newcommand{\semLbl}[1]{S\textrm{-}#1}
\newcommand{\conLbl}[1]{T\textrm{-}#1}

\newcommand{\sr}[1]{#1}

\newcommand{\define}{\triangleq}

\newcommand{\ALGOInline}[3]
{
  \lstset{
    basicstyle=\scriptsize\ttfamily,
    stringstyle=\color{red},
    commentstyle=\color{verde},
    keywordstyle=\color{blue}\bfseries\underbar,
    frame={tb},
    numbers=none,
    numberstyle=\tiny,
    numbersep=5pt,
    tabsize=4,
    language=Java,
    showtabs=false,
    showstringspaces=false,
    identifierstyle=,
    breaklines,
    literate={~=}{{$\neq$}}2{<=}{{$\leq$}}2{>=}{{$\geq$}}2{&}{{$\&$}}2,
    backgroundcolor=\color{grigiomoltochiaro}
    }
  \lstinputlisting{#3}
}

\newcommand{\ie}{i.e.~}
\newcommand{\eg}{e.g.~}

\newcommand{\xacml}{\ac{XACML}}
\newcommand{\facpl}{\ac{FACPL}}

\newcommand{\pdp}{\ac{PDP}}
\newcommand{\pep}{\ac{PEP}}

\newcommand{\pr}{\ac{PR}}

\newcommand{\denyOver}{\x{d}\textrm{-}\x{over}_{\delta}}
\newcommand{\permitOver}{\x{p}\textrm{-}\x{over}_{\delta}}
\newcommand{\permitUnless}{\x{p}\textrm{-}\x{unless}\textrm{-}\x{d}_{\delta}}
\newcommand{\denyUnless}{\x{d}\textrm{-}\x{unless}\textrm{-}\x{p}_{\delta}}
\newcommand{\onlyOneApp}{\x{one}\textrm{-}\x{app}_{\delta}}
\newcommand{\firstApp}{\x{first}\textrm{-}\x{app}_{\delta}}
\newcommand{\weakCon}{\x{weak}\textrm{-}\x{con}_{\delta}}
\newcommand{\strongCon}{\x{strong}\textrm{-}\x{con}_{\delta}}

\newcommand{\denyOverO}[1]{\x{d}\textrm{-}\x{over}_{\x{#1}}}
\newcommand{\permitOverO}[1]{\x{p}\textrm{-}\x{over}_{\x{#1}}}
\newcommand{\permitUnlessO}[1]{\x{p}\textrm{-}\x{unless}\textrm{-}\x{d}_{\x{#1}}}
\newcommand{\denyUnlessO}[1]{\x{d}\textrm{-}\x{unless}\textrm{-}\x{p}_{\x{#1}}}
\newcommand{\onlyOneAppO}[1]{\x{one}\textrm{-}\x{app}_{\x{#1}}}
\newcommand{\firstAppO}[1]{\x{first}\textrm{-}\x{app}_{\x{#1}}}
\newcommand{\weakConO}[1]{\x{weak}\textrm{-}\x{con}_{\x{#1}}}
\newcommand{\strongConO}[1]{\x{strong}\textrm{-}\x{con}_{\x{#1}}}

\newcommand{\obM}{\x{m}}
\newcommand{\obO}{\x{o}}
\newcommand{\oblp}{\x{obl}\textrm{-}\x{p}}
\newcommand{\obld}{\x{obl}\textrm{-}\x{d}}

\newcommand{\based}{\x{base}}

\newcommand{\denyBiased}{\x{deny}\textrm{-}\x{biased}}
\newcommand{\permitBiased}{\x{permit}\textrm{-}\x{biased}}

\newcommand{\alice}{\x{Alice}}

\newcommand{\pdpPol}[2]{\{ #1 \,  \, #2 \}}
\newcommand{\obl}[1]{\, PepAction( #1^* ) }

\newcommand{\streq}{\x{equal}}
\newcommand{\exprOp}{\x{eop}}
\newcommand{\exprOperator}{\x{op}}

\newcommand{\attribute}[2]{( #1 , #2 )}

\newcommand{\x}[1]{{\sf #1}}

\newcommand{\Sep}{\ \mid\ }

\newcommand{\policyBegin}{{\bf{\{}}}

\newcommand{\policiesBegin}{\x{policies:}\,}
\newcommand{\targetBegin}{\x{target:}\,}

\newcommand{\policyEnd}{{\bf{\}}}}

\newcommand{\oblspBegin}{\x{\oblp:}\,}
\newcommand{\oblsdBegin}{\x{\obld:}\,}

\newcommand{\oblBegin}{{\bf{[}}}
\newcommand{\oblEnd}{{\bf{]}}}

\newcommand{\ruleOpt}[1]{{\bf{(}}#1{\bf{)}}}

\newcommand*\lfrac[2]{{}_{#1}\!\backslash\!^{#2}} 

\newcommand{\permit}{\x{permit}}
\newcommand{\deny}{\x{deny}}
\newcommand{\notApp}{\x{not}\textrm{-}\x{app}}
\newcommand{\indet}{\x{indet}}
\newcommand{\excpt}{\perp}
\newcommand{\err}{\x{error}}

\newcommand{\indetD}{\x{indetD}}
\newcommand{\indetP}{\x{indetP}}
\newcommand{\indetDP}{\x{indetDP}}

\newcommand{\true}{\x{true}} 
\newcommand{\false}{\x{false}} 

\newcommand{\pepSemR}[1]{(\!( #1 )\!)}
\newcommand{\concat}{{\bullet}}

\newcommand{\denSemF}[1]{\mathcal{ #1 }}

\newcommand{\denSem}[2]{[\![ #1 ]\!] #2}

\newcommand{\policySem}[2]{\denSemF{P}\denSem{#1}{#2}}
\newcommand{\algSem}[2]{\denSemF{A}\denSem{#1}{#2}}
\newcommand{\exprSem}[2]{\denSemF{E}\denSem{#1}{#2}}
\newcommand{\oblSem}[2]{\denSemF{O}\denSem{#1}{#2}}

\newcommand{\pdpSem}[2]{\denSemF{P}dp\denSem{#1}{#2}}
\newcommand{\pepSem}[2]{\denSemF{E}A\denSem{#1}{#2}}
\newcommand{\pasSem}[1]{\denSemF{P}as\denSem{#1}{}}
\newcommand{\reqSemS}[2]{\denSemF{R}\denSem{#1}{#2}}  
\newcommand{\reqSem}[1]{\denSemF{R}\denSem{#1}{}} 

\newcommand{\expr}{\mathit{expr}}
\newcommand{\effect}{\mathit{e}}
\newcommand{\ob}{\mathit{o}}
\newcommand{\obType}{\mathit{t}}
\newcommand{\fo}{\mathit{i\!o}}
\newcommand{\foS}{\mathit{i\!o}^*}
\newcommand{\rSyntax}{\mathit{req}} 
\newcommand{\req}{\mathit{r}} 
\newcommand{\policy}{\mathit{p}}
\newcommand{\algSyntax}{\mathit{a}}
\newcommand{\pdpRes}{\mathit{res}}
\newcommand{\dec}{\mathit{dec}} 
\newcommand{\extVal}{\mathit{w}} 
\newcommand{\enfAlg}{\mathit{ea}}
\newcommand{\pdpSyntax}{\mathit{pdp}}
\newcommand{\name}{\mathit{n}}
\newcommand{\val}{\mathit{v}}
\newcommand{\pepAction}{\mathit{pepAct}}

\newcommand{\algNT}{\mathit{Alg}} 
\newcommand{\algName}{\mathsf{alg}} 

\newcommand{\alg}[1]{\mathsf{alg}_{#1}}

\newcommand{\algOp}{\otimes \mathsf{alg}}
\newcommand{\algOpAlg}[1]{\otimes #1}

\newcommand{\all}{\x{all}}
\newcommand{\greedy}{\x{greedy}}
\newcommand{\isFinal}[2]{\mathit{isFinal}_{#1}({#2})}
\newcommand{\isFinalPred}[1]{\mathit{isFinal}_{#1}}

\newcommand{\Req}{R}

\newcommand{\Ext}[1]{\mathit{Ext}(#1)}

\newcommand{\EvalS}{\mathtt{eval}}
\newcommand{\Eval}[2]{#1\  \EvalS\ #2}
\newcommand{\May}[2]{#1\ \EvalS_{\mathtt{may}}\ #2}
\newcommand{\Must}[2]{#1\ \EvalS_{\mathtt{must}}\ #2}

\newcommand{\NSatLit}{\mathtt{unsat}}
\newcommand{\SatLit}{\mathtt{sat}}
\newcommand{\Sat}[1]{#1 \ \SatLit\ }
\newcommand{\NSat}[1]{#1 \ \NSatLit\ }

\newcommand{\complete}{\mathtt{complete}}
\newcommand{\disjoint}{\mathtt{disjoint}}
\newcommand{\cover}{\mathtt{cover}}

\newcommand{\const}{\mathit{c}}
\newcommand{\smtconst}{\mathit{smtlib\textrm{-}c}}

\newcommand{\proj}[1]{\downarrow_{#1}}

\newcommand{\opC}{\x{cop}}
\newcommand{\isBot}[1]{\mathtt{isMiss}(#1)}
\newcommand{\isErr}[1]{\mathtt{isErr}(#1)}
\newcommand{\fnot}{\dot{\neg}}
\newcommand{\fand}{\dot{\wedge}}
\newcommand{\for}{\dot{\vee}}

\newcommand{\boolT}{\mathit{Bool}}
\newcommand{\stringT}{\mathit{String}}
\newcommand{\doubleT}{\mathit{Double}}
\newcommand{\dateT}{\mathit{Date}}
\newcommand{\setT}{2^{\mathit{Value}}}
\newcommand{\type}{\mathit{T}}
\newcommand{\typeVar}{\mathit{X}}
\newcommand{\nvType}{\mathit{U}}

\newcommand{\constrSem}{\mathcal{C}}
\newcommand{\cs}[1]{\constrSem[\![{#1}]\!]{\req}}

\newcommand{\translSymbol}{\mathcal{T}} 
\newcommand{\transFunct}[1]{\translSymbol_{#1}} 
\newcommand{\transl}[2]{\transFunct{#1}\{\!|#2|\!\}} 
\newcommand{\translExpr}[1]{\transl{E}{#1}}

\newcommand{\translObl}[1]{\transl{Ob}{#1}}
\newcommand{\translPol}[1]{\transl{P}{#1}}
\newcommand{\translAlg}[1]{\transl{A}{#1}}

\newcommand{\isBool}[1]{\mathtt{isBool}(#1)}

\newcommand{\AT}{PCT}

\newcommand{\Fval}{\texttt{TValue}}
\newcommand{\FvalField}{\texttt{val}}
\newcommand{\FvalErr}{\texttt{err}}
\newcommand{\FvalBot}{\texttt{miss}}

\author[1]{\normalsize Andrea~Margheri}

\author[2]{Massimiliano~Masi}

\author[3]{Rosario~Pugliese}

\author[4]{Francesco~Tiezzi}

\affil[1]{Electronics and Computer Science, University of Southampton}

\affil[2]{Tiani ``Spirit'' GmbH}

\affil[3]{Universit\`a degli Studi di Firenze}

\affil[4]{Universit\`a di Camerino}

\date{}

\begin{document}

\maketitle

\begin{abstract}
Access control systems are widely used means for the protection of computing systems. They are defined in terms of access control policies regulating the accesses to system resources. In this paper, we introduce a formally-defined, fully-implemented framework for specification, analysis and enforcement of attribute-based access control policies. The framework rests on FACPL, a language with a compact, yet expressive, syntax for specification of real-world access control policies and with a rigorously defined denotational semantics. The framework enables the automatic verification of properties regarding both the authorisations enforced by single policies and the relationships among multiple policies. Effectiveness and performance of the analysis rely on a semantic-preserving representation of FACPL policies in terms of SMT formulae and on the use of efficient SMT solvers. Our analysis approach explicitly addresses some crucial aspects of policy evaluation, as \eg missing attributes, erroneous values and obligations, which are instead overlooked in other proposals. The framework is supported by Java-based tools, among which an Eclipse-based IDE offering a tailored development and analysis environment for FACPL policies and a Java library for policy enforcement. We illustrate the framework and its formal ingredients by means of an e-Health case study, while its effectiveness is assessed by means of performance stress tests and experiments on a well-established benchmark.
\end{abstract}

%

\section{Introduction}
\label{sec:intro}

Nowadays computing systems have pervaded every daily activity and prompted the proliferation of several innovative services and applications. These modern distributed systems manage a huge amount of data that, due to its importance and societal impact, has brought out security issues of paramount importance. Controlling the access to system resources is thus crucial to prevent unauthorised accesses that could jeopardise trustworthiness of data. 

This has prompted an increasing research interest towards access control systems, which are the first line of defence for the protection of computing systems. They are defined by \emph{rules} that establish under which conditions a subject's \emph{request} for accessing a resource has to be permitted or denied. In practice, it amounts to restrict physical and logical access rights of subjects to system resources. 

Access control is a broad field, covering several different approaches, using different technologies and involving various degrees of complexity.  Since the first applications in operating systems, to the more recent ones in distributed systems, many access control approaches have been proposed. Traditional approaches are based on the identity of subjects, either directly \,--\,e.g., Access Control Matrix~\cite{Lampson74} \,--\, or through predefined features, such as roles or groups --\,e.g., Role-Based Access Control (RBAC~\cite{rbac}). These approaches are however inadequate for dealing with modern distributed systems, as they suffer from scalability and interoperability issues. Moreover, they cannot easily encompass information representing the evaluation context, as \eg system status or current time. An alternative approach that permits to overcome these problems is Attribute-Based Access Control (ABAC)~\cite{HuKF15}. Here, the rules are based on \emph{attributes}, which represent arbitrary security-relevant information exposed by the system, the involved subjects, the action to be performed, or by any other entity of the evaluation context relevant to the rules at hand. Thus, ABAC permits defining fine-grained, flexible and context-aware access control rules that are expressive enough to uniformly represent all the other approaches~\cite{JinKS12}. Attribute-based rules are typically hierarchically structured and paired with strategies for resolving possible conflicting authorisation results. These structured specifications are called \emph{policies}; from this name derives the terminology Policy-Based Access Control (PBAC) \cite{nistsurvey}, sometimes used in place of ABAC.

Many languages have been proposed for the specification of access control policies (see, e.g.,~\cite{Han2012477} for a survey). Among the proposed languages, in the authors' knowledge, the OASIS standard \xacml~\cite{XACML3} is the best-known one. Due to its XML-based syntax and the advanced access control features it provides, \xacml\ is commonly used in many real-world systems, e.g., in service-oriented ones. However, the management of real access control policies is in practice cumbersome and error-prone, and should be supported by rigorous analysis techniques. Unfortunately, \xacml\ is generally acknowledged as lacking of a formally defined semantics (see, e.g.,~\cite{RaoLBLL09,CramptonM12,RamliNN12,ArkoudasCC14}), which makes it difficult the specification and realisation of analysis techniques. 

To tackle these difficulties, we introduce a formally-defined, fully-implemented framework, based on \facpl, supporting developers in the specification, analysis and enforcement of access control policies. 

\subsection*{The FACPL-based Access Control Framework}

The \facpl\ language defines a core, yet expressive, syntax for specification of high-level access control policies. It is inspired by \xacml\ (with which it shares the main traits of the policy structure and some terminology), but it refines some aspects of \xacml\ and introduces novel features from the access control literature. Evaluation of \facpl\ policies is formalised by a denotational semantics, which clarifies intricate aspects of access controls like, e.g., management of missing attributes (\ie attributes controlled by a policy but not provided by the request to authorise) and formalisation of combining algorithms (\ie strategies to resolve conflictual decisions that policy evaluation can generate).

The analysis functionalities offered by our framework enable verification of two main groups of properties of \facpl\ policies. The \emph{authorisation properties} permit to statically reason on the result of the evaluation of a policy with respect to a specific request, by also considering additional attributes that can be possibly introduced in the request at run-time and that might lead to unexpected authorisations. Instead, the \emph{structural properties} permit to statically reason on the whole set of results of the evaluation of one or more policies and can be exploited, e.g., to implement maintenance and \emph{change-impact analysis}~\cite{FislerKMT05} techniques.

The verification of these properties requires extensive checks on very large (possibly infinite) amounts of requests, hence support through software tools is essential. As no off-the-shelf analysis tool directly takes \facpl\ specifications in input, our framework exploits a constraint formalism that permits uniformly representing policy elements and enabling automated analysis. The constraint formalism we introduce is based on Satisfiability Modulo Theories (SMT) formulae, that is formulae defining satisfiability problems involving multiple theories, e.g. boolean and linear arithmetic ones. The relevant progress made in the development of automatic SMT solvers has led SMT to be extensively employed in diverse analysis applications~\cite{MouraB11}, even for access control policies~\cite{ArkoudasCC14,TurkmenHRZ15}. In practice, SMT-based approaches are more effective than many other ones, like e.g. decision diagrams~\cite{FislerKMT05} or description logic~\cite{KolovskiHP07}. The soundness of our analysis techniques is guaranteed by the correspondence, which we formally prove, between the semantics of \facpl\ policies and that of their constraint-based representations.

Our framework is supported by a Java-based software \emph{toolchain}. The key software tool is an Eclipse-based IDE that offers a tailored development and analysis environment for \facpl\ policies. Specifically, it helps access control policy developers in the tasks of specification, analysis and enforcement of policies by providing, e.g., static checks on the code and automatic generation of runnable SMT and Java code. The evaluation of the SMT code relies on the Z3 solver~\cite{MouraB08}, while the policy enforcement relies on  an expressly developed Java library.

\subsection*{Contributions}

The main contribution of this paper is the development of a comprehensive methodology supporting the whole life-cycle of access control policies, from their specification and analysis to their enforcement. Each ingredient of the methodology is formally introduced in this paper, together with its tool implementation. The tools allow access control system developers to use formally-defined functionalities without requiring them to be familiar with formal methods. 

Our methodology enhances the proposals from the literature to different extents, in order to provide a single framework where all the relevant functionalities are expressed and formalised in a uniform manner. Indeed, \xacml\ does not come with any formal specification and, hence, analysis; the formally-grounded proposals in~\cite{Jajodia97alogical,CramptonM12,RamliNN12,CramptonW15} do not offer any supporting tools; the SMT-based analysis proposals in~\cite{ArkoudasCC14,TurkmenHRZ15} do not support some crucial features (e.g., missing attributes and obligation instantiation). Detailed comparisons with the relevant literature are in Section~\ref{sec:relatedWork}.

Our aim is to design an expressive language whose formal foundations enable tool-supported analysis, rather than to face \xacml\ semantic issues or supersede it. Further contributions of this paper are summarised below.
\begin{itemize}
\item The \facpl\ semantics manages missing attributes in a way similar to~\cite{CramptonM12,CramptonW15} and extends it with explicit error management. 
\item The formalisation of combining algorithms extends that of~\cite{LiWQBRLL09} with explicit combination of obligations and with different instantiation strategies. 
\item The authorisation properties explicitly take into account the non-monotonicity issue of policy evaluation~\cite{TschantzK06} by appropriately employing the request extensions set of~\cite{CramptonMZ15} for property formalisation.
\item The main structural properties of~\cite{FislerKMT05} and~\cite{KolovskiHP07} are uniformly formalised in terms of policy semantics.  
\item The constraint formalism defines a low-level, tool-in\-de\-pendent representation of attribute-based policies that is capable to deal with all issues regarding policy evaluation. 
\item The validation of the proposal is carried out through experiments on a standard benchmark in the field of access control tools, \ie the CONTINUE~\cite{Krishnamurthi03} case study.
\end{itemize}

\begin{figure*}[!t]
\centering{
\includegraphics[scale=.5]{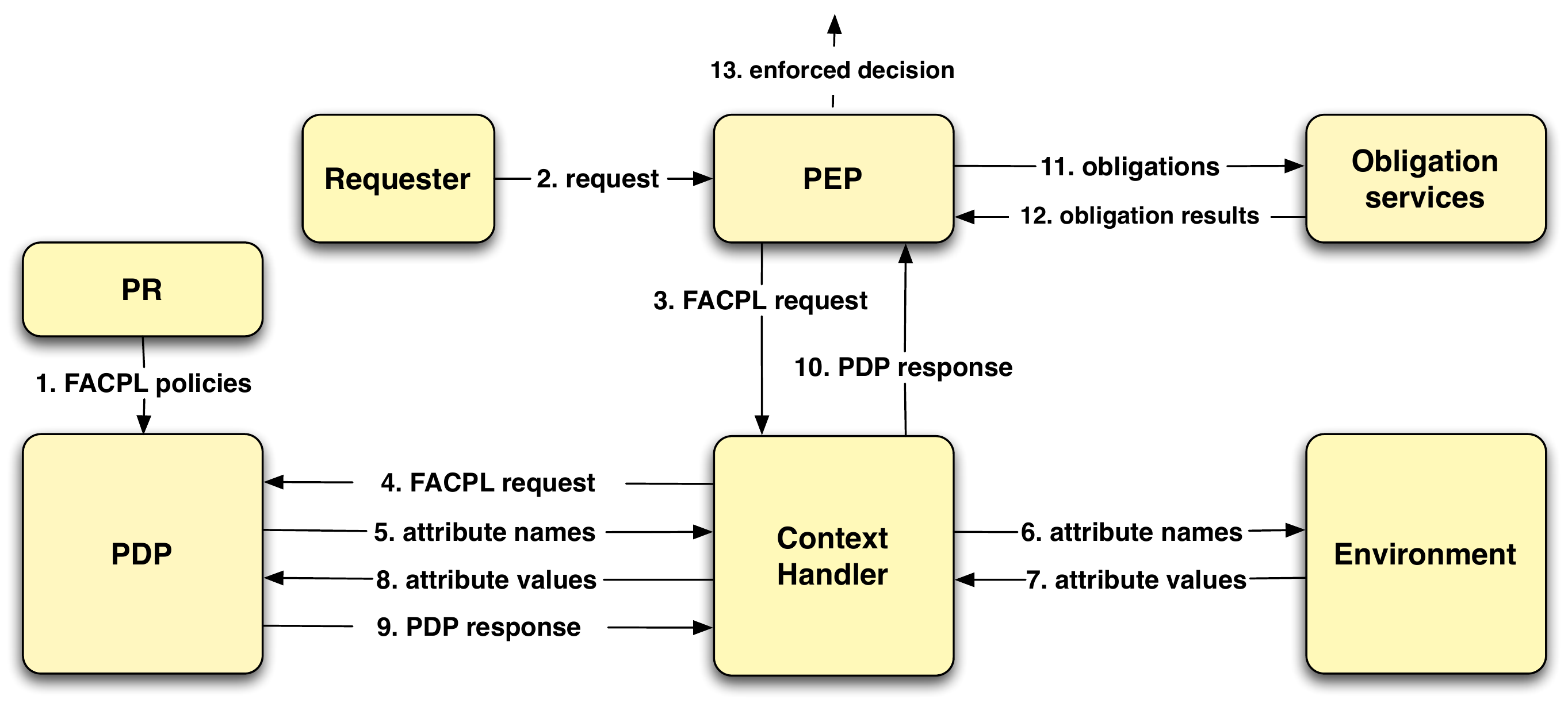}
\vspace*{-.3cm}
\caption{The \facpl\ evaluation process}
\label{fig:facplModel}}
\end{figure*}

This paper is a revised and extended version of~\cite{ESSOS,WWV15}. Besides significant revisions and extensions of syntax and semantics of \facpl\ (we refer to Section~\ref{sec:relatedWork} for a detailed comparison) this paper proposes a  complete development methodology for access control policies. Most of all, differently from previous works, we introduce a constraint-based representation of \facpl\ policies enabling the verification of various properties through SMT solvers.

\medskip
\noindent
\textit{Summary of the rest of the paper.}\ \ 
In Section~\ref{sec:model} we overview the \facpl\ evaluation process. In Section~\ref{sec:casestudy} we introduce an e-Health case study we use throughout the paper as a running example. In Section~\ref{sec:facplSyntax} we present the syntax of \facpl\ and its informal semantics, together with the \facpl-based specification of the case study. In Section~\ref{sec:formal_sem} we formally define the \facpl\ semantics. In Section~\ref{sec:constraint} we introduce the constraint formalism and the representation it enables of \facpl\ policies. In Section~\ref{sec:analysis} we introduce various properties for access control policies and their verification via SMT solvers. In Section~\ref{sec:tool} we outline the Java-based software toolchain. In Section~\ref{sec:relatedWork} we discuss the closest related work and, finally, in Section~\ref{sec:conclusions} we conclude and touch upon directions for future work. Appendixes~\ref{sec:appendixA} and~\ref{sec:appendix3} report, respectively, all the definitions for combining algorithms, and the proofs of the formal results.


\section{The FACPL Evaluation Process}
\label{sec:model}

The \facpl\ evaluation process of (access control) policies and requests is shown in Figure~\ref{fig:facplModel}. It defines the interactions, leading to the final authorisation decision, among three key components: the \pr, the \acf{PDP}\ and the \acf{PEP}. These entities and their interactions were introduced in~\cite{rfc2753} to define the evaluation process of policy-based systems. Each policy language, e.g. \xacml, has then tailored them according to its specific features.

The evaluation process assumes that system resources are paired with one or more \facpl\ policies, which define the credentials necessary to gain access to such resources. The \pr\ stores the policies and supplies them to the \pdp\ (step~1), which then decides if the access can be granted.

When \pep\ receives a request (step~2), the credentials contained in the request are encoded as a sequence of \emph{attribute} elements (i.e., name-value pairs representing arbitrary information relevant for evaluating the access request) forming a \facpl\ request (step~3). \pep s can have many different forms, \eg a gateway or a Web server. Therefore, this encoding allows policies and requests to be written and evaluated independently of their specific nature.

The \emph{context handler} sends the request to the \pdp\ (step~4), by possibly adding environmental attributes, e.g. request receiving time, that may be used in the evaluation.

The \pdp\ \emph{authorisation process} computes the \emph{PDP response} for the request by checking the attributes, that may belong either to the request or to the environment (steps 5-8), against the controls contained in the policies. The \pdp\ response (steps~9-10) contains an authorisation \emph{decision} and possibly some \emph{obligations}.

The decision is one among $\permit$, $\deny$, $\notApp$ and $\indet$\footnote{The \facpl\ supporting tools can handle the same extended indeterminate values dealt with by \xacml\ (see Section~\ref{sec:tool}). However, for the sake of presentation, in the formal specification of \facpl\ we only consider a single indeterminate value, rather than the whole set.}. The meaning of the first two ones is obvious, the third one means that there is no policy that applies to the request and the latter one means that some errors have occurred during the evaluation. Policies can automatically manage these errors by using operators that combine, according to different strategies, $\indet$ decisions with the others. 

Obligations are instead additional actions connected to the access control system that must be discharged by the \pep\ through appropriate \emph{obligation services} (steps~11-12). Obligations usually correspond to, e.g., updating a log file, sending a message or executing a command. The \emph{enforcement process} performed by the \pep\ determines the \emph{enforced decision} (step~13) on the basis of the result of obligations discharge. This decision could differ from that of the \pdp\ and is the outcome of the evaluation process.

It is worth noticing that the \facpl\ evaluation process guarantees separation of concerns among policies, their evaluation and the system itself. Among others, the main advantages it ensures are: (i) different types of requests can be handled, as the \pep\ can appropriately encode them in the format required by the \pdp; (ii) the \pdp\ can be placed in any point of the system, with the \pep\ acting as a gateway or a proxy; (iii) the \pr\ can be also instantiated to support dynamic, possibly regulated, modifications of policies\footnote{When \pr\ provides also support for the specification of administration controls on policy modifications, it is usually called \emph{Policy Administration Point} (PAP).}.


\section{An e-Health case study}
\label{sec:casestudy}

\begin{figure*}[!t]
\vspace*{-.2cm}
\centerline{\includegraphics[scale=.45]{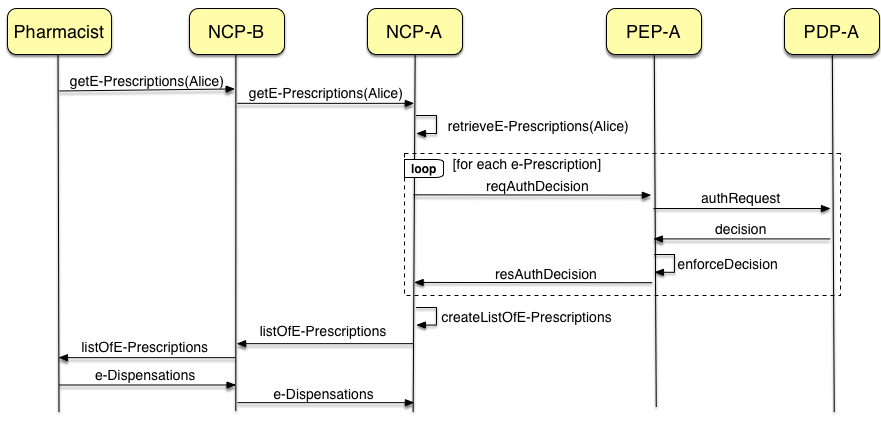}}
\vspace*{-.3cm}
\caption{e-Prescription service protocol}
\label{fig:ePrescription}
\end{figure*}

The case study we consider throughout this paper concerns the provision of e-Health services for exchanging private health data. In this context, we will show that access control policies expressed in \facpl\ can control accesses to health data in order to preserve data confidentiality and integrity.

The exchange of patients health data among European points of care (such as clinics, hospitals, pharmacies, etc.) has been pursued by the EU through the large scale pilot project epSOS (\url{http://www.epsos.eu}). The goal is to establish a suite of standardised data exchanging services for facilitating the cross-border interoperability~\cite{Kovac14} of the EU country healthcare systems and professionals (such as doctors, nurses, pharmacists, etc.), thus ultimately improving the effectiveness of healthcare treatments to EU citizens that are abroad. These services must respect a set of requirements in order to comply with country-specific legislations \cite{DataProtectionDirective,Art29} and to enforce the \emph{patient informed consent}, \ie the patients informed indications pertaining to personal data processing.

As a case study for this paper, we consider the \emph{electronic prescription} (e-Prescription) service. This service allows EU patients, while staying in a foreign country $\x{B}$ participating to the project, to have dispensed a medicine prescribed by a doctor in the country $\x{A}$ where the patient is insured. The protocol implemented by this service is illustrated in the message sequence diagram in Figure~\ref{fig:ePrescription}. The e-Prescription service helps $\x{pharmacist}$s in country $\x{B}$ to retrieve (and properly convert) e-Prescriptions from country $\x{A}$; this is due to trusted actors named National Contact Points (NCPs). Therefore, once a pharmacist has identified the patient ($\alice$), the remote access is requested to the local NCP ($\x{NCP}$-$\x{B}$), which in its own turn contacts the remote NCP ($\x{NCP}$-$\x{A}$)\footnote{For the sake of presentation, we abstract from the authentication process carried out by the pharmacist to ascertain the patient identity.}. The latter one retrieves the e-Prescriptions of the patient from the national infrastructure and, for each e-Prescription, performs through $\x{PEP}$-$\x{A}$ an authorisation check against the patient informed consent. In details, $\x{PEP}$-$\x{A}$ asks $\x{PDP}$-$\x{A}$ to evaluate the pharmacist request with respect to the e-Prescription and the policies expressing the patient consent. Once all $\x{decision}$s are enforced by $\x{PEP}$-$\x{A}$, $\x{NCP}$-$\x{A}$ creates the list of e-Prescriptions, by transcoding and translating them into the code system and language of the country $\x{B}$. Finally, the $\x{pharmacist}$ dispenses the medicine to the patient and updates the e-Prescription, \ie it returns e-Dispensation documents. 

Starting from the epSOS specifications, we deduced a set of business requirements concerning the e-Prescription service. To streamline the presentation, we explicitly report in Table~\ref{tab:privacy} all and only those requirements authorising some actions. Hence, every action not explicitly authorised is forbidden. For instance, it is not allowed to pharmacists to write e-Prescriptions, which is instead allowed to doctors exhibiting specific permissions. All the requirements are self-explanatory. We just want to point out that the first three requirements deal with access restrictions, while the other ones deal with additional functionalities that sophisticated access control systems, like the one we present, can provide.

\begin{table}[h]
\centering
\footnotesize
\caption{Requirements for the e-Prescription service}\label{tab:privacy}
\begin{tabular}{@{}l@{}|@{}l}
\\[-.4cm]
\hline
\sr{} \# & \multicolumn{1}{c}{\emph{Description}} \\
\hline
\sr{1} & Doctors with $\x{e}\textrm{-}\x{Pre}\textrm{-}\x{Read}$ and $\x{e}\textrm{-}\x{Pre}\textrm{-}\x{Write}$ permissions
can write e-Prescriptions\\
\sr{2} & Doctors with $\x{e}\textrm{-}\x{Pre}\textrm{-}\x{Read}$ permission can read
e-Prescriptions\\
\sr{3} & Pharmacists with $\x{e}\textrm{-}\x{Pre}\textrm{-}\x{Read}$ permission can read
e-Prescriptions\!\!\\
\sr{4} & Authorised user accesses must be recorded by the system\\
\sr{5} & Patients must be informed of unauthorised access
attempts\\
\sr{6} & Data exchanged should be compressed\\
\hline
\end{tabular}
\end{table}


\section{The FACPL Language}
\label{sec:facplSyntax}

In this section we present \facpl, the language we propose for defining high-level access control policies and requests. First, we introduce its syntax (Section~\ref{sec:policySyntax}). Then, we informally explain the semantics of its linguistic constructs (Section~\ref{sec:informal_sem}) and employ them to implement the access control system of the e-Health case study (Section~\ref{sec:facplEx}).


\begin{table*}[!t]
\caption{Syntax of \facpl} 
\vspace*{-.5cm}
\label{tab:facpl_syntax}
$$
\footnotesize
\begin{array}{@{\,}r@{\ \ }r@{\ }r@{\ \ }l@{\ }}
\hline
&&&\\[-.2cm]
{\textbf{Policy Auth. Systems}} &
\mathit{PAS} & ::= & \{ \,  \x{pep:} \, \mathit{EnfAlg}\ \ \x{pdp:}\, \mathit{PDP} \, \}
\\[.2cm]
{\textbf{Enforcement algorithms}} &
\mathit{EnfAlg}
& ::= & \based \Sep \denyBiased \Sep \permitBiased 
\\[.2cm]
{\textbf{Policy Decision Points}} &
\mathit{PDP} & ::= & \mathit{Policy} \Sep\ \pdpPol{\algNT\ }{\x{policies:} \, \mathit{Policy}^{+}}\
\\[.2cm]
{\textbf{Combining algorithms}} &
\algNT & ::= & \permitOver \Sep \denyOver \Sep \denyUnless \Sep \permitUnless 
\Sep \firstApp \Sep \onlyOneApp \Sep \\
&&&
\weakCon \Sep \strongCon 
\\[.2cm]
\textbf{Instantiation strategies} & \delta
& ::= & 
\greedy \Sep \all 
\\[.2cm]
{\textbf{Policies}} &
\mathit{Policy} & ::= &
\mathit{Rule}
\\
&& \mid &
\{ \algNT\ \ \x{target:} \, Expr\ \ \x{policies:} \, \mathit{Policy}^{+}  \ \ \x{\oblp:} \, \mathit{Obligation}^{*}\ \ \x{\obld:} \, \mathit{Obligation}^{*} \, \}
\\[.2cm]
{\textbf{Rules}} &
\mathit{Rule} & ::= &
\ruleOpt{\mathit{Effect}\ \ \x{target:} \, Expr\ \ \x{obl:} \, \mathit{Obligation}^{*} \, }
\\[.2cm]
{\textbf{Effects}} &
\mathit{Effect} & ::= & \permit \Sep \deny
\\[.2cm]
{\textbf{Obligations}} &
\mathit{Obligation} & ::= & [ \, \mathit{ObType} \ \ \obl{Expr} \, ]
\\[.2cm]
{\textbf{Obligation types}} &
\mathit{ObType} & ::= & \obM \Sep \obO
\\[.2cm]
\textbf{Expressions}&
\mathit{Expr} & ::= &
\mathit{Name} \Sep \mathit{Value}  
\Sep \x{and(\mathit{Expr}, \mathit{Expr})} \Sep \x{or(\mathit{Expr}, \mathit{Expr})} \Sep \x{not(\mathit{Expr})} 
\\
& & \mid &
 \x{equal(\mathit{Expr},\mathit{Expr})} 
\x{in}(\mathit{Expr}, \mathit{Expr}) 
\Sep \x{greater}\textrm{-}\x{than(\mathit{Expr},\mathit{Expr})} \Sep \x{add(\mathit{Expr} ,\mathit{Expr} )}\\
& & \mid & \x{subtract(\mathit{Expr} ,\mathit{Expr} )} \Sep \x{divide(\mathit{Expr} ,\mathit{Expr} )} \Sep \x{multiply(\mathit{Expr} ,\mathit{Expr} )} \\[.15cm]
\textbf{Attribute names} & 
\mathit{Name} & ::= & \mathit{Identifier}/\mathit{Identifier} \\[.2cm]
\textbf{Literal values} &
\mathit{Value} & ::= & \x{true} \mid \x{false} \mid \mathit{Double} \mid \mathit{String} \mid \mathit{Date}
\\[.2cm]
{\textbf{Requests}} &
\mathit{Request} & ::= & {\attribute{\mathit{Name}}{\mathit{Value}}}^{+} 
\\[.1cm]
\hline
\end{array}
$$
\vspace*{-.1cm}
\end{table*}

\subsection{Syntax}
\label{sec:policySyntax}

Intuitively, \facpl\ policies are hierarchically structured lists of elements containing controls on the value of attributes that should be provided by \facpl\ access requests. Together with $\permit$ or $\deny$ decisions, policies specify the combining algorithms to be used in their evaluation and the obligations for the enforcement process.

Formally, the syntax of \facpl\ is reported in Table~\ref{tab:facpl_syntax}. It is given through EBNF-like grammars, where as usual the symbol $?$ stands for optional items, $*$ for (possibly empty) sequences, and $+$ for non-empty sequences. 

A top-level term is a \emph{\ac{PAS}} encompassing the specifications of a \pep\ and a \pdp. The \pep\ is defined by the \emph{enforcement algorithm} applied for establishing how decisions have to be enforced, e.g. if only decisions $\permit$ and $\deny$ are admissible, or also $\notApp$ and $\indet$ can be returned. The \pdp\ is instead defined by a policy, or by a sequence of policies and an algorithm for combining the results of the evaluation of these policies.

A \emph{policy} is made of a sequence of fields separated by keywords. It can be either a basic authorisation \emph{rule}
or a \emph{policy set}
collecting rules and other policy sets, so that policies can be hierarchically structured. A rule specifies an \emph{effect}, that is the $\permit$ or $\deny$ decision returned when the rule is successfully evaluated, a \emph{target}, that is an expression indicating the set of access requests to which the rule applies, and a sequence of obligations, that is actions to be discharged by the enforcement process. A policy set specifies a target, a sequence of enclosed policies along with an algorithm for combining the results of their evaluation, and two sequences of obligations, one to be discharged if the resulting effect is $\permit$, the other if it is $\deny$. Obligation sequences may be empty, while policy sequences cannot.

An attribute $\mathit{name}$ refers to the literal value associated to the attribute. The name is structured in the form $\mathit{Iden}\-\mathit{ti}\-\mathit{fi}\-\mathit{er}$/$\mathit{Identifier}$, where the first identifier stands for a category name and the second for an attribute name. For example, the structured name $\x{subject/role}$ represents the value of the attribute $\x{role}$ within the category $\x{subject}$. Categories permit a fine-grained classification of attributes, varying from the usual categories of access control, i.e. \emph{subject}, \emph{resource} and \emph{action}, to possibly application-dependent ones. 

\emph{Expressions} are built from attribute names and \emph{literal} values, \ie booleans, doubles, strings, and dates, by using standard operators. As usual, string values are written as sequences of characters delimited by double quotes. 

\emph{Combining algorithms} offer different strategies to merge the decisions resulting from the evaluation of various policies (\eg the $\permitOver$ algorithm states that decision $\permit$ takes precedence over the others). They can be specialised by choosing different strategies for the instantiation of obligations (\eg the $\greedy$ strategy states that only the obligations resulting from the actually evaluated policies are returned). In the algorithm names, $\x{p}$ and $\x{d}$ are shortcuts for $\permit$ and $\deny$, respectively. 

An \emph{obligation} specifies 
a type, \ie mandatory ($\obM$) or optional ($\obO$), and identifier and arguments of an action to be performed by the \pep. The set of action identifiers accepted by the \pep\ can be chosen, from time to time, according to the specific application (therefore, $\mathit{PepAction}$ is intentionally left unspecified). Action arguments are  expressions.

A \emph{request} consists of a (non-empty) sequence of \emph{attributes}, i.e.~name-value pairs, that enumerate request credentials in the form of literal values. \emph{Multivalued attributes}, i.e. names associated to a set of values, are rendered as multiple attributes sharing the same name.

\begin{table*}[!t]
\vspace*{-.6cm}
\caption{Auxiliary syntax for \facpl\ responses} \label{tab:facpl_context_syntax}
\centering
$
\footnotesize
\begin{array}{@{\!}r@{\ }r@{\ }r@{\ }l@{\!}}
\hline
&&&\\[-.2cm]
{\textbf{\pdp\ responses }} &
\mathit{PDPResponse} & ::= & \langle \,\mathit{Decision} \ \ \ \mathit{IObligation}^* \rangle
\\[.2cm]
{\textbf{Decisions}} &
\mathit{Decision} & ::= & \permit \Sep\! \deny \Sep\! \notApp \Sep\! \indet
\\[.2cm]
{\textbf{Instantiated oblig.}} &
\mathit{IObligation} & ::= &  [ \, \mathit{ObType} \ \ \obl{\mathit{Value}} \, ]\\[.1cm]
\hline
\end{array}
$
\end{table*}

The responses resulting from the evaluation of \facpl\ requests are written using the auxiliary syntax reported in Table~\ref{tab:facpl_context_syntax}. The two-stage evaluation process described in Section~\ref{sec:model} produces two different kinds of responses: \emph{PDP\ responses} and \emph{decisions} (\ie responses by the \pep). The former ones, in case of decision $\permit$ and $\deny$, pair the decision with a (possibly empty) sequence of instantiated obligations. An \emph{instantiated obligation} is a pair made of a type (i.e., $\obM$ or $\obO$) and an action whose arguments are values.

To simplify notations, in the sequel we will omit the keyword preceding a sub-term generated by the grammar in Table~\ref{tab:facpl_syntax} whenever the sub-term is missing or is the expression $\x{true}$. Thus, e.g., the rule $\ruleOpt{\deny\ \ \x{target:}\ \x{true}\ \ \x{obl:}\, }$ will be simply written as $\ruleOpt{\deny}$. Moreover, when in the $\mathit{PDPResponse}$ the sequence of instantiated obligations is empty, we sometimes write $\mathit{Decision}$ instead of $\langle\mathit{Decision} \rangle$.

\subsection{Informal Semantics}
\label{sec:informal_sem}

We now informally explain how the \facpl\ linguistic constructs are dealt with in the evaluation process of access requests described in Section~\ref{sec:model}. We first present the \pdp\ authorisation process and then the \pep\ enforcement process.

When the \pdp\ receives an access request, first it evaluates the request on the basis of the available policies. Then, it determines the resulting decision by combining the decisions returned by these policies through the top-level combining algorithm.

The evaluation of a policy with respect to a request starts by checking its applicability to the request, which is done by evaluating the expression defining its target. Let us suppose that the applicability holds, \ie the expression evaluates to $\true$. In case of rules, the rule effect is returned. In case of policy sets, the result is obtained by evaluating the contained policies and combining their evaluation results through the specified algorithm. In both cases, the evaluation ends with the instantiation of the enclosed obligations. Let us suppose now that the applicability does not hold. If the expression evaluates to $\false$, the policy evaluation returns $\notApp$, while if the expression returns an error or a non-boolean value, the policy evaluation returns $\indet$. Clearly, the target of enclosed policies may refine that of the enclosing ones, while a policy with target expression $\true$ (resp., $\false$) applies to all (resp., no) requests.

Evaluating expressions amounts to apply operators and to \emph{resolve} the attribute names occurring within, that is to determine the value corresponding to each such name. This value can either be contained in the request or retrieved from the environment by the context handler (steps~\mbox{5-8} in Figure~\ref{fig:facplModel}). Thus, if an attribute with that name is missing in the request and its retrieval by the context handler fails, the special value $\excpt$ is returned. Taking the value $\excpt$ apart from errors permits both carefully managing those requests only containing a limited set of attributes and reasoning on the role of missing attributes in policy evaluation (see Section~\ref{sec:analysis} for details).

It is worth noticing that the syntax of policies, and in particular that of attribute names and expressions, does not consider types. Indeed, we want a policy to provide a response to any request, not only to those complying with the expected type of (the values referred by) the attribute names controlled by the policy. Since we do not filter requests on the base of the type of their attributes, we cannot in general statically ensure that expressions within policies are well-typed. Consequently, errors will be generated at evaluation-time , and possibly managed, when expression operators are applied to arguments of unexpected type.

Indeed, the evaluation of expressions takes into account the types of the operators' arguments, and possibly returns the special values $\excpt$ and $\err$. In details, if the arguments are of the expected type, the operator is applied, else, \ie at least one argument is $\err$, $\err$ is returned; otherwise, \ie at least one argument is $\excpt$ and none is $\err$, $\excpt$ is returned. The operators $\x{and}$ and $\x{or}$ implement a different treatment of these special values. Specifically, $\x{and}$ returns $\true$ if both operands are $\true$, $\false$ if at least one operand is $\false$, $\excpt$ if at least one operand is $\excpt$ and none is $\false$ or $\err$, and $\err$ otherwise (\eg when an operand is not a boolean value). The operator $\x{or}$ is the dual of $\x{and}$. Hence, $\x{and}$ and $\x{or}$ may mask $\excpt$ and $\err$. Instead, the unary operator $\x{not}$ only swaps values $\true$ and $\false$ and leaves $\excpt$ and $\err$ unchanged. In the rest, we use operators $\x{and}$ and $\x{or}$ in infix notation, and assume that they are commutative and associative, and that operator $\x{and}$ takes precedence over $\x{or}$.

The evaluation of a rule ends with the instantiation of all the enclosed obligations, while that of a policy set ends with the instantiation of all the obligations in the sequence corresponding to the decision calculated for the policy. The instantiation of an obligation consists in evaluating every expression argument of the enclosed action. If an error occurs, the policy decision is changed to $\indet$. Otherwise, the instantiated obligations are paired with the policy decision to form the \pdp\ response. 

Evaluating a policy set requires the application of the specified algorithm for combining the decisions resulting from the evaluation of various policies and, thus, resolving possible conflicts, \eg whenever both decisions $\permit$ and $\deny$ occur. Given a sequence of policies in input, the combining algorithms prescribe the sequential evaluation of the given policies and behave as follows:
\begin{itemize}
\item $\permitOver$ ($\denyOver$ is specular): if the evaluation of a policy returns $\permit$, then the result is $\permit$. In other words, $\permit$ takes precedence, regardless of the result of any other policy. Instead, if at least one policy returns $\deny$ and all others return $\notApp$ or $\deny$, then the result is $\deny$. If all policies return $\notApp$, then the result is $\notApp$. In the remaining cases, the result is $\indet$.
\item $\denyUnless$ ($\permitUnless$ is specular): similarly to $\permitOver$, this algorithm gives precedence to $\permit$ over $\deny$, but it always returns $\deny$ in all the other cases.
\item $\firstApp$: the algorithm returns the evaluation result of the first policy in the sequence that does not return $\notApp$, otherwise the result is $\notApp$. 
\item $\onlyOneApp$: when exactly one policy is applicable, the result of the algorithm is that of the applicable policy. If no policy applies, the algorithm returns $\notApp$, while if more than one policy is applicable, it returns $\indet$.
\item $\weakCon$: the algorithm returns $\permit$ (resp., $\deny$) if some policies return $\permit$ (resp., $\deny$) and no other policy returns $\deny$ (resp., $\permit$); if both decisions are returned, the algorithm returns $\indet$. If policies only return $\notApp$ or $\indet$, then $\indet$, if present, prevails.
\item $\strongCon$: this algorithm is the stronger version of the previous one, in the sense that to obtain $\permit$ (resp., $\deny$) all policies have to return $\permit$ (resp., $\deny$), otherwise $\indet$ is returned. If all policies return $\notApp$ then the result is $\notApp$.
\end{itemize}
The algorithms described in the first four items above are those popularised by XACML. They combine decisions either according to a given precedence criterium or to policy applicability. The remaining two algorithms, instead, are borrowed from \cite{LiWQBRLL09} and compute the combined decision by achieving different forms of consensus.

If the resulting decision is $\permit$ or $\deny$, each algorithm also returns the sequence of instantiated obligations according to the chosen instantiation strategy $\delta$. There are two possible strategies. The $\all$ strategy requires evaluation of all policies in the input sequence and returns the instantiated obligations pertaining to all decisions. Instead, the $\greedy$ strategy prescribes that, as soon as a decision is obtained that cannot change due to evaluation of subsequent policies in the input sequence, the execution halts. Hence, the result will not consider the possibly remaining policies and only contains the obligations already instantiated. Therefore, the instantiation strategies mainly affect the amount of instantiated obligations possibly returned. The $\greedy$ strategy, that reflects the management of obligations in \xacml, may significantly improve the policy evaluation performance. Instead, the $\all$ strategy may require additional workload but, on the other hand, ensures that all the policies and their obligations are always taken into account.

As last step, the calculated \pdp\ response is sent to the \pep\ for the enforcement. To this aim, the \pep\ must discharge all obligations and decide, by means of the chosen enforcement algorithm, how to enforce decisions $\notApp$ and $\indet$. The algorithms are those popularised by \xacml\ and, in particular, the $\denyBiased$ (resp., $\permitBiased$) one enforces $\permit$ (resp., $\deny$) only when all the corresponding obligations are correctly discharged, while enforces $\deny$ (resp., $\permit$) in all other cases. Instead, the $\based$ algorithm leaves all decisions unchanged but, in case of decisions $\permit$ and $\deny$, enforces $\indet$ if an error occurs while discharging obligations. This means that obligations not only affect the authorisation process due to their instantiation, but also the enforcement one. However, errors caused by optional obligations, \ie with type $\obO$, are safely ignored.


\subsection{Policies for the e-Health case study}
\label{sec:facplEx}

We now use the \facpl\ linguistic abstractions to formalise the requirements for the e-Health case study reported in Table~\ref{tab:privacy}. These rules are meant to prevent unauthorised access to patient data and hence to ensure their confidentiality and integrity. The specification of this access control system is introduced bottom-up, from single rules to whole policies, thus illustrating in a step-by-step fashion the combination strategies that could be pursued and their effects.

The system resources to protect via the access control system are \mbox{\emph{e-Prescriptions}}. The access control rules need to deal with requester credentials, i.e. $\x{doctor}$ and $\x{pharmacist}$ roles, along with their assigned permissions, and with $\x{read}$ or $\x{write}$ actions.

Requirement (\sr{1}), allowing doctors to write e-Prescriptions, can be formalised as a \emph{positive} \facpl\ rule (i.e., a rule with effect $\permit$) as follows
$$
\begin{array}{@{}l@{}}
\ruleOpt{\,\permit\ \
\targetBegin
\streq(\x{subject/role},``\x{doctor}")\\
\qquad\qquad\quad
\x{and}\
\streq(\x{action/id},``\x{write}") \\
\qquad\qquad\quad
\x{and}\ \x{in}(``\x{e}\textrm{-}\x{Pre}\textrm{-}\x{Write}",\x{subject/permission})\\
\qquad\qquad\quad
\x{and}\ \x{in}(``\x{e}\textrm{-}\x{Pre}\textrm{-}\x{Read}",\x{subject/permission})}  
\end{array}
$$
The rule target\footnote{To improve code readability, we use the infix operators, 
a textual notation for permissions and an additional check on the subject role. Of course, in a setting with semantically different roles, a stan\-dard\-ised permission-based coding, \eg HL7 (\url{http://www.hl7.org}), should be used for defining role checks.} checks if the requester role is $\x{doctor}$, if the action is $\x{write}$, and if the subject's permissions include those for writing and reading an e-Prescription. 
The control that the resource type is equal to e-Prescription will be performed by the target of the policy enclosing the rule. 
This, because of the hierarchical processing of \facpl\ elements, is enough to ensure that the rule will only be applied to e-Prescriptions. 

Requirement (\sr{2}) can be expressed like the previous: it differs for the action identifier and for the required permissions, i.e. only  $\x{e}\textrm{-}\x{Pre}\textrm{-}\x{Read}$. Requirement (\sr{3}) only differs from the second for the role value.

These three rules, modelling Requirements (\sr{1}), (\sr{2}) and (\sr{3}), can be combined together in a policy set whose target specifies the check on the resource type $\x{e}\textrm{-}\x{Prescription}$ (again, to improve code readability, we use textual encoding for resources).
Since all granted requests are explicitly authorised, choosing the $\permitOverO{all}$ algorithm as combining strategy seems a natural choice. Let thus Policy~(\ref{policy1}) be defined as follows
\[
\begin{array}{@{\hspace*{-.4cm}}l@{\ }}
\policyBegin\, \permitOverO{all} \\
 \ \ \targetBegin
\streq(\x{resource/type},``\x{e}\textrm{-}\x{Prescription}")\\
\ \ \policiesBegin\\
\qquad
\ruleOpt{\,\permit\ \
\targetBegin
\streq(\x{subject/role},``\x{doctor}")\\
\qquad\quad\quad
\x{and}\
\streq(\x{action/id},``\x{write}") \\
\qquad\quad\quad
\x{and}\ \x{in}(``\x{e}\textrm{-}\x{Pre}\textrm{-}\x{Write}",\x{subject/permission})\\
\qquad\quad\quad
\x{and}\ \x{in}(``\x{e}\textrm{-}\x{Pre}\textrm{-}\x{Read}",\x{subject/permission})}\\
\qquad
\ruleOpt{\,\permit\ \
\targetBegin
\streq(\x{subject/role},``\x{doctor}")\\
\qquad\quad\quad
 \x{and}\
\streq(\x{action/id},``\x{read}") \\
\qquad\quad\quad
\x{and}\ \x{in}(``\x{e}\textrm{-}\x{Pre}\textrm{-}\x{Read}",\x{subject/permission})}\\
\qquad
\ruleOpt{\,\permit\ \
\targetBegin
\streq(\x{subject/role},``\x{pharmacist}")\\
\qquad\quad\quad
 \x{and}\
\streq(\x{action/id},``\x{read}") \\
\qquad\quad\quad
\x{and}\ \x{in}(``\x{e}\textrm{-}\x{Pre}\textrm{-}\x{Read}",\x{subject/permission})}\\
\ \ \oblspBegin
\oblBegin\, \obM\ \ \x{log}(\x{system/time},\x{resource/type}, 
\x{subject/id}, \x{action/id})\, \oblEnd\ \policyEnd
\end{array}
\hspace*{-1.2cm}
\tag{P1}
\label{policy1}
\]
Policy~(\ref{policy1}) reports not only access controls but also an obligation formalising Assumption (\sr{4}) about the logging of each authorised access.
The arguments of the obligation action are separated by commas to increase their readability.

Let us now consider a \facpl\ request and evaluate it with respect to Policy~(\ref{policy1}). For the sake of presentation, hereafter we write $A \define t$ to assign the symbolic name $A$ to the term $t$. Let us suppose that doctor $\x{Dr.\ House}$ wants to write an e-Prescription; the corresponding request is defined as follows
$$
\begin{array}{@{ }r@{\,}c@{\, }l}
\x{req1} & \define &
\attribute{\x{subject/id}}{``\x{Dr.\ House}"}
\\
 & & 
\attribute{\x{subject/role}}{``\x{doctor}"}\, \attribute{\x{action/id}}{``\x{write}"}
\\
& & 
\attribute{\x{resource/type}}{``\x{e}\textrm{-}\x{Prescription}"}
\\
& &
\attribute{\x{subject/permission}}{``\x{e}\textrm{-}\x{Pre}\textrm{-}\x{Read}"}\\
& &
\attribute{\x{subject/permission}}{``\x{e}\textrm{-}\x{Pre}\textrm{-}\x{Write}"}
\ \ \ldots \\
\end{array}
$$
where attributes are organised into the categories \emph{subject}, \emph{resource} and \emph{action}. Additional attributes possibly included in the request are omitted because they are not relevant for this evaluation. Notice that $\x{subject/permission}$ is a multivalued attribute and it is properly handled in the previous rules by using the $\x{in}$ operator, which verifies the membership of its first argument to the set that forms its second argument.

The authorisation process of $\x{req1}$ returns a $\permit$ decision. In fact, the request matches the policy target, as the resource type is $\x{e}\textrm{-}\x{Prescription}$, and exposes all the permissions required in the first rule for the $\x{write}$ action and the $\x{doctor}$ role. The response, that is a $\permit$ including a $\x{log}$ obligation, is defined, e.g., as follows
$$
\begin{array}{l}
\langle \ \permit\, \ [ \, \obM \ \ \x{log}(\x{2016}\textrm{-}\x{10}\textrm{-}\x{22}\, \x{10}\textrm{:}\x{15}\textrm{:}\x{12},
``\x{e}\textrm{-}\x{Prescription}",``\x{Dr.\ House}",``\x{write}") \, ] \rangle
\end{array}
$$
The instantiated obligation indicates that the \pdp\ succeeded in retrieving and evaluating all the attributes occurring within the arguments of the action; run-time information, such as the current time, is retrieved through the context handler.

The evaluation of $\x{req1}$ returns the expected result. We might be led to believe that due to the simplicity of Policy~(\ref{policy1}), this is true for all requests. However, this correctness property cannot be taken for granted as, in general, even though the meaning of a rule is straightforward, this may not be the case for a combination of rules. Depending on the chosen combination strategy, some unexpected results can arise. For example, a request by a $\x{pharmacist}$ for a $\x{write}$ action on an $\x{e}\textrm{-}\x{Prescription}$ is not explicitly allowed by the requirements in Table~\ref{tab:privacy}; hence, it should be forbidden. However, the corresponding request
$$
\hspace*{-.2cm}
\begin{array}{r@{\, }c@{\, }l}
\x{req2} & \define &
\attribute{\x{subject/id}}{``\x{Dr.\ Wilson}"}
\\
 & & 
\attribute{\x{subject/role}}{``\x{pharmacist}"}
\,
 \attribute{\x{action/id}}{``\x{write}"}
\\
&& 
\attribute{\x{resource/type}}{``\x{e}\textrm{-}\x{Prescription}"}
\\
& &
\attribute{\x{subject/permission}}{``\x{e}\textrm{-}\x{Pre}\textrm{-}\x{Read}"}
\ \ \ldots \\
\end{array}
$$
would evaluate to $\notApp$. In fact, all enclosed rules do not apply (i.e., their targets do not match) and the resulting $\notApp$ decisions are combined by the $\permitOverO{all}$ algorithm to $\notApp$ as well. Therefore, the enforcement algorithm of the \pep\ is entrusted with the task of taking the final decision for request $\x{req2}$. Even though this is correct in a setting where the \pep\ is well-defined, \eg the epSOS system, it is not recommended when design assumptions on the \pep\ implementation are missing. In fact, a biased algorithm might transform $\notApp$ into $\permit$, possibly causing unauthorised accesses.

To prevent $\notApp$ decisions to be returned by the policy, we can replace the combining algorithm of Policy~(\ref{policy1}) with the $\denyUnlessO{all}$ one. This implies that $\deny$ is taken as the default decision and is returned whenever no rule returns $\permit$. Alternatively, we can get the same achievement by using a policy set defined as the combination, through the $\permitOverO{all}$ algorithm, of Policy~(\ref{policy1}) and a rule forbidding all accesses.  This rule is simply defined as $\ruleOpt{\deny}$: the absence of the target and the \emph{negative} effect means that it always returns $\deny$. Now, let Policy~(\ref{policyPatientConsent}) be defined as
\[
\begin{array}{@{\hspace*{-.3cm}}l@{\hspace*{-.1cm}}}
\policyBegin\, \permitOverO{all} \\
\quad
\policiesBegin
\ \policyBegin\ldots\mathit{Policy~(\ref{policy1})}\ldots\policyEnd
\ \
\ruleOpt{\deny}\\
\quad
\oblspBegin\ 
\oblBegin\, \obO \ \  \x{compress}(\,) \, \oblEnd\\
\quad
\oblsdBegin\ 
\oblBegin\, \obM \ \  \x{mailTo}(\x{resource/patient}\textrm{-}\x{mail} ,
\\ \qquad\qquad\quad
``\x{Data\ request\ by\ unauthorised\ subject}") \, 
\oblEnd
\,\policyEnd
\end{array}
\hspace*{-.7cm}
\tag{P2}
\label{policyPatientConsent}
\]
Policy~(\ref{policyPatientConsent}) reports two obligations formalising, respectively, the last two requirements of Table~\ref{tab:privacy}: (i) a patient is informed about unauthorised attempts to access her data by means of an obligation for the effect $\deny$ and (ii) if possible, data are exchanged in compressed form by means of an obligation for the effect $\permit$. Notably, the type `optional' is exploited so that compressed exchanges are not strictly required but, e.g., only whenever the corresponding service is available.

Policy~(\ref{policyPatientConsent}) can be used as a basis for the definition of the \emph{patient informed consent} (see Section~\ref{sec:casestudy}). For instance, $\alice$'s policy for the management of her health data could be simply obtained by adding a check on the patient identifier to which the policy applies, such as 
$
\begin{array}{@{}l@{}}
\targetBegin \streq(``\alice",\x{resource/patient}\textrm{-}\x{id})
\end{array}
$,
to Policy~(\ref{policyPatientConsent}). In this way, $\alice$ grants access to her e-Prescription data to the healthcare professionals that satisfy the requirements expressed in her consent policy. Another patient expressing a more restrictive consent, where \eg writing of e-Prescriptions is disabled, will have a similar policy set where the rule modelling Requirement~(\sr{1}) is not included. In a more general perspective, the \pdp\ could have a policy set for each patient, that encloses the policies expressing the consent explicitly signed by the patient. This is the approach followed, e.g., in the Austrian e-Health platform (\url{http://www.elga.gv.at/}).

As shown before, it could be challenging to identify unexpected authorisations and to determine whether policy fixes affect authorisations that should not be altered. The combination of a large number of complex policies is indeed an error-prone task that has to be supported with effective analysis techniques. Therefore we equip \facpl\ with a formal semantics and then define a constraint-based analysis providing effective supporting techniques for the verification of properties on policies.

\section{FACPL Formal Semantics}
\label{sec:formal_sem}

\begin{table*}[tb]
\caption{Correspondence between syntactic and semantic domains}
\label{table:syntsemcorr}
\centering
\begin{tabular}{|@{\hspace{.1cm}}c@{\hspace{.1cm}}|@{\hspace{.1cm}}c@{\hspace{.1cm}}|@{\hspace{.1cm}}c@{\hspace{.1cm}}|@{\hspace{.1cm}}c@{\hspace{.1cm}}|@{\hspace{.1cm}}c@{\hspace{.1cm}}|}
\hline
\textbf{Syntactic} & \textbf{Generic} & \textbf{Semantic} & \textbf{Syntactic} & \textbf{Semantic}  \\ 
\textbf{category} & \textbf{synt. elem}. & \textbf{function} &  \textbf{domain} & \textbf{domain} \\ \hline\hline
Attribute names & $\name$ & & $\mathit{Name}$ & \\
Literal values & $\val$& & $\mathit{Value}$ &  \\
Requests & $\rSyntax$ & $\denSemF{R}$ & $\mathit{Request}$ & 
$R \define \mathit{Name} \rightarrow (\mathit{Value} \cup 2^{\mathit{Value}}\cup \{ \excpt \})$  \\ \hline
Expressions & $\expr $ & $\denSemF{E}$ & $\mathit{Expr}$ & $R \rightarrow \mathit{Value} \cup 2^{\mathit{Value}}\cup \{\err, \excpt \}$ \\ \hline
Effects & $\effect$ & & 
$\mathit{Effect}$ & \\  
Obligation Types & $\obType$ & & 
$\mathit{ObType}$ & \\  
Pep Actions & $\pepAction$ & & $\mathit{PepAction}$ & \\
Instantiated obligations & $\fo$ & & $\mathit{IObligation}$ & \\ 
Obligations & $\ob$ & $\denSemF{O}$ & $\mathit{Obligation}$ &  
{$R \rightarrow \ \mathit{IObligation} \cup \{\mathtt{\err}\}$}\\
\hline
PDP Responses & $\pdpRes$ & & $\mathit{PDPReponse}$ & \\ 
Policies & $\policy$ & $\denSemF{P}$ & $\mathit{Policy}$ & $R \rightarrow \mathit{PDPReponse}$ \\
\hline
Policy Decision Points & $\pdpSyntax$ & $\denSemF{P}dp$ & $\mathit{PDP}$ & $R \rightarrow \mathit{PDPReponse}$ \\ \hline
Combining algorithms & $\algSyntax$ & $\denSemF{A}$ & $\mathit{Alg}\times \mathit{Policy}^+$ &  $R \rightarrow \mathit{PDPReponse}$  \\ \hline
Decisions & $\dec$& & $\mathit{Decision}$ & \\
Enforcement algorithms & $\enfAlg$ & $\denSemF{E}A$ & $\mathit{EnfAlg}$ & $\mathit{PDPReponse} \rightarrow \mathit{Decision}$ \\  \hline
Policy Auth. System & $pas$ & $\denSemF{P}as$ & $\mathit{PAS}$ & $\mathit{Request} \rightarrow \mathit{Decision}$ \\ \hline
\end{tabular}
\end{table*}

In this section, we present the formal semantics of \facpl\ by formalising the evaluation process introduced in Section~\ref{sec:model} and detailed in Section~\ref{sec:informal_sem}. The semantics is defined by following a denotational approach which means that 
\begin{itemize}
\item we introduce some semantic functions mapping each \facpl\ syntactic construct to an appropriate \emph{denotation}, that is an element of a semantic domain representing the meaning of the construct;
\item the semantic functions are defined in a \emph{compositional} way, so that the semantic of each construct is formulated in terms of the semantics of its sub-constructs.
\end{itemize}
To this purpose, we specify a family of semantic functions mapping each syntactic domain to a specific semantic domain. These functions are inductively defined on the \facpl\ syntax through appropriate semantic clauses following a `point-wise' style. For instance, on the syntactic domain \emph{Policy} representing all \facpl\ policies, we formalise the function $\denSemF{P}$ that defines a semantic domain mapping \facpl\ requests to \pdp\ responses.

In the sequel, we convene that the application of the semantic functions is left-as\-so\-cia\-tive, omits parenthesis whenever possible, and surrounds syntactic objects with the emphatic brackets $[\![$~and~$]\!]$ to increase readability. For instance, $\exprSem{\name}{\req}$ stands for $( \denSemF{E}(\name) )(\req)$ and indicates the application of the semantic function $\denSemF{E}$ to (the syntactic object) $\mathit{n}$ and (the semantic object) $\req$. We also assume that each nonterminal symbol in Tables~\ref{tab:facpl_syntax} and~\ref{tab:facpl_context_syntax} (defining the \facpl\ syntax) denominates the set of constructs of the syntactic category defined by the corresponding EBNF rule, e.g. the nonterminal $\mathit{Policy}$ identifies the set of all \facpl\ policies. The used notations are summarised in Table~\ref{table:syntsemcorr} (the missing semantic domains coincide with the corresponding syntactic ones).

In the rest of this section, we detail the semantics of requests (Section~\ref{sez:formalRequest}), \pdp\ (Sections~\ref{sez:formalPDP} and~\ref{sec:sem_alg}), \pep\ (Section~\ref{sec:sem_enfAlg}), Policy Authorisation System (Section~\ref{sec:sem_pas}) and present some properties of the semantics (Section~\ref{sec:prop_sem}).

\subsection{Semantics of Requests}
\label{sez:formalRequest}

The meaning of a request\footnote{For simplicity sake, here we assume that, when the evaluation of a request takes place, the original request has been already enriched with the information that would be retrieved at run-time from the environment by the context handler (steps 5-8 in Figure~\ref{fig:facplModel}).} is an element of the set $R \define \mathit{Name} \rightarrow (\mathit{Value} \cup 2^{\mathit{Value}}\cup \{ \excpt \})$, that is a total function that maps attribute names to either a literal value, or a set of values (in case of multivalued attributes), or the special value $\excpt$ (if the value for an attribute name is missing). The mapping from a request to its meaning is given by the semantic function $\denSemF{R}: \mathit{Request} \rightarrow R$, defined as follows:
\[
\begin{array}{l}
\reqSemS{\attribute{\name'}{\val'}}{\name} = 
\left\{
    \begin{array}{l@{\ }l}
    	\val' 
		& \mathtt{if}\ \name = \name'\\[.1cm] 
	\excpt
		& \mathtt{otherwise}\\
    \end{array}
\right.
\\
\mbox{}\\[-.05cm]
\reqSemS{{\attribute{\name_i}{\val_i}}^{+} {\attribute{\name'}{\val'}}}{\name} =
\left\{
    \begin{array}{l@{\ }l}
    	\reqSemS{{\attribute{\name_i}{\val_i}}^{+}}{\name} \Cup  \val'  
		& \mathtt{if}\ \name = \name'\\[.1cm] 
	\reqSemS{{\attribute{\name_i}{\val_i}}^{+}}{\name}
	 	& \mathtt{otherwise}\\
    \end{array}
\right.
\end{array}
\hspace*{-1.5cm}
\tag{\semLbl{1}}
\label{sem:req}
\]
The semantics of a request, which is a function $\req \in R$, is thus inductively defined on the length of the request. To deal with multivalued attributes we introduce the operator $\Cup$, 
which is straightforwardly defined by case analysis on the first argument as follows 
$$
v \Cup v' = \{v,v'\}
\qquad
V \Cup v' = V \cup \{v'\}
\qquad
\excpt \Cup\, v' = v' \ 
$$
where we let $V \in 2^{\mathit{Value}}$.


\subsection{Semantics of the Policy Decision Process}
\label{sez:formalPDP}

We start defining the semantics of expressions and obligations that will be then exploited for defining the semantics of policies.

In Table~\ref{tab:sem_expression} we report (an excerpt of) the clauses defining the function $\denSemF{E} : \mathit{Expr} \rightarrow (R \rightarrow \mathit{Value}\ \cup\ 2^{\mathit{Value}}\cup \{\err, \excpt \})$ modelling the semantics of expressions. This means that the semantics of an expression is a function of the form $R \rightarrow \mathit{Value} \cup 2^{\mathit{Value}}\cup \{\err, \excpt \}$ that, given a request, returns a literal value, or a set of values, or the special value $\excpt$, or an error (\eg when an argument of an operator has unexpected type). The evaluation order of sub-expressions is not relevant, as they do not generate side-effects.

\begin{table*}[!h]
\caption{Semantics of (an excerpt of) \facpl\ Expressions ($T$ stands for one of the sets of literal values or for the powerset of the set of all literal values, and $i, j \in \{1, 2\}$ with $i \neq j$)} \label{tab:sem_expression}
\vspace*{-.2cm}
\centering
$$
\footnotesize
\begin{array}{l@{\!\! }l@{}}
\hline
\mbox{}\\[-.1cm]
\exprSem{\name}{\req} = \req(\name) & \exprSem{\val}{\req} = \val 
\\
\mbox{}\\[-.1cm]
\begin{array}{@{ }l}
\exprSem{\x{or}(\expr_1,\expr_2)}{\req} = \\
\ \left\{
   \begin{array}{l@{\ \ }l}
      	\true  & \mathtt{if}\  \exprSem{\expr_1}{\req} = \true \vee \exprSem{\expr_2}{\req} = \true\! \\
	\false  & \mathtt{if}\  \exprSem{\expr_1}{\req} = \exprSem{\expr_2}{\req} = \false \\
	\excpt & \mathtt{if}\   \exprSem{\expr_i}{\req} = \excpt 
	\wedge\ \exprSem{\expr_j}{\req} \in \{\false,\excpt\} \\
	\err & \mathtt{otherwise} \\
   \end{array}
    \right.\\
\end{array} 
&
\begin{array}{@{ }l}
\exprSem{\x{and}(\expr_1,\expr_2)}{\req} = \\
\ \left\{
   \begin{array}{l@{\ \ }l}
      	\true  & \mathtt{if}\   \exprSem{\expr_1}{\req} = \exprSem{\expr_2}{\req} = \true \\
	\false  & \mathtt{if}\   \exprSem{\expr_1}{\req} = \false \vee \exprSem{\expr_2}{\req} = \false \\
	\excpt & \mathtt{if}\   \exprSem{\expr_i}{\req} = \excpt 
	\wedge\ \exprSem{\expr_j}{\req} \in \{\true,\excpt\} \\
	\err & \mathtt{otherwise} \\
   \end{array}
    \right.\\
\end{array}
\\
\mbox{}\\[-.1cm]
\begin{array}{@{ }l}
\exprSem{\x{not}(\expr)}{\req} = \\
\ \left\{
   \begin{array}{l@{\ \ }l}
      	\true  & \mathtt{if}\   \exprSem{\expr}{\req} = \false \\
	\false  &\mathtt{if}\   \exprSem{\expr}{\req} = \true \\
	\excpt &\mathtt{if}\   \exprSem{\expr}{\req} = \excpt \\ 
	\err &  \mathtt{otherwise} \\[.1cm]
   \end{array}
    \right.
\end{array}
&
\begin{array}{@{ }l}
\exprSem{\x{add}(\expr_1,\expr_2)}{\req} =  \\
\ \left\{
   \begin{array}{l@{\ \ }l}
      	(\exprSem{\expr_1}{\req} + \exprSem{\expr_2}{\req}) & \mathtt{if}\  \exprSem{\expr_1}{\req}, \exprSem{\expr_2}{\req} \in \doubleT\\
	\excpt & \mathtt{if}\   \exprSem{\expr_i}{\req} = \excpt\ \wedge\ \exprSem{\expr_j}{\req} \neq \err \hspace*{-.4cm} \\
	\err & \mathtt{otherwise} \\
   \end{array}
    \right.  
 \end{array}
\\
\mbox{}\\[-.1cm]
\multicolumn{2}{l}{
 \begin{array}{@{ }l}
\exprSem{\x{in}(\expr_1,\expr_2)}{\req} = \\
\ \left\{
   \begin{array}{l@{\ \ }l}
      	(\exprSem{\expr_1}{\req} \in \exprSem{\expr_2}{\req}) & \mathtt{if}\  \exprSem{\expr_1}{\req} \in \type \ \wedge\ \exprSem{\expr_2}{\req} \in 2^\type \\
	\excpt & \mathtt{if}\ \exprSem{\expr_i}{\req} =\ \excpt	\ \wedge\ \exprSem{\expr_j}{\req} \neq \err \\
	\err & \mathtt{otherwise} \\
   \end{array}
    \right.  
 \end{array}
 }
 \\
\mbox{}\\[-.1cm]
\multicolumn{2}{l}{
 \begin{array}{@{ }l}
\exprSem{\x{equal}(\expr_1,\expr_2)}{\req} = \\
\ \left\{
   \begin{array}{l@{\ \ }l}
      	(\exprSem{\expr_1}{\req} = \exprSem{\expr_2}{\req})\! & \mathtt{if}\  \exprSem{\expr_1}{\req}, \exprSem{\expr_2}{\req} \in\type \hspace*{-.6cm}\\
	\excpt & \mathtt{if}\ \exprSem{\expr_i}{\req} =\ \excpt \\
	& \quad\ \wedge\ \exprSem{\expr_j}{\req} \neq \err\!\\
	\err & \mathtt{otherwise} \\
   \end{array}
    \right.  
 \end{array}
 }
\\
\mbox{}\\[-.1cm]
\hline
\end{array}
$$
\end{table*}

The first raw of the table contains the clauses for basic expressions, \ie attribute names and literal values. The semantics of the expression formed by a name $\name$ is a function that, given a (semantic) request $\req$ in input, returns the value that $\req$ associates to $\name$. 
This is written as the clause $\exprSem{\name}{\req} = \req(\name)$. Similarly, the case of a value $v$ is a function that always returns the value itself, that is the clause $\exprSem{\val}{\req} = v$.

The remaining clauses, one for each operator, 
present (an excerpt of) the semantics of expression operators. In particular, each clause uses straightforward semantic operators for composing denotations (\eg $=$ corresponds to $\x{equal}$), and implements the management strategy for the special values $\excpt$ and $\err$. The clauses establish that $\err$ takes precedence over $\excpt$ and is returned every time the operator arguments have unexpected types; whereas $\excpt$ is returned when at least an argument is $\excpt$ and there is no $\err$. The clauses of operators $\x{and}$ and $\x{or}$ possibly mask these special values by implementing the behaviour informally described in Section~\ref{sec:informal_sem}. It is worth noticing that the explicit management of missing attributes and evaluation errors ensures a full account of crucial aspects of access control policy evaluation, usually neglected by other proposals from the literature (see, e.g.,~\cite{Jajodia97alogical,RamliNN12,ArkoudasCC14}). The only proposals considering the role of missing attributes are those in~\cite{CramptonM12,CramptonW15}, but they only consider a simplified policy language and assume that expressions cannot generate errors.

Function $\denSemF{E}$ is straightforwardly extended to sequences of expressions by the following clauses
\[
\begin{array}{c}
\exprSem{\epsilon}{\req}  =  \epsilon
\\[.2cm]
\exprSem{\expr'\ \expr^*}{\req} = \exprSem{\expr'}{\req} \ \concat \ \exprSem{\expr^*}{\req}
\end{array}
\tag{\semLbl{2}}
\label{sem:exp_con}
\]
The operator $\concat$ denotes concatenation of sequences of semantic elements and $\epsilon$ denotes the empty sequence. We assume that $\concat$ is strict on $\err$ and $\excpt$, \ie $\err$ is returned whenever an $\err$ or $\excpt$ is in the sequence. Therefore, the evaluation of $\exprSem{\expr^*}{\req}$ fails if any of the expressions forming $\expr^*$ evaluates to $\err$ or $\excpt$.

The semantics of the instantiation of obligations is formalised by the function $\denSemF{O}: \ \mathit{Obligation} \rightarrow\ ( R \rightarrow \ \mathit{IOblgation} \cup \{\mathtt{\err}\} )$ defined by the clause
\[
\begin{array}{l}
\oblSem{[ \, \obType \ \ \pepAction(\expr^*) \, ]}{\req} = \\[.2cm]
\ \
\left\{
\begin{array}{ll}
[ \, \obType \ \  \pepAction(\extVal^*) \ ]  & \ \mathtt{if}\ \  \exprSem{\expr^*}{\req} = \extVal^*\\[.1cm]
\err & \ \mathtt{otherwise}
\end{array}
\right.
\end{array}
\hspace*{-.5cm}
\tag{\semLbl{3a}}
\label{sem:obl}
\]
where $\extVal$ stands for a literal value or a set of literal values. Thus, given a request, the instantiation of an obligation returns an instantiated obligation, if the evaluation of every expression argument of the action returns a value. Otherwise, it returns an error.

Function $\denSemF{O}$ is straightforwardly extended to sequences of obligations as follows
\[
\hspace*{-.3cm}\oblSem{\epsilon}{\req}  =  \epsilon
\qquad
\oblSem{\ob' \ob^*}{\req} = \oblSem{\ob'}{\req} \ \concat \ \oblSem{\ob^*}{\req}
\tag{\semLbl{3b}}
\label{sem:obl2}
\]
Notably, a sequence of instantiated obligations is returned only if every obligation in the sequence is successfully instantiated; otherwise, $\err$ is returned (indeed, $\concat$ is strict on $\err$).

We can now define the semantics of a policy as a function that, given a request, returns an authorisation decision paired with a (possibly empty) sequence of instantiated obligations. Formally, it is given by the function $\denSemF{P}: \mathit{Policy} \rightarrow (R \rightarrow \mathit{PDPReponse})$ that has two defining clauses: one for rules and one for policy sets. The clause for rules is
\[
\!\! \begin{array}{l}
    \policySem{ \ruleOpt{\effect\ \ \x{target:} \, \expr\ \ \x{obl:} \, \ob^{*} \,} }{\req} = \\[.2cm]
    \ \ 
    \left\{
    \begin{array}{l@{\ \ }l}
    \langle \effect\ \ \fo^* \rangle & \mathtt{if} \ \exprSem{\expr}{\req}=\true\\ 
    & \qquad\qquad
    \wedge\ \oblSem{\ob^{*}}{\req} = \fo^* \\[.1cm]
    \notApp & \mathtt{if} \ \exprSem{\expr}{\req}=\false \\
    & \qquad\qquad
     \vee\ \exprSem{\expr}{\req}=\ \excpt\\[.1cm]
    \indet &  \mathtt{otherwise}\\
    \end{array}
    \right.
\end{array}
\hspace*{-1.4cm}
\tag{\semLbl{4a}}
\label{sem:rule}
\]
Thus, the rule effect is returned as a decision when the target evaluates to $\true$, which means that the rule applies to the request, and all obligations are successfully instantiated. In this case, the instantiated obligations are also part of the response. Otherwise, it could be the case that \emph{(i)} the rule does not apply to the request, \ie the target evaluates to $\false$ or to $\excpt$, or that \emph{(ii)} an error has occurred while evaluating the target or instantiating the obligations. 

The semantics of policy sets relies on the semantics of combining algorithms. Indeed, as detailed in Section~\ref{sec:sem_alg}, we use a semantic function $\denSemF{A}$ to map each combining algorithm $\algSyntax$ to a function that, to a sequence of policies, associates a function from requests to \pdp\ responses. The clause for policy sets is
\[
\hspace*{-.4cm}
\begin{array}{@{}l}
\policySem{\{ \algSyntax\ \x{target:} \, \expr\ \x{policies:} \, \policy^{+}\, {\x{\oblp:} \, \ob_{p}^{*}\ \ \x{\obld:} \, \ob_{d}^{*}} \, \}}{\req}\\
=\\
\ 
\left\{
\begin{array}{@{\!\!}l@{\, }l@{}}
{\langle \permit\ \ \foS_1\ \concat\ \foS_2 \rangle} &  
\mathtt{if} \ \exprSem{\expr}{\req}=\true \\ 
& \wedge\, \algSem{\algSyntax, \policy^+}{\req}= \langle \permit\  \foS_1 \rangle \\ 
&  \wedge\ \oblSem{{\ob_{p}^{*}}}{\req} = \foS_2 \\[.1cm]
{\langle \deny\ \ \foS_1\ \concat\ \foS_2 \rangle} &  
\mathtt{if} \ \exprSem{\expr}{\req}=\true \\ 
& \wedge\, \algSem{\algSyntax, \policy^+}{\req}= \langle \deny\  \foS_1 \rangle \\
& \wedge\, \oblSem{{\ob_{d}^{*}}}{\req} = \foS_2 \\[.1cm]
\notApp 
& \mathtt{if} \ \exprSem{\expr}{\req}=\false \\
& \vee\ \exprSem{\expr}{\req}=\ \excpt \\
& \vee\ (\exprSem{\expr}{\req}=\true \\
& \quad \ 
   \wedge \ \algSem{\algSyntax, \policy^{+}}{\req}=\notApp) \\[.1cm]
\indet &  \mathtt{otherwise}\\
\end{array}
\right.
\end{array}
\hspace*{-2cm}
\tag{\semLbl{4b}}
\label{sem:pol}
\]
Thus, the policy set applies to the request when the target evaluates to $\true$, the semantic of the combining algorithm $\algSyntax$ (which is applied to the enclosed sequence of policies and the request) returns the effect $\effect$ and a sequence of instantiated obligations $\foS_1$, and all the enclosed obligations for the effect $\effect$ are successfully instantiated and return a sequence $\foS_2$. In this case, the \pdp\ response contains $\effect$ and the concatenation of the sequences $\foS_1$ and $\foS_2$. Instead, if the target evaluates to $\false$ or to $\excpt$, or the combining algorithm returns $\notApp$, the policy set does not apply to the request. The response is $\indet$ in the remaining cases, \ie when an error occurred in the evaluation of the target or of the obligations, or when the evaluation of the combining algorithm returned $\indet$.

Finally, the semantic of a \pdp\ is that function from requests to \pdp\ responses obtained by applying the combining algorithm to the enclosed sequence of policies, \ie
\[
\pdpSem{\pdpPol{\algSyntax }{\x{policies:} \, \policy^+}}{\req} = \algSem{\algSyntax, \policy^+}{\req}
\tag{\semLbl{5}}
\label{sem:pdp}
\]

\subsection{Semantics of Combining Algorithms}
\label{sec:sem_alg}

\begin{table*}[!t]
\caption{Auxiliary definitions for the semantics of combining algorithms: (a) combination matrix for the $\algOpAlg{\permitOverO{}}$ operator ($\pdpRes_1$ and $\pdpRes_2$ indicate the first and the second argument, respectively); (b) definition of the $\isFinal{\x{alg}}{\pdpRes}$ predicate}
\label{tab:auxAlg}
\vspace*{-.3cm}
$$
\footnotesize
\begin{array}{@{\!\!}l@{\ \ }c@{}}
(a) & 
\begin{array}{l||c|c|c|c|}
\ \ \lfrac{\pdpRes_1}{\pdpRes_2}
	    & \langle \permit\ \ \foS_2 \rangle & \langle \deny\ \ \foS_2 \rangle & \notApp & \indet \\[.08cm]
            \hline\hline
            \langle \permit\ \ \foS_1 \rangle  & \langle \permit\ \ \foS_1\ \concat\ \foS_2 \rangle & \langle \permit\ \ \foS_1 \rangle & \langle \permit\ \ \foS_1 \rangle & \langle \permit\ \ \foS_1 \rangle \\
            \langle \deny\ \ \foS_1 \rangle & \langle \permit\ \ \foS_2 \rangle & \langle \deny\ \ \foS_1\ \concat\ \foS_2 \rangle & \langle \deny\ \ \foS_1 \rangle & \indet \\
            \notApp & \langle \permit\ \ \foS_2 \rangle & \langle \deny\ \ \foS_2 \rangle & \notApp & \indet \\
            \indet & \langle \permit\ \ \foS_2 \rangle & \indet & \indet & \indet \\
            \hline
\end{array}
\\
\mbox{}\\
(b) &
\begin{array}{l@{\ \ \ }ll@{}}
\hline 
\mbox{}\\[-.2cm]
\begin{array}{l}
\isFinal{\permitOverO{}}{\pdpRes} = \\[.05cm]
\qquad
    \left\{
        \begin{array}{ll}
        \true & \mathtt{if}\ \pdpRes.\dec = \permit\\
        \false & \mathtt{otherwise}
        \end{array}
    \right.
  \end{array}
&
\begin{array}{l}
\isFinal{\denyOverO{}}{\pdpRes} = \\[.05cm]
\qquad
    \left\{
        \begin{array}{ll}
        \true & \mathtt{if}\ \pdpRes.\dec = \deny\\
        \false & \mathtt{otherwise}
        \end{array}
    \right.
    \end{array}
&
\begin{array}{l}
\isFinal{\denyUnlessO{}}{\pdpRes} = \\[.05cm]
\qquad
    \left\{
        \begin{array}{ll}
        \true & \mathtt{if}\ \pdpRes.\dec = \permit\\
        \false & \mathtt{otherwise}
        \end{array}
    \right.
    \end{array}
\\[.5cm]
\begin{array}{l}
\isFinal{\permitUnlessO{}}{\pdpRes} = \\[.05cm]
\qquad
    \left\{
        \begin{array}{ll}
        \true & \mathtt{if}\ \pdpRes.\dec = \deny\\
        \false & \mathtt{otherwise}
        \end{array}
    \right.
    \end{array}
&
\begin{array}{l}
\isFinal{\firstAppO{}}{\pdpRes} = \\[.05cm]
\qquad
    \left\{
        \begin{array}{ll}
        \false & \mathtt{if}\ \pdpRes.\dec = \notApp\\
        \true & \mathtt{otherwsise}
        \end{array}
    \right.
    \end{array}
&
\begin{array}{l}
\isFinal{\onlyOneAppO{}}{\pdpRes} = \\[.05cm]
\qquad
    \left\{
        \begin{array}{ll}
        \true & \mathtt{if}\ \pdpRes.\dec = \indet\\
        \false & \mathtt{otherwsise}
        \end{array}
    \right.    
   \end{array}
\\[.2cm]
\begin{array}{l}
\isFinal{\weakConO{}}{\pdpRes} = \\[.05cm]
\qquad
 \left\{
        \begin{array}{ll}
        \true & \mathtt{if}\ \pdpRes.\dec = \indet\\
        \false & \mathtt{otherwsise}
        \end{array}
  \right.
\end{array}
&
\begin{array}{l}
\isFinal{\strongConO{}}{\pdpRes} = \\[.05cm]
\qquad
    \left\{
        \begin{array}{ll}
        \true & \mathtt{if}\ \pdpRes.\dec = \indet\\
        \false & \mathtt{otherwsise}
        \end{array}
    \right.
    \end{array}
\\
\mbox{}\\[-.2cm]
\hline
\end{array}
\end{array}
$$
\end{table*}

The semantics of combining algorithms is defined in terms of a family of binary operators. Let $\algName$ denote the name of a combining algorithm (i.e., $\permitOverO{}$, $\denyOverO{}$, etc.); the corresponding semantic operator is identified as $\algOp$ and is defined by means of a two-dimensional matrix that, given two PDP responses, calculates the resulting combined response. For instance, Table~\ref{tab:auxAlg}(a) reports the combination matrix for the $\algOpAlg{\permitOverO{}}$ operator. Basically, the matrix specifies the precedences among the $\permit$, $\deny$, $\notApp$ and $\indet$ decisions, and shows how the resulting (sequence of) instantiated obligations is obtained, i.e. by concatenating the instantiated obligations of the responses whose decision matches the combined one. All other combining algorithms described in Section~\ref{sec:informal_sem}, and possibly many others, can be defined in the same manner (see Appendix~\ref{sec:appendixA}). 

The semantics of the combining algorithms can be now formalised by the function $\denSemF{A}: \mathit{Alg}\times \mathit{Policy}^+ \rightarrow (R \rightarrow\ \mathit{PDPReponse})$. This function is defined in terms of the iterative application of the binary combining operators by means of two definition clauses according to the adopted instantiation strategy: the $\all$ strategy always requires evaluation of all policies, while the $\greedy$ strategy halts the evaluation as soon as a final decision is determined (\ie without necessarily taking into account all policies in the sequence). If the $\all$ strategy is adopted, the definition clause is as follows
\[
\hspace*{-.4cm}
\begin{array}{l}
\algSem{\alg{\all}, \policy_1\  \ldots\ \policy_s}{\req} = \\[.2cm]
\ \algOp (\algOp(\ldots \algOp (\policySem{\policy_1}{\req}, \policySem{\policy_2}{\req}),\ldots), \policySem{\policy_s}{\req})
\end{array}
\hspace*{-.9cm}
\raisetag{.2cm}
\tag{\semLbl{6a}}
\label{sem:algA}
\]
meaning that the combining operator is sequentially applied to the denotations of all input policies\footnote{In case of a single policy, operators $\algOpAlg{\permitUnlessO{}}$ and $\algOpAlg{\denyUnlessO{}}$ turn the $\notApp$ and $\indet$ responses into, respectively, $\langle \permit \ \ \epsilon \rangle$ and $\langle \deny\ \ \epsilon \rangle$, while the remaining operators leave them unchanged.}. Instead, if the $\greedy$ strategy is used, the definition clause is as follows
\[
\hspace*{-.4cm}
\begin{array}{@{}l@{}}
\algSem{\alg{\greedy}, \policy_1\ \ldots\ \policy_s}{\req} = \\[.1cm]
\left\{
\begin{array}{@{\:}l@{\,  }r@{\ }l@{}}
\pdpRes_1 & \mathtt{if}& \policySem{\policy_1}{\req} = \pdpRes_1 \, \wedge\, \isFinal{\mathsf{alg}}{\pdpRes_1} \\[.1cm]
\pdpRes_2 & \mathtt{elseif}& \algOp(\pdpRes_1,\policySem{\policy_2}{\req}) = \pdpRes_2 \\
& & \quad
 \wedge\, \isFinal{\mathsf{alg}}{\pdpRes_2} \\[-.2cm]
\vdots & \vdots \\
\pdpRes_{s\textrm{-}1} & \mathtt{elseif}& \algOp(\pdpRes_{s\textrm{-}2},\policySem{\policy_{s\textrm{-}1}}{\req}) = \pdpRes_{s\textrm{-}1} \\ 
& & \quad \wedge\, \isFinal{\mathsf{alg}}{\pdpRes_{s\textrm{-}1}} \\[.1cm]
\pdpRes_{s}
 & \mathtt{otherwise}& (\mbox{where\ }  \pdpRes_s = \algOp(\pdpRes_{s\textrm{-}1},\policySem{\policy_s}{\req})\\
\end{array}
\right.
\end{array}
\hspace*{-1.5cm}
\tag{\semLbl{6b}}
\label{sem:algB}
\]
where 
the $\mathtt{elseif}$ notation is a shortcut to represent mutually exclusive conditions. The auxiliary predicates $\isFinalPred{\x{alg}}$ (one for each combining algorithm $\algName$), given a response in input, check if the response decision is final with respect to the algorithm $\x{alg}$, \ie if such decision cannot change due to further combinations. Their definition is in Table~\ref{tab:auxAlg}(b); as a matter of notation, we use $\pdpRes.\dec$ to indicate the decision of response $\pdpRes$. These predicates are straightforwardly derived from the combination matrices of the binary operators, thus we only comment on salient points. In case of the $\permitOverO{}$ algorithm (and similarly for the others in the first two rows of the table), the $\permit$ decision is the only decision that can never be overwritten, hence, it is final. In case of the $\firstAppO{}$ algorithm, instead, all decisions except $\notApp$ are final since they represent the fact that the first applicable policy has been already found. Both consensus algorithms have $\indet$ as final decision, because no form of consensus can be reached once an $\indet$ is obtained. Similarly, the $\onlyOneAppO{}$ algorithm has $\indet$ as final decision.

\subsection{Semantics of the Policy Enforcement Process}
\label{sec:sem_enfAlg}

The semantics of the enforcement process defines how the \pep\ discharges obligations and enforces authorisation decisions. To define this process, we use the auxiliary function $\pepSemR{\ }: \mathit{IOblgation}^* \rightarrow \{ \true, \false \}$ that, given a sequence of instantiated obligations, executes such obligations and returns a boolean value that indicates whether the evaluation is successfully completed. Since failures caused by optional obligations can be safely ignored by the \pep, only failures of mandatory obligations (\ie of type $\obM$) have to be taken into account. The function is defined as follows
$$
\begin{array}{r@{\quad }c@{\ }l}
\pepSemR{\epsilon} & = & \x{true}
\\[.25cm]
\pepSemR{[\,\obO \ \ \pepAction(\extVal^*)\,]\,\concat\,\fo^*} & = & \pepSemR{\fo^*} 
\\[.25cm]
\multicolumn{3}{l}{
\begin{array}{l}
\pepSemR{[\,\obM \ \ \pepAction(\extVal^*)\,]\,\concat\,\fo^*}  =  
\left\{\begin{array}{@{}l@{\quad}l}
\pepSemR{\fo^*} & \mathtt{if}\ \pepAction(\extVal^*) \Downarrow\!\x{ok}
\\[.15cm]
\ \x{false} & \mathtt{otherwise}
\end{array}\right.
\end{array}
}
\end{array}
$$
where \mbox{$\Downarrow\!\x{ok}$} denotes if the discharge of the action $\pepAction(\extVal^*)$ succeeded. Since the set of action identifiers is intentionally left unspecified (see Section~\ref{sec:policySyntax}), the definition of the predicate \mbox{$\Downarrow\!\x{ok}$} is hence unspecified too; we just assume that it is total and deterministic. In other words, the syntactic domain $\mathit{PepAction}$ is a parameter of the syntax, while the predicate \mbox{$\Downarrow\!\x{ok}$} is a parameter of the semantics. The latter parameter could be refined to deal with, e.g., obligations to be enforced after the decision releasing (see Section~\ref{sec:conclusions}). For example, discharging obligations could simply refer to the fact that the system has taken charge of their execution, rather than to the fact that they have been completely executed.

The semantics of $\pep$ is thus defined with respect to the enforcement algorithms. Formally, given an enforcement algorithm and a \pdp\ response, the function $\denSemF{E}A: \mathit{EnfAlg} \rightarrow (\mathit{PDPReponse} \rightarrow \mathit{Decision})$ returns the enforced decision. It is defined by three clauses, one for each algorithm. The clause for the $\denyBiased$ algorithm follows
\[
\hspace*{-.4cm}
\begin{array}{l}
\pepSem{\denyBiased}{\pdpRes} =\\
\ \
\left\{
\begin{array}{l@{\quad}l}
\permit  & \mathtt{if}\ \pdpRes.\dec = \permit\ \wedge\ \pepSemR{\pdpRes.\fo} \\
\deny & \mathtt{otherwise}
\end{array}
\right.
\end{array}
\hspace*{-.3cm}
\tag{\semLbl{7a}}
\label{sem:pepA}
\]
Likewise $\pdpRes.\dec$ that indicates the decision of the response $\pdpRes$, notation $\pdpRes.\fo$ indicates the sequence of instantiated obligations of $\pdpRes$. The $\permit$ decision is enforced only if this is the decision returned by the \pdp\ and all accompanying obligations are successfully discharged. If an error occurs, as well as if the \pdp\ decision is not $\permit$, a $\deny$ is enforced. The clause for the $\permitBiased$ algorithm is the dual one, whereas the clause for the $\based$ algorithm is as follows
\[
\hspace*{-.2cm}
\begin{array}{l}
\!\! \pepSem{\based}{\pdpRes} =\\
\
\left\{
\begin{array}{l@{\quad}l}
\permit  & \mathtt{if}\ \pdpRes.\dec = \permit\ \wedge\ \pepSemR{\pdpRes.\fo}  \\[.1cm]
\deny & \mathtt{if}\ \pdpRes.\dec = \deny\ \wedge\ \pepSemR{\pdpRes.\fo} \\[.1cm]
\notApp &  \mathtt{if}\ \pdpRes.\dec =\notApp \\[.1cm]
\indet &  \mathtt{otherwise}
\end{array}
\right.
\end{array}
\hspace*{-.65cm}
\tag{\semLbl{7b}}
\label{sem:pepB}
\]
Both decisions $\permit$ and $\deny$ are enforced only if all obligations in the \pdp\ response are successfully discharged, otherwise they are enforced as $\indet$. Instead, decisions $\notApp$ and $\indet$ are enforced without modifications. 

\subsection{Semantics of the Policy Authorisation System} 
\label{sec:sem_pas}

The semantics of a Policy Authorisation System is defined in terms of the composition of the semantics of \pep\ and \pdp. It is given by the function 
$\denSemF{P}as : \mathit{PAS} \rightarrow (\mathit{Request} \rightarrow \mathit{Decision})$ defined by the following clause
\[
\begin{array}{l}
\pasSem{\{ \,  \x{pep:} \, \enfAlg\ \ \x{pdp:}\, \pdpSyntax \, \},\rSyntax} = \\[.1cm]
\qquad\qquad\quad 
\pepSem{\enfAlg}{(\pdpSem{\pdpSyntax}{(\reqSem{\rSyntax}{})})}
\end{array}
\tag{\semLbl{8}}
\label{sem:pas}
\]
Basically, given a request $\rSyntax$ in the \facpl\ syntax, this is converted into its functional representation by  the function $\denSemF{R}$ (see Section~\ref{sez:formalRequest}). This result is then passed to the semantics of the \pdp, \ie $\pdpSem{\pdpSyntax}{}$, which returns a response that on its turn is passed to the  semantics of the \pep, \ie $\pepSem{\enfAlg}{}$. The latter function returns the final decision of the Policy Authorisation System when given the request $\rSyntax$ in input.

\subsection{Properties of the Semantics} 
\label{sec:prop_sem}

We conclude this section with some properties and results regarding the \facpl\ semantics. 

The main result is that the semantics is \emph{total} and \emph{deterministic}. This means that it is defined for all possible input pairs consisting of a \facpl\ specification, \ie a Policy Authorisation System, and a request, and that it always returns the same decision any time it is applied to a specific pair.

\begin{theorem}[Total and Deterministic Semantics]
\label{theo:deterministic}\ 
\begin{enumerate}
\item For all $pas \in \mathit{PAS}$ and $\rSyntax\in\mathit{Request}$, there exists a $\dec \in\mathit{Decision}$, such that $\pasSem{pas,\rSyntax} = \dec$. 
\smallskip
\item For all $pas \in \mathit{PAS}$, $\rSyntax\in\mathit{Request}$ and $\dec,\dec' \in\mathit{Decision}$, it holds that\\[.1cm]
$
\begin{array}{l}
\quad \pasSem{pas,\rSyntax} = \dec \ \ \wedge\ \ \pasSem{pas,\rSyntax} = \dec'\\ 
\quad \Rightarrow\ \ \dec=\dec'\,.
\end{array}
$ 
\end{enumerate}
\end{theorem}
\begin{proof}
It boils down to show that $\denSemF{P}as$ is a total and deterministic function (see Appendix~\ref{app:formal_sem}).
\end{proof}


We now consider the so-called \emph{reasonability properties} of~\cite{TschantzK06} that precisely characterise the expressiveness of a policy language. \facpl\ enjoys the property called \emph{independent composition} of policies, which means that the results of the combining algorithms depend only on the decisions of the policies given in input. This clearly follows from the use of combination matrices. On the contrary, \facpl\ ensures neither \emph{safety}, i.e. a request that is granted may not be granted anymore if it is extended with new attributes, nor \emph{monotonicity}, i.e. the introduction of a new policy in a combination of policies may change a $\permit$ decision to a different one. This should be somehow expected as these latter two properties are enjoyed neither by \xacml\ nor by other policy languages featuring $\deny$ rules and combining algorithms similar to those we have presented.

We conclude by highlighting the relationship between attribute names occurring in a policy and names defined by requests. By letting $\mathit{Names}(\policy)$ to indicate the set of attribute names occurring in (the expressions within) $\policy$, we can state the following result which has important practical implications on the feasibility of the automatic analysis.

\begin{lemma}[Policy relevant attributes]  
\label{lemma1}
For all $\policy \in \mathit{Policy}$ and $\req,\req' \in R$ such that $\req(\name) = \req'(\name)$ for all $\name \in \mathit{Names}(\policy)$, it holds that 
$\policySem{\policy}{\req} = \policySem{\policy}{\req'}$.
\end{lemma}
\begin{proof}
The property straightforwardly derives from the semantics of \facpl\ expressions and from Theorem~\ref{theo:deterministic} (see Appendix~\ref{app:formal_sem}).
\end{proof}


\section{FACPL Constraint-based Representation}
\label{sec:constraint}

The analysis of access control policies is essential for ensuring confidentiality and integrity of system resources. In the case of \facpl, the analysis is made difficult by the hierarchical structure of policies, the presence of conflict resolution strategies and the intricacies deriving from the many involved controls. Moreover, no off-the-shelf analysis tool directly takes \facpl\ specifications in input. Hence, for enabling the analysis of \facpl\ policies through well-established and efficient software tools, we introduce a constraint formalism that permits, on the one hand, to uniformly represent policies and, on the other hand, to perform extensive checks of (a possibly infinite number of) requests. 

The constraint-based representation we propose specifies satisfaction problems in terms of formulae based on multiple theories as, e.g., boolean and linear arithmetics. Such kind of formulae are usually called \emph{satisfiability modulo theories} (SMT) formulae. The SMT-based approach is supported by the relevant progress made in the development of automatic SMT solvers (e.g., Z3~\cite{MouraB08}, CVC4~\cite{CVC4}, Yices~\cite{Yices}), which make SMT formulae to be extensively employed in diverse analysis applications~\cite{MouraB11}. 

This section introduces our constraint-based representation of \facpl\ policies, while the analysis it enables is presented in Section~\ref{sec:analysis}. We first introduce the constraint formalism (Section~\ref{sec:constraintFormalism}), then we present the constraint representation of \facpl\ policies (Section~\ref{sec:translation}) and some crucial results stating that it is a semantic-preserving representation (Section~\ref{sec:prop_transl}), and finally we show some examples of constraints obtained from our e-Health case study (Section~\ref{sec:constr_ehealth}). 

\subsection{A Constraint Formalism}
\label{sec:constraintFormalism}

The constraint formalism we present here extends boolean and inequality constraints with a few additional operators aiming at precisely representing FACPL constructs. Intuitively, a constraint is a relation defined  through some conditions on a set of attribute names\footnote{In the literature, constraints are typically defined on a set of \emph{variables}. In our framework, the role of variables is played by attribute names. Therefore, to maintain a coherent terminology throughout the paper, we refer to constraint variables as attribute names.}. An assignment of values to attribute names satisfies a constraint if all constraint conditions are matched. Our formalism, besides usual operators and values, explicitly considers the role of missing attributes, by assigning $\excpt$ to attribute names, and of run-time errors, \ie type mismatches in constraint evaluations. In fact, according to the usually accepted semantics of access control policies (besides \xacml, see, e.g.,~\cite{CramptonM12,CramptonW15}), a condition involving a missing attribute should not be evaluated to $\false$ by default.

\begin{table*}[!h]
\caption{Semantics of constraints ($T$ stands for one of the sets of literal values or for the powerset of the set of all literal values, and $i, j \in \{1, 2\}$ with $i \neq j$)}
\label{tab:constr_sem}
\vspace*{-.3cm}
$$
\footnotesize
\begin{array}{@{}l@{}l@{}l@{}}
\hline
\mbox{}\\[-.2cm]
\multicolumn{3}{c}{
\begin{array}{c}
\cs{\name} = \req(n)
\hspace*{4cm} 
\cs{\val} =  \val
\\[.2cm]
    \begin{array}{@{ }l}
    \cs{\isBot{\const}} = \\
    \
        \left\{
       \begin{array}{l@{\ \ }l}
    	\true & \mathtt{if}\  \cs{\const} = \excpt\\[.1cm]
       	\false & \mathtt{otherwise}\\
       \end{array}
        \right.
    \end{array}
\ \ 
    \begin{array}{@{ }l}
    \cs{\isErr{\const}} = \\
    \    \left\{
       \begin{array}{l@{\ \ }l}
    	\true & \mathtt{if}\  \cs{\const} = \err\\[.1cm]
       	\false & \mathtt{otherwise}\\
       \end{array}
        \right.
    \end{array}  
\ \  
{
    \begin{array}{@{ }l}
    \cs{\isBool{\const}} = \\
    \    \left\{
       \begin{array}{l@{\ \ }l}
    	\true & \mathtt{if}\  \cs{\const} \in \{ \true, \false\}\\[.1cm]
       	\false & \mathtt{otherwise}\\
       \end{array}
        \right.
    \end{array}  }
\end{array}
}
\\
\mbox{}\\[-.1cm]
\begin{array}{@{ }l}
\cs{\fnot\ \const} = \\
\    \left\{
   \begin{array}{l@{\ \ }l}
      	\true  & \mathtt{if}\   \cs{\const} = \false \\[.1cm]
	\false  & \mathtt{if}\   \cs{\const} = \true \\[.1cm]
	\excpt & \mathtt{if}\   \cs{\const} = \excpt \\[.1cm]
	\err & \mathtt{otherwise} \\
   \end{array}
    \right.   
\end{array}    
&
\begin{array}{@{ }l}
\cs{\const_1\ \fand\ \const_2} = \\
\    \left\{
   \begin{array}{l@{\ \ }l}
      	\true  & \mathtt{if}\   \cs{\const_1} =  \cs{\const_2} = \true \\[.1cm]
	\false  & \mathtt{if}\   \cs{\const_1} = \false 
	\ \ \mathtt{or}\ \ \cs{\const_2} = \false \\[.1cm]
	\excpt & \mathtt{if}\   \cs{\const_i} = \excpt
	\ \ \mathtt{and}\ \cs{\const_j} \in \{\true,\excpt\} \\[.1cm] 
	\err & \mathtt{otherwise} \\
   \end{array}
    \right.  
\end{array}
&
\begin{array}{@{ }l@{\hspace*{-.2cm}}}
\cs{\const_1\ \for\ \const_2} = \\
\    \left\{
   \begin{array}{l@{\ \ }l}
      	\true  & \mathtt{if}\   \cs{\const_1} = \true 
	\ \ \mathtt{or}\ \ \cs{\const_2} = \true \\[.1cm]
	\false  & \mathtt{if}\   \cs{\const_1} =  \cs{\const_2} = \false \\[.1cm]
	\excpt & \mathtt{if}\   \cs{\const_i} = \excpt 
	\ \ \mathtt{and}\ \cs{\const_j} \in \{\false,\excpt\} \\[.1cm] 
	\err & \mathtt{otherwise} 
   \end{array}
    \right.
\end{array}
\\
\mbox{}\\[-.1cm]
\begin{array}{@{ }l}
\cs{\lnot\ \const} = \\
\quad    \left\{
   \begin{array}{l@{\ \ }l}
      	\true  & \mathtt{if}\   \cs{\const} = \false \\
	&\ \ \mathtt{or}\ \ \cs{\const} = \excpt  \\[.1cm]
	\false  & \mathtt{otherwise} \\
   \end{array}
    \right.
\end{array}
&
\begin{array}{@{ }l}
\cs{\const_1\ \wedge\ \const_2} = \\
\ 
    \left\{
   \begin{array}{l@{\ \ }l}
      	\true  & \mathtt{if}\   \cs{\const_1} = \true 
	\ \ \mathtt{and}\ \ \cs{\const_2} = \true  \\[.1cm]
	\false  & \mathtt{otherwise} \\
   \end{array}
    \right.
 \end{array}
&
\begin{array}{@{ }l}
\cs{\const_1\ \vee\ \const_2} = \\
\quad    \left\{
   \begin{array}{l@{\ \ }l}
      	\true  & \mathtt{if}\   \cs{\const_1} = \true 
	\ \ \mathtt{or}\ \ \cs{\const_2} = \true  \\[.1cm]
	\false  & \mathtt{otherwise} \\
   \end{array}
    \right.
\end{array}
\\
\mbox{}\\[-.1cm]
\multicolumn{3}{c}{
\begin{array}{@{ }l}
\cs{\const_1\ =\ \const_2} = \\
\    \left\{
   \begin{array}{l@{\ \ }l}
      	\true  & \mathtt{if}\   
	\cs{\const_1}, \cs{\const_2} \!\in\! \type
	\ \mathtt{and}\  
	\cs{\const_1} = \cs{\const_2}\\[.1cm]
	\false  & \mathtt{if}\
	\cs{\const_1}, \cs{\const_2} \!\in\! \type
	\ \mathtt{and}\ 
	\cs{\const_1} \neq \cs{\const_2} \\[.1cm]
	\excpt & \mathtt{if}\   \cs{\const_i} = \excpt	\, \mathtt{and}\ \cs{\const_j} \neq \err \\[.1cm] 
	\err & \mathtt{otherwise}
   \end{array}
    \right.  
 \end{array}
\ \
\begin{array}{@{ }l@{}}
\mbox{}\\[-.75cm]
\cs{\const_1\ +\ \const_2} = \\
\quad    \left\{
   \begin{array}{l@{\  }l}
      	\cs{\const_1} + \cs{\const_2}  & \mathtt{if}\ 
	\ \cs{\const_1}, \cs{\const_2} \!\in\! \doubleT
	 \\[.1cm]
	\excpt & \mathtt{if}\   \cs{\const_i} = \excpt  \mathtt{and}\ \cs{\const_j} \neq \err \\[.1cm]
	\err & \mathtt{otherwise} \\
   \end{array}
    \right.
\end{array}
}
\\
\mbox{}\\[-.2cm]
\hline
\end{array}
$$
\end{table*}

\smallskip
\noindent
\emph{Syntax}.
Constraints are written according to the following grammar. 
$$
\begin{array}{@{}l@{\, }l@{ \, }l}
\mathit{Constr} & ::= & \mathit{Value} \Sep \mathit{Name} \Sep \isBot{\mathit{Constr}} \\[.1cm]
& \ \ \mid &  \isErr{\mathit{Constr}} \Sep {\isBool{\mathit{Constr}}}\\[.1cm]
 & \  \ \mid &  \lnot\, \mathit{Constr} \Sep \fnot\, \mathit{Constr} \Sep \mathit{Constr}\, \opC\, \mathit{Constr}  
\\[.2cm]
\ \opC\ & ::= & \ \wedge \ \Sep \ \vee \ \Sep \ \fand \ \Sep \ \for \ \Sep \ = \ \Sep \ > \ \Sep \ \in \\
  & \ \ \mid &  \ + \, \Sep \ - \, \Sep \ \ast \ \, \Sep \ /\\[.15cm]

\end{array}
$$
where the nonterminals $\mathit{Value}$ and $\mathit{Name}$ are defined in Table~\ref{tab:facpl_syntax}. Thus, a constraint can be a literal value, an attribute name, or a more complex constraint obtained through predicates $\isBot{}$, $\isErr{}$ and $\isBool{}$, or through boolean, comparison and arithmetic operators. The operators $\lnot$, $\wedge$ and $\vee$ are the usual boolean ones, while $\fnot$, $\fand$ and $\for$ correspond to the 4-valued ones of \facpl\ expressions which implement the special management of $\excpt$ and $\err$ values.

In the sequel, in addition to the notations of Table~\ref{table:syntsemcorr}, we use the letter $\const$ to denote a generic element of the set of all constraints identified by the nonterminal $\mathit{Constr}$.

\smallskip
\noindent
\emph{Semantics}.
The semantics of constraints is modelled by the function $\constrSem : \mathit{Constr} \rightarrow (R \rightarrow \mathit{Value}\ \cup\ 2^{\mathit{Value}}\cup \{\err, \excpt \})$ inductively defined by the clauses in Table~\ref{tab:constr_sem} (the clauses for $>$, $\in$, $-$, $\ast$ and $/$ are omitted as they are similar to those for $=$ or $+$). Hence, the semantics of a constraint is a function that, given the functional representation of a request (i.e., an assignment of values to attribute names), returns a literal value or a set of literal values or one of the special values $\excpt$ and $\err$.

The semantics of constraints, except for the cases of predicates and usual boolean operators, mimics the semantic definitions of the corresponding \facpl\ expression operators defined in Table~\ref{tab:sem_expression} (e.g., the constraint operator $\for$ corresponds to the expression operator $\x{or}$, as well as $+$ corresponds to $\x{add}$). The clause defining the semantics of predicate $\isBot{\const}$ (resp. $\isErr{\const}$) returns $\true$ only if the constraint $\const$ evaluates to $\excpt$ (resp. $\err$), while that of predicate $\isBool{\const}$ returns $\true$ only if the constraint $\const$ evaluates to a boolean value. The clauses for usual boolean operators are instead defined by ensuring that only boolean values can be returned. Specifically, they explicitly define conditions leading to result $\true$, while in all the other cases the result is $\false$. The constraint $\lnot\, \const$ evaluates to $\true$ not only when the evaluation of $\const$ returns $\false$, but also when it returns $\excpt$. This is particularly convenient for translating \facpl\ policies because, in case of $\notApp$ decisions, $\excpt$ is treated as $\false$.

\subsection{From FACPL Policies  to Constraints}
\label{sec:translation}

The constraint-based representation of a \facpl\ policy is a logical combination of the constraints representing targets, obligations and combining algorithms occurring within the policy. Of course, combining algorithms using the $\greedy$ instantiation strategy are not dealt with, as we cannot statically predict when the (sequential) evaluation of a sequence of policies can stop since the decision, that would have resulted from evaluating the whole sequence, has been obtained.
The translation is formally, and \emph{compositionally}, defined by a family of translation functions $\translSymbol_\cdot$, that return the constraints representing the different \facpl\ terms. We use the emphatic brackets $\{\!|$ and $|\!\}$ to represent the application of a translation function to a syntactic term. 

We start by presenting the translation of \facpl\ expressions, whose operators are very close to (some of) those on constraints. The translation is formally given by the function $\transFunct{E}: \mathit{Expr} \rightarrow \mathit{Constr}$, whose defining clauses are given below
\[
\hspace*{-.5cm}
\begin{array}{@{}l@{}}
\translExpr{\mathit{\val}} =  \val
\qquad\qquad\qquad\qquad\qquad
\translExpr{\mathit{\name}}=  \name
\\[.2cm]
\translExpr{\x{not}(\expr)} = \fnot{\translExpr{\expr}} 
\\[.2cm]
\translExpr{\exprOperator(\expr_1, \expr_2)} = 
\translExpr{\expr_1}\: \mathtt{getCop}(\exprOperator) \: \translExpr{\expr_2} 
\end{array}
\hspace*{-1.3cm}
\tag{\conLbl{1}}
\label{cstr:expr}
\]
Thus, $\transFunct{E}$ acts as the identity function on attribute names and values, and as an homomorphism on operators. In fact, \facpl\ negation corresponds to the constraint operator $\fnot$, while the binary \facpl\ operators correspond to the constraint operators returned by the auxiliary function $\mathtt{getCop}()$. Its definition is straightforward, the main cases are defined as follows
\begin{center}
$
\small
\begin{array}{l@{\quad\qquad}l}
\mathtt{getCop}(\x{and})\!=\ \fand
&
\mathtt{getCop}(\x{or}) \!=\ \for
\\[.15cm]
\mathtt{getCop}(\x{equal}) \!=\ \ =
&
\mathtt{getCop}( \x{in}) \!=\ \in
\\[.15cm]
\mathtt{getCop}(\x{greater}\textrm{-}\x{than}) \!=\ >
&
\mathtt{getCop}(\x{add}) \!=\ +
\\[.15cm]
\end{array}
$
\end{center}

The translation of (sequences of) obligations returns a constraint whose satisfiability corresponds to the successful instantiation of all the input obligations. The translation function $\transFunct{Ob}: \mathit{Obligation}^* \rightarrow \mathit{Constr}$ is defined below
\[
\hspace*{-.4cm}
\begin{array}{l}
\translObl{\epsilon} = \true 
\\[.2cm]
\translObl{\ob\, \ob^*} = \translObl{\ob}\,\wedge\,\translObl{\ob^*}
\\[.2cm]
\translObl{[\obType\ \mathit{PepAction}(\expr^*)]} =
        \bigwedge_{\expr \in \expr^*} 
	        \lnot\isBot{\translExpr{\expr}} 
	     \wedge \lnot\isErr{\translExpr{\expr}}\\
\end{array}
\hspace*{-.4cm}
\tag{\conLbl{2}}
\label{cstr:obl}
\]
Hence, a sequence of obligations corresponds to the conjunction of the constraints representing each obligation. When translating a single obligation, predicates $\isBot{}$ and $\isErr{}$ are used to check the instantiation conditions, \ie that the occurring expressions cannot evaluate to $\excpt$ or $\err$. The n-ary conjunction operator returns $\true$ if the considered obligation contains no expression (i.e., $\expr^* = \epsilon$).

The translation function for policies, $\transFunct{P}$, exploits the translation functions previously introduced, as well as a function $\transFunct{A}$ representing the result of applying a combining algorithm to a sequence of policies. Functions $\transFunct{P}$ and $\transFunct{A}$ are indeed mutually recursive. Moreover, for representing all the decisions that a policy can return, both these two functions return 4-tuples of constraints of the form
$$
\langle \permit: \const_p\ \ \ \deny: \const_d\ \ \ \notApp: \const_n\ \ \ \indet: \const_i\ \rangle
$$
where each constraint represents the conditions under which the corresponding decision is returned. We call these tuples  \emph{policy constraint tuples} and denote their set by $\AT$. As a matter of notation, we will use the projection operator $\proj{l}$ which, when applied to  a constraint tuple, returns the value of the field labelled by $l'$, where $l$ is the first letter of $l'$ (e.g., $\proj{p}$ returns the $\permit$ constraint $\const_p$).

The function $\transFunct{P}: \mathit{Policy} \rightarrow \AT$ is defined by two clauses for rules, i.e.~one for each effect, and one clause for policy sets. The clause for rules with effect $\permit$ is 
\[
\hspace*{-.4cm}
\begin{array}{l}
\translPol{\ruleOpt{\mathit{\permit}\ \ \x{target:} \, \expr\ \ \x{obl:} \, \ob^{*} \,}} = \\[.1cm]
 \, 
\begin{array}{l}
\langle\:
\permit: \
\translExpr{\expr} \wedge\ \translObl{\ob^{*}}  \\[.1cm]
\,
\deny: \ \false \\[.1cm]
 \,
\notApp:\ \neg\ \translExpr{\expr}  \\[.1cm]
 \, \indet: 
\lnot\, (\isBool{\translExpr{\expr}}\ \vee\ \isBot{\translExpr{\expr}})
\vee\, (\translExpr{\expr} \wedge\, \lnot\, \translObl{\ob^{*}})\ 
\,\rangle 
\end{array}
\end{array}
\hspace*{-2cm}
\tag{\conLbl{3a}}
\label{cstr:rule}
\]
(the clause for effect $\deny$ is omitted, as it only differs from the previous one because it swaps the $\permit$ and $\deny$ constraints). The clause takes into account the rule constituent parts and combines them according to the rule semantics (see clause~(\ref{sem:rule})). Because of the semantics of the constraint operator $\lnot$, the $\notApp$ constraint is satisfied when the constraint corresponding to the target expression evaluates to $\false$ or to $\excpt$. Instead, the negation of a constraint corresponding to a sequence of obligations represents the failure of their instantiation. 
In the $\indet$ constraint, together with condition $\lnot\ \isBool{\translExpr{\expr}}$, we introduce $\lnot\ \isBot{\translExpr{\expr}}$ because we want to exclude that $\translExpr{\expr}\,= \excpt$ (otherwise, we would fall in the case of decision $\notApp$ ).

The clause for policy sets is as follows
\[
\hspace*{-.3cm}
\begin{array}{l}
\translPol{\langle\ \algSyntax\ \ \x{target:} \, \expr\ \ \x{policies:} \, \policy^{+}\, {\x{\oblp:} \, \ob_{p}^{*}\ \ \x{\obld:} \, \ob_{d}^{*}}\ \rangle} \\
\ = \\[.15cm]
\!\! 
\begin{array}{l}
\langle\:
\permit: \
\translExpr{\expr} \wedge\, \translAlg{\algSyntax,\policy^+}\proj{p} \wedge\, \translObl{{\ob_{p}^{*}}}  \\[.15cm]
 \ \ \,
\deny: \
\translExpr{\expr} \wedge\, \translAlg{\algSyntax,\policy^+}\proj{d} \wedge\, \translObl{{\ob_{d}^{*}}}  \\[.15cm]
 \ \ \,
\notApp:\
\lnot\ \translExpr{\expr} \vee (\translExpr{\expr} 
\wedge\, \translAlg{\algSyntax,\policy^+}\proj{n})\\[.15cm]
 \ \ \, 
\indet: \\
\qquad\ \ 
\lnot\, (\isBool{\translExpr{\expr}}\, \vee\, \isBot{\translExpr{\expr}})\\ 
 \qquad\ \ \vee\  (\translExpr{\expr} \wedge\, \translAlg{\algSyntax,\policy^+}\proj{i}) \\
 \qquad\ \ \vee\ (\translExpr{\expr} \wedge\, \translAlg{\algSyntax,\policy^+}\proj{p} \wedge \, \lnot\ \translObl{{\ob_{p}^{*}}}) \\
 \qquad\ \ \vee\ (\translExpr{\expr} \wedge\, \translAlg{\algSyntax,\policy^+} \proj{d} \wedge \, \lnot\ \translObl{{\ob_{d}^{*}}} \,)\
\,\rangle\\[.2cm]
\end{array}
\end{array}
\hspace*{-2.2cm}\tag{\conLbl{3b}}
\label{cstr:pol}
\]
With respect to the clauses for rules, it additionally takes into account the result of the application of the combining algorithm according to the policy set semantics (see clause~(\ref{sem:pol})). It is worth noticing that the exclusive use of operators $\lnot$, $\wedge$ and $\vee$ ensures that constraint tuples are only formed by boolean constraints.


Combining algorithms are dealt with by the function $\transFunct{A}: \mathit{Alg} \times \mathit{Policy^+} \rightarrow \AT$ that, given an algorithm (using the $\x{all}$ instantiation strategy) and a sequence of policies, returns a constraint tuple representing the result of the algorithm application. Its definition is
\[
\begin{array}{l}
\translAlg{\alg{\all}, \policy_1\: \ldots \: \policy_s} = \\[.1cm]
\quad\
\alg{}(\ldots \alg{}(\translPol{\policy_1}, \translPol{\policy_2}),\ldots,\translPol{\policy_s})
\end{array}
\hspace*{-.4cm}
\tag{\conLbl{4}}
\label{cstr:alg}
\]
By means of $\transFunct{P}$, the policies given in input are translated into constraint tuples which are then iteratively  combined, two at a time, according to the algorithm combination strategy. By way of example, the combination of two constraint tuples, say $A$ and $B$, according to the $\permitOverO{}$ algorithm, is defined as follows
$$
\begin{array}{@{}l@{}}
\permitOverO{}(A, B)  = \\[.1cm]
\,
\begin{array}{l@{\ \  }l}
  \langle \permit :  &A \proj{p} \vee  B \proj{p} \\
\, \deny: &(A \proj{d} \wedge B \proj{d}) \vee (A \proj{d} \wedge B \proj{n}) 
 \vee (A \proj{n} \wedge B \proj{d}) \\
\, \notApp: & A \proj{n} \wedge B \proj{n} \\
\, \indet :  &(A \proj{i} \wedge \lnot B \proj{p})   \vee (\lnot A \proj{p} \wedge B \proj{i})  \rangle
\end{array}
\end{array}
$$
The combinations for the remaining algorithms are in Appendix~\ref{sec:appendixA}. If $s=1$, \ie there is only one argument tuple, all the algorithms leave the input tuple unchanged, but for $\permitUnlessO{}$, which given an input tuple $A$ returns the tuple
$$
\begin{array}{l@{\qquad \ }l}
\langle 
\, \permit : A\proj{p} \vee\ A\proj{n} \vee\ A\proj{i} \quad
& 
\deny : A\proj{d} \quad \\[.2cm]
\ \ \notApp : \false \quad
&
\indet : \false \,
\rangle
\end{array}
$$
and $\denyUnlessO{}$, which behaves similarly.

Finally, the translation of top-level \pdp\ terms $\pdpPol{\algNT\ }{\x{policies:} \, \mathit{Policy}^{+}}$ is the same as that of the corresponding policy sets with target $\true$ and no obligations, \ie $\{ \algNT\ \ \x{target:} \, \true\ \ \x{policies:} \, \mathit{Policy}^{+} \, \}$.
\smallskip

\subsection{Properties of the Translation}
\label{sec:prop_transl}

The key result regarding the translation is that the semantics of the constraint-based representation of a policy and the semantics of the policy itself do agree. This correspondence is clearly limited to only those policies using the instantiation strategy $\all$. Before presenting this result, we show 
for the constraint semantics a result analogous to Theorem~\ref{theo:deterministic}.

\begin{theorem}[Total and Deterministic Constraint Semantics]
\label{thr:constr_fun}\ 
\begin{enumerate}
\item For all $\const \in \mathit{Constr}$ and $\req \in R$, there exists an $\mathit{el} \in (\mathit{Value}\ \cup\ 2^{\mathit{Value}}\cup \{\err, \excpt \})$, such that $\cs{\const} = \mathit{el}$. 
\smallskip
\item For all $\const \in \mathit{Constr}$, $\req \in R$ and $\mathit{el}, \mathit{el}' \in (\mathit{Value}\ \cup\ 2^{\mathit{Value}}\cup \{\err, \excpt \})$, it holds that
$$
\cs{\const} = \mathit{el} \ \  \wedge \ \ \cs{\const} = \mathit{el}' \ \ \Rightarrow \ \ \mathit{el} = \mathit{el}'\,.
$$
\end{enumerate}
\end{theorem}
\begin{proof}
By structural induction on the syntax of $\const$ (see Appendix~\ref{sec:appendix4}).
\end{proof}

\begin{theorem}[Policy Semantic Correspondence]
\label{thr:constr_sem}
For all $\policy \in \mathit{Policy}$ enclosing combining algorithms only using $\all$ as instantiation strategy, and $\req \in R$, it holds that
$$
\policySem{\policy}{\req}= \langle \dec \ \fo^* \rangle
\ \ \Leftrightarrow \ \ 
\cs{\translPol{\policy}\proj{\dec}} =  \true
$$
\end{theorem}
\begin{proof}
The proof (see Appendix~\ref{sec:appendix4}) is by induction on the \emph{depth}, \ie the nesting level, of $\policy$ and relies on three auxiliary correspondence results regarding expressions (Lemma~\ref{lemma:expr}), obligations (Lemma~\ref{lemma:obl}) and combining algorithms (Lemma~\ref{lemma:alg}).
\end{proof}

This theorem implies that the properties verified over the constraints resulting from the translation of a \facpl\ policy would return the same results as if they were directly proven on the \facpl\ policy itself. Thus, it ensures that the analysis we present in Section~\ref{sec:analysis} is sound. 

From the previous theorems it follows that policy constraint tuples partition the set of input requests, in other words each access request satisfies only one of the constraints of a policy constraint tuple. Essentially, the following corollary extends Theorem~\ref{thr:constr_fun} to constraint tuples.

\begin{corollary}[Constraint-based partition]
\label{theo:partition}
For all $\req \in R$ and $\policy \in \mathit{Policy}$, such that $\translPol{\policy} =  \langle \permit: \const_1\  \deny: \const_2\ \notApp: \const_3\  \indet: \const_4\, \rangle$, it holds that
$$
\begin{array}{l}
\exists! k\in\{1,\ldots,4\}\ : 
  \cs{\const_k} = \true 
\ \wedge\ 
\textstyle{\bigwedge_{j \in\{1,\ldots,4\}\backslash\{k\} }} \cs{\const_j} = \false
\end{array}
$$
\end{corollary}
\begin{proof}
The thesis immediately follows from Theorems~\ref{thr:constr_fun} and~\ref{thr:constr_sem}.
\end{proof}

\subsection{Constraint-based Representation of the e-Health case study}
\label{sec:constr_ehealth}

We now apply the translation functions introduced in Section~\ref{sec:translation} to (a part of) the considered case study. For the sake of presentation, we shorten the attribute names used within policies. For instance, the rule addressing Requirement (\sr{1}) becomes as follows
$$
\begin{array}{@{}l@{}}
\ruleOpt{\,\permit\ \
\targetBegin
\streq(\x{sub/role},``\x{doctor}")\\
\qquad\quad\quad
\x{and}\
\streq(\x{act/id},``\x{write}") \\
\qquad\quad\quad
\x{and}\ \x{in}(``\x{e}\textrm{-}\x{Pre}\textrm{-}\x{Write}",\x{sub/perm})\\
\qquad\quad\quad
\x{and}\ \x{in}(``\x{e}\textrm{-}\x{Pre}\textrm{-}\x{Read}",\x{sub/perm})}\\
\end{array}
$$
Its translation starts by applying function $\transFunct{E}$ to the target expression. The resulting constraint is as follows
$$
\begin{array}{@{}l@{}}
\const_{trg1} \define 
\x{sub/role} = ``\x{doctor}"\ \fand\ \x{act/id} = ``\x{write}" \\ 
\qquad\quad 
 \fand\ ``\x{e}\textrm{-}\x{Pre}\textrm{-}\x{Write}" \in \x{sub/perm}
\\ 
\qquad\quad
\fand\ ``\x{e}\textrm{-}\x{Pre}\textrm{-}\x{Read}" \in \x{sub/perm} 
\end{array}
$$
The translation proceeds by considering obligations; in this case they are missing  (i.e., they correspond to the empty sequence $\epsilon$), hence the constraint $\true$ is obtained. Function $\transFunct{P}$ finally defines the constraint tuple for the rule as follows
$$
\hspace*{-.2cm}
\begin{array}{l@{}}
\langle \permit : \ \const_{trg1} \wedge \true  
\\[.1cm] 
\ \, \deny:\ \false 
\\[.1cm] 
\  \, \notApp:  \ \lnot \const_{trg1}
\\[.1cm]
\ \, \indet:
\lnot (\isBool{\const_{trg1}}  \vee \isBot{\const_{trg1}}) 
\vee (\const_{trg1} \wedge \lnot \true) \rangle
\end{array}
$$
The tuples for the rules addressing Requirements (\sr{2}) and (\sr{3}) are defined similarly, they only differ in the constraints representing their targets, which are denoted as $\const_{trg2}$ and $\const_{trg3}$, respectively. 

We can now define the constraint-based representation of Policy~(\ref{policy1}). Besides the target expression, which is straightforwardly translated to the constraint $\const_{trgP} \define\, \x{res/typ} = ``\x{e}\textrm{-}\x{Pre}"$, the constraint tuple is built up from the result of function $\transFunct{A}$ representing the application of the algorithm $\permitOverO{}$. Specifically, the constraint tuples of rules are iteratively combined according to the definition of $\permitOverO{}(A, B)$ previously reported.
For example, the combination of the first two rules generates the following tuple
$$
 \begin{array}{@{}l@{}}
\langle\, \permit :\, (\const_{trg1} \wedge \true) \vee (\const_{trg2} \wedge \true)\\
\, \deny : \, (\false \wedge \false) \vee (\false \wedge \lnot \const_{trg2}) \vee (\lnot \const_{trg_1} \wedge \false)\\
\, \notApp : \ \lnot \const_{trg1} \wedge \lnot \const_{trg2}\\
\, \indet : 
(
(
\lnot (\isBool{\const_{trg1}} \vee\, \isBot{\const_{trg1}})  \\
\qquad\qquad\quad \vee (\const_{trg1} \wedge \lnot \true)
)
\wedge \lnot (\const_{trg2} \wedge \true)
)\\
\qquad\quad
\ \ \ 
\vee 
(
\lnot (\const_{trg1} \wedge \true) \wedge\ (\lnot (\isBool{\const_{trg2}}\,\\
\qquad\qquad\quad
 \vee\, \isBot{\const_{trg2}}) \vee (\const_{trg2} \wedge \lnot \true))
)
\ \,\rangle
\end{array}
$$ 
Notably, the $\deny$  constraint is never satisfied, because it is a disjunction of conjunctions having at least one $\false$ term as argument. This is somewhat expected, because the rules have the $\permit$ effect and the used combining algorithm is $\permitOverO{}$. This tuple is then combined with that of the remaining rule in a similar way. 

To generate the constraint tuple of the policy, we also need the constraint-based representation of its obligations. The policy contains only one obligation for the effect $\permit$, whose corresponding constraint is as follows
$$
\begin{array}{l}
\!\!\const_{obl\_p} \define
\bigwedge_{n \in \{\x{sys/time}, \x{res/typ}, \x{sub/id}, \x{act/id}\}} \lnot \isBot{n} 
\wedge \lnot \isErr{n}
\end{array}
$$
The constraint corresponding to obligations for the effect $\deny$, which are missing, is instead $\true$.

Finally, the constraint tuple of Policy~(\ref{policy1}) generated by function $\transFunct{P}$ is as follows
$$
\begin{array}{@{}l@{}}
\hspace*{-.5cm}\!\!\! \langle \permit : 
\const_{trgP}
\wedge 
((\const_{trg1} \wedge \true) \vee (\const_{trg2} \wedge \true) \vee (\const_{trg3} \wedge \true))
 \wedge \const_{obl\_p} \\[.1cm]
 \hspace*{-.5cm}\deny : 
 \const_{trgP}
\wedge (
(((\false \wedge \false) \vee (\false \wedge \lnot \const_{trg2}) 
\vee (\lnot \const_{trg_1} \wedge \false)) \wedge \false) \\
\qquad\quad
\vee (((\false \wedge \false) \vee (\false \wedge \lnot \const_{trg2}) 
\vee (\lnot \const_{trg_1} \wedge \false)) \wedge \lnot \const_{trg3}) \\
\qquad\quad
\vee ((\lnot \const_{trg1} \wedge \lnot \const_{trg2}) \wedge \false)
)
\wedge\ \true \\[.1cm]
\hspace*{-.5cm}\notApp:  
\lnot 
\const_{trgP}
\vee (
\const_{trgP}  \wedge 
(\lnot \const_{trg1} \wedge \lnot \const_{trg2} \wedge \lnot \const_{trg3}))
\\[.1cm]
\hspace*{-.5cm}\indet:  
\lnot(\isBool{\const_{trgP}} 
 \, \vee\, \isBot{\const_{trgP}}) 
 \\
\vee 
(
\const_{trgP} 
 \wedge
( 
( (\lnot (\isBool{\const_{trg1}}\, \vee\, \isBot{\const_{trg1}}) 
\vee (\const_{trg1} \wedge \lnot \true))  \wedge \lnot (\const_{trg2} \wedge \true)) \\
\vee\: \lnot ( (\const_{trg1} \wedge \true)  \wedge ( \lnot (\isBool{\const_{trg2}} \vee \isBot{\const_{trg2}}) \\
\qquad\quad
 \vee (\const_{trg2} \wedge \lnot \true)))
\wedge \lnot (\const_{trg3} \wedge \true)
) 
\vee\: 
(
\lnot ((\const_{trg1} \wedge \true) \vee (\const_{trg2} \wedge \true))
\wedge 
(\lnot (\isBool{\const_{trg3}} \, \\
\qquad\quad
\vee\, \isBot{\const_{trg3}})  \vee (\const_{trg3} \wedge \lnot \true))
)\\
\vee 
(
\const_{trgP}  \wedge 
((\const_{trg1} \wedge \true) \vee (\const_{trg2} \wedge \true) 
\vee (\const_{trg3} \wedge \true)) 
\wedge \lnot \const_{obl\_p}
)
\\
\vee 
(
\const_{trgP} \wedge 
((((\false \wedge \false) \vee (\false \wedge \lnot \const_{trg2})
\vee (\lnot \const_{trg_1} \wedge \false)) \wedge \false) \\
 \vee (((\false \wedge \false) \vee (\false \wedge \lnot \const_{trg2}) \vee (\lnot \const_{trg_1} \wedge \false))
  \wedge\  \lnot \const_{trg3}) \\
 \vee ((\lnot \const_{trg1} \wedge \lnot \const_{trg2}) \wedge \false))
 \wedge \lnot \true
 )
 \rangle
\end{array}
$$

As this example demonstrates, the constraints resulting from the translation are a single-layered representation of policies that fully details all the aspects of policy evaluation. However, it is also evident that the evaluation, as well as the generation, of such constraints cannot be done manually, but requires a tool support. 


\section{Analysis of FACPL Policies}
\label{sec:analysis}

The analysis of \facpl\ policies we propose aims at verifying different types of properties by exploiting the constraint-based representation of policies. We first formalise a relevant set of properties in terms of expected authorisations for requests, and then we define the strategies for their automated verification by means of constraints.

Furthermore, since \facpl\ does not enjoy the \emph{safety} property (see Section~\ref{sec:prop_sem}), the analysis investigates how the extension of a request through the addition of further attributes might change its authorisation in a possibly unexpected way. Intuitively, it is important to consider the authorisation decisions not only of specific requests, but also of their extensions because, e.g., a malicious user could try to exploit them to circumvent the access control system. This analysis approach is partially inspired by the probabilistic analysis on missing attributes introduced in~\cite{CramptonMZ15}.

In the following, we first formalise the proposed properties (Section~\ref{sec:prop}) and present some concrete examples of them from the case study (Section~\ref{sec:prop_ex}). Afterwards, we show how to express the constraint formalism into a tool-accepted specification (Section~\ref{sec:contr_eval}) and exploit it to automatically verify the properties with an SMT solver (Section~\ref{sec:prop_verify}).

\subsection{Formalisation of Properties}
\label{sec:prop}

We consider both properties that refer to the expected authorisation of single requests, \ie \emph{authorisation properties} (Section~\ref{sec:sec_prop}), and to the relationships among policies on the base of the whole set of authorisations they establish, \ie \emph{structural properties} (Section~\ref{sec:struct_prop}); afterwards we comment on their automatic verification (Section~\ref{sec:towautverif}).

\subsubsection{Authorisation Properties} 
\label{sec:sec_prop}

To formalise the authorisations properties, we introduce the notion of \emph{request extension set} of a given request $\req$. It is defined as follows
$$
\Ext{\req} \define \{ \req' \in R \mid \req(\name) \neq \excpt \ \ \ \Rightarrow \ \ \req'(n) = \req(n)\}
$$
The set is formed by all those requests that possibly extend request $\req$ with new attributes not already defined by $\req$. 


\smallskip
\noindent
\textit{Evaluate-To}. This property, written $\Eval{\req}{\dec}$, requires the policy under examination to evaluate the request $\req$ to decision $\dec$. The satisfiability, written $\SatLit$, of the \emph{Evaluate-To} property by a policy $\policy$ is defined as follows
$$
\Sat{\policy}  \Eval{\req}{\dec} \qquad \mathit{iff} \qquad
\policySem{\policy}{\req} = {\langle \dec\ \ \fo^*\rangle}
$$
In practice, the verification of the property boils down to apply the semantic function $\denSemF{P}$ to $\policy$ and $\req$, and to check that the resulting decision is $\dec$.

\smallskip
\noindent
\textit{May-Evaluate-To}. This property, written $\May{\req}{\dec}$, requires that \emph{at least one} request extending the request $\req$ evaluates to decision $\dec$. The satisfiability of the \emph{May-Evaluate-To} property by a policy $\policy$ is defined as follows
$$
\begin{array}{l}
\Sat{\policy}  \May{\req}{\dec} \quad \mathit{iff} \
\exists\, \req' \in \Ext{\req}\  :\  \policySem{\policy}{\req'} = {\langle \dec\ \ \fo^*\rangle} 
\end{array}
$$
This property, as well as the next one, addresses additional attributes extending the request $\req$ by considering the requests in its extension set $\Ext{\req}$.

\smallskip
\noindent
\textit{Must-Evaluate-To}. This property, written $\Must{\req}{\dec}$, differs from the previous one as it requires \emph{all} the extended requests to evaluate to decision $\dec$. The satisfiability of the \emph{Must-Evaluate-To} property by a policy $\policy$ is defined as follows
$$
\begin{array}{l}
\Sat{\policy}  \Must{\req}{\dec} \qquad \mathit{iff}\
 \forall \req' \in \Ext{\req}\  :\  \policySem{\policy}{\req'} = {\langle \dec\ \ \fo^*\rangle}
\end{array}
$$

\medskip

Of course, additional properties can be obtained by combining the previous ones like, e.g., a property requiring that all requests in $\Ext{\req}$ may evaluate to $\dec$ and must not evaluate to $\dec'$. Again, request extensions can be exploited to track down possibly unexpected authorisations.

\subsubsection{Structural Properties}
\label{sec:struct_prop}

A structural property refers to the structure of the sets of authorisations established by one or multiple policies. In case of multiple policies, the properties aim at characterising the relationships among the policies. Different structural properties have been proposed in the literature (e.g. in~\cite{FislerKMT05} and~\cite{KolovskiHP07}) by pursuing different approaches for their definition and verification. Here, we consider a set of commonly addressed properties and provide a uniform characterisation thereof in terms of requests and policy semantics. 

\smallskip
\noindent
\textit{Completeness}. A policy is $\complete$ if it applies to all requests. Thus, the satisfiability of the \emph{Completeness} property by a policy $\policy$ is defined as follows
$$
\begin{array}{l}
\policy\ \mathtt{sat}\ \complete\ \qquad\mathit{iff}\
\forall\ \req\ \in \Req \ : \ \policySem{\policy}{\req} = {\langle \dec \ \ \fo^*\rangle, \dec} \neq \notApp
\end{array}
$$
Essentially, we require that the policy applies to any request, \ie it always returns a decision different from $\notApp$. Notably, in this formulation $\indet$ is considered as an acceptable decision; a more restrictive formulation could only accept $\permit$ and $\deny$.

\smallskip
\noindent
\textit{Disjointness}. Disjointness among policies means that such policies apply to disjoint sets of requests. Thus, this property, written $\disjoint\ \policy'$, requires that there is no request for which both the policy under examination and the policy $\policy'$ evaluate to a decision considered \emph{admissible}, \ie $\permit$ or $\deny$. The satisfiability of the \emph{Disjointness} property by a policy $\policy$ is defined as follows
$$
\begin{array}{@{}l}
\policy\ \mathtt{sat}\ \disjoint\ \policy' \qquad\mathit{iff}\qquad 
\forall\ r \in \Req \  : \  \policySem{\policy}{\req} = \langle \dec \ \ \fo^*\rangle , \policySem{\policy'}{\req} = \langle \dec' \ \ \fo'^*\rangle, \\
\qquad\qquad\qquad\quad\qquad\quad\qquad\quad\qquad\quad\qquad\quad
 \{\: \dec, \dec' \:\}  \not\subseteq \{\permit, \deny\}
\end{array}
$$
It is worth noticing that disjoint polices can be combined with the assurance that the allowed or forbidden authorisations established by each of them are not in conflict, which simplifies the choice of the combining algorithm to be used.

\smallskip
\noindent
\textit{Coverage}.
Coverage among policies means that one of such policies establishes the same decisions as the other ones. More specifically, the property $\cover\ \policy'$ requires that for each request $\req$ for which $\policy'$ evaluates to an admissible decision, the policy under examination evaluates to the same decision. The satisfiability of the \emph{Coverage} property by a policy $\policy$ is defined as follows
$$
\begin{array}{l}
\policy\ \mathtt{sat}\ \cover\ \policy' \qquad\mathit{iff}\qquad \forall\ \req \in \Req \ : \ 
\policySem{\policy'}{\req} = \langle \dec \ \ \fo^*\rangle, 
\dec \in \{\permit, \deny\} \\
\qquad\qquad\qquad\qquad\qquad\qquad\qquad\qquad\qquad
\Rightarrow\ \ 
\policySem{\policy}{\req} = \langle \dec \ \ \fo'^*\rangle 
\end{array}
$$
Thus, $\policy$ calculates at least the same admissible decisions as $\policy'$. Consequently, if $\policy'$ also covers $\policy$, the two policies establish exactly the same admissible authorisations.

\smallskip

These structural properties permit statically reasoning on the relationships among policies and support system designers in developing and maintaining policies. One technique they enable is the \emph{change-impact analysis}~\cite{FislerKMT05}. This analysis examines policy modifications for discovering unintended consequences of such changes.

\subsubsection{Towards Automated Verification}
\label{sec:towautverif}

It is worth noticing that the analysis approach we propose is feasible in practice, although the involved sets of requests might be infinite, \eg the request extension set of a given request and the set of all possible requests. Indeed, Lemma~\ref{lemma1} implies that only the attribute names occurring within the policies of interest are relevant for their analysis, and these are finite in number; any other name cannot affect policy evaluation. For instance, to analyse a policy $\policy$, we must not consider the set $R$ of all possible requests, but only the set of those requests whose domain is $\mathit{Names}(\policy)$, \ie the finite set of attribute names occurring in $\policy$. This property paves the way for carrying out automated property verification by means of SMT solvers as described in Section~\ref{sec:prop_verify}.

\subsection{Properties on the e-Health case study}
\label{sec:prop_ex}

By way of example, we address in terms of authorisation and structural properties the case of pharmacists willing to write an e-Prescription in the e-Health case study.

Given the patient consent policies in Section~\ref{sec:facplEx}, i.e. Policies~(\ref{policy1}) and~(\ref{policyPatientConsent}), we can verify whether they disallow the access to a pharmacist that wants to write an e-Prescription. To this aim, we define an \textit{Evaluate-To} property\footnote{For the sake of presentation, in this subsection we write requests using the FACPL syntax (i.e., they are specified as sequences of attributes) rather than using their semantical functional representation.} as follows
\[
\Eval{
\begin{array}{@{}l@{}}
(\x{sub/role}, ``\x{pharmacist}")
(\x{act/id}, ``\x{write}")
 (\x{res/typ}, ``\x{e}\textrm{-}\x{Pre}")
 \end{array} }{\ \ \deny} \label{prop:1} \tag{Pr1}
\]
which requires that such request evaluates to $\deny$. Alternatively, by exploiting request extensions, we can check if there exists a request for which a pharmacist acting on e-Prescription can be evaluated to $\notApp$. This corresponds to the \textit{May-Evaluate-To} property defined as follows
\[
\May{
\begin{array}{l}
(\x{sub/role}, ``\x{pharmacist}")
 (\x{res/typ}, ``\x{e}\textrm{-}\x{Pre}")
 \end{array} 
}{\ \ \notApp} 
\label{prop:2} \tag{Pr2}
\]
The verification of these properties with respect to Policy~(\ref{policy1}) results in
$$
\NSat{Policy~(\ref{policy1})\ }{\ (\ref{prop:1})} \quad \Sat{Policy~(\ref{policy1})\ }{\ (\ref{prop:2})}
$$
where $\NSatLit$ indicates that the policy does not satisfy the property. Indeed, as already pointed out in Section~\ref{sec:facplEx}, each request assigning to $\x{act/id}$ a value different from $\x{read}$ evaluates to $\notApp$, hence property (\ref{prop:1}) is not satisfied while property (\ref{prop:2}) holds. On the contrary, the verification with respect to Policy~(\ref{policyPatientConsent}) results in
$$
\Sat{Policy~(\ref{policyPatientConsent})\ }{\ (\ref{prop:1})} \quad \NSat{Policy~(\ref{policyPatientConsent})\ }{\ (\ref{prop:2})}
$$
Both results are due to the internal policy $(\deny)$ which, together with the algorithm $\permitOverO{}$, prevents $\notApp$ to be returned and establishes $\deny$ as default decision. 

The analysis can also be conducted by relying on the structural properties. By verifying completeness, we can check if there is a request that evaluates to $\notApp$. We get
$$
\begin{array}{c}
\NSat{Policy~(\ref{policy1})\ }{\ \complete} \qquad
 \Sat{Policy~(\ref{policyPatientConsent})\ }{\ \complete} 
 \end{array}
$$
As expected, Policy~(\ref{policy1}) does not satisfy completeness, i.e. there is at least one request that evaluates to $\notApp$, whereas Policy~(\ref{policyPatientConsent}) is complete. Instead, we can check if Policy~(\ref{policyPatientConsent}) correctly refines Policy~(\ref{policy1}) by simply verifying coverage. We get
$$
\Sat{Policy~(\ref{policyPatientConsent})\ }{\ \cover\ Policy~(\ref{policy1})} 
$$
This follows from the fact that Policy~(\ref{policyPatientConsent}) evaluates to $\permit$ the same set of requests as Policy~(\ref{policy1}) and that Policy~(\ref{policy1}) never returns $\deny$; clearly, the opposite coverage property does not hold. It should be also noted that the two policies are not disjoint (as they share the set of permitted requests).

\subsection{Expressing Constraints with SMT-LIB}
\label{sec:contr_eval}

Property verification requires extensive checks on large (possibly infinite) amounts of requests, hence, in order to be practically effective, tool support is essential. To this aim, we express the constraints defined in Section~\ref{sec:constraint} by means of the SMT-LIB language (\url{http://smtlib.cs.uiowa.edu/}), that is a standardised constraint language accepted by most of the SMT solvers. Intuitively, SMT-LIB is a strongly typed functional language expressly defined for the specification of constraints. Of course, the feasibility of the SMT-based reasoning crucially depends on decidability of the satisfiability checks to be done; in other words, the used SMT-LIB constructs must refer to decidable theories, as \eg uninterpreted function and array theories. We now provide a few insights on the SMT-LIB coding of our constraints. 

The key element of the coding strategy is the parametrised record type representing attributes. This type, named \Fval, is defined as follows
\begin{verbatim}
(declare-datatypes (T) ((TValue 
    (mk-val (val T)(miss Bool)(err Bool)))))
\end{verbatim}
Hence, each attribute consists of a 3-valued record, whose first field \FvalField\ is the value with parametric type \texttt{T} assigned to the attribute, while the boolean fields \FvalBot\ and \FvalErr\ indicate, respectively, if the attribute value is missing or has an unexpected type. Additional assertions, not shown here for the sake of presentation, ensure that the fields \FvalBot\ and \FvalErr\ cannot be true at the same time, and that, when one of the last two fields is true, it takes precedence over \FvalField. Of course, a specification formed by multiple assertions is satisfied when all the assertions are satisfied.

\begin{table*}[t]
\caption{Type inference rules for (an excerpt of) \facpl\ expressions; we use $\typeVar$ as a type variable, $\nvType$ as a type name or a type variable, and we assume that $\boolT$, $\doubleT$, $\stringT$, $\dateT$, $\setT$ identify both the values' domains and their type names}
\label{tab:typeInfer}
\hspace*{-.6cm}
\centering
$
\footnotesize
\begin{array}{c}
\hline
\mbox{}\\[-.2cm]
\infer[]{\Gamma \vdash \val: \boolT \mid \true}{\val \in \mathit{Bool}}
\quad\
\infer[]{\Gamma \vdash \val: \doubleT  \mid \true}{\val \in \mathit{Double}}
\quad\
\infer[]{\Gamma \vdash \val: \stringT  \mid \true}{\val \in \mathit{String}}
\quad\
\infer[]{\Gamma \vdash \val: \dateT  \mid \true}{\val \in \mathit{Date}}
\quad\
\infer[]{\Gamma \vdash \val: \setT  \mid \true}{\val \in 2^{\mathit{Value}}}
\quad\
\infer[]{\Gamma \vdash \name : \typeVar \mid \true}{\Gamma(\name) = \typeVar}
\\[.3cm]
\infer[]{\Gamma \vdash \x{not}({\expr}) : \boolT\mid C \wedge \nvType = \boolT}{\Gamma \vdash  \expr : \nvType \mid C}
\qquad
\infer[\exprOp \in \{\x{and},\x{or}\}]{\Gamma \vdash \exprOp(\expr_1, \expr_2) : \boolT \mid C_1 \wedge C_2 \wedge \nvType_1 = \boolT \wedge \nvType_2 = \boolT }
{\Gamma \vdash  \expr_1 : \nvType_1 \mid C_1 \qquad \Gamma \vdash  \expr_2 : \nvType_2 \mid C_2}
\\[.3cm]
\infer[]{\Gamma \vdash \x{equal}(\expr_1, \expr_2) : \boolT \mid C_1 \wedge C_2 \wedge \nvType_1 = \nvType_2 }
{\Gamma\vdash  \expr_1 : \nvType_1 \mid C_1 \qquad \Gamma\vdash  \expr_2 : \nvType_2 \mid C_2}
\qquad
\infer[]{\Gamma \vdash \x{in}(\expr_1, \expr_2) : \boolT \mid C_1 \wedge C_2 \wedge \nvType_1 = \nvType_2}
{\Gamma \vdash  \expr_1 : \nvType_1 \mid C_1 \qquad \Gamma \vdash  \expr_2 : 2^{\nvType_2} \mid C_2}
\\[.3cm]
\hline
\end{array}
$
\end{table*}

The declaration of \Fval\ outlines the syntax of SMT-LIB and its strongly typed nature. This means that each attribute occurring in a policy has to be typed, by properly instantiating the type parameter \texttt{T}. Since \facpl\ is an untyped language, to reconstruct the type of each attribute, we define the type inference system (whose excerpt is) reported in Table~\ref{tab:typeInfer}. The rules are straightforward and infer the judgment $\Gamma \vdash \expr: \nvType \mid C$ which, under the typing context $\Gamma$, assigns the type (or the type variable) $\nvType$ to the \facpl\ expression $\expr$ and generates the typing constraint $C$. Specifically, $\Gamma$ is an injective function that associates a type variable to each attribute name, while $C$ is basically made of conjunctions and disjunctions of equalities between variables and types. The generated typing constraint will be processed at the end of the inference process to establish well-typedness of an expression. Thus, a \facpl\ expression is \emph{well-typed} if $C$ is satisfiable, \ie there exists a type assignment for the typing variables occurring in $C$ that satisfies $C$. Moreover, a \facpl\ policy is \emph{well-typed} if the typing constraints generated by all the expressions occurring in the policy are satisfied by a same assignment. These type assignments are then used to instantiate the type parameters of the SMT-LIB constraints representing well-typed policies.

The type inference system aims at statically getting rid of all those policies containing expressions that are not well-typed. For instance, given the expression $\x{or}(\x{cat/id}, \x{equal}(\x{cat/id}, 5))$ and the typing context $\Gamma(\x{cat/id})$ = $\typeVar_{\x{cat/id}}$, the inference rules assign the type $\boolT$ to the expression and generate the constraint $\typeVar_{\x{cat/id}} = \doubleT \wedge \typeVar_{\x{cat/id}} = \boolT \wedge \boolT = \boolT$. This constraint is clearly unsatisfiable (as attribute $\x{cat/id}$ cannot simultaneously be a double and a boolean), hence a policy containing such expression is not well-typed and would be statically discarded. Notably, the use of the field \FvalErr\ allows the analysis to  address the role of errors in policy evaluations, \ie to reason on the authorisations of requests assigning values of unexpected type to attribute names. It is indeed crucial to analyse also these requests, since a possible attacker can leverage on them to circumvent the access control system.

On top of the \Fval\ datatype we build the uninterpreted functions expressing the operators of the proposed constraint formalism. By way of example, the operator $\fand$ corresponds to the \texttt{FAnd} function defined as follows
\begin{verbatim}
(define-fun FAnd 
   ((x (TValue Bool)) (y (TValue Bool)))
   (TValue Bool)
   (ite (and (isTrue x) (isTrue y))
      (mk-val true false false)
      (ite (or (isFalse x) (isFalse y))
          (mk-val false false false)
          (ite (or (err x) (err y))
               (mk-val false false true)
               (mk-val false true  false)))))
\end{verbatim}
where \texttt{mk-val} is the constructor of \Fval\ records. Hence, the function takes as input two \Fval\ \texttt{Bool} records (i.e.,  type \texttt{Bool} is the instantiation of the type parameter \texttt{T}) and returns a \texttt{Bool} record as well. The conditional if-then-else assertions \texttt{ite} are nested to form a structure that mimics the semantic conditions of Table~\ref{tab:constr_sem}, so that different \Fval\ records are returned according to the input. The function \texttt{isFalse} (resp. \texttt{isTrue}) is used to compactly check that all fields of the record are $\false$ (resp. only the field \FvalField\ is $\true$). All the other constraint operators, except $\in$, are defined similarly. 

To express the operator $\in$, we need to represent multivalued attributes. Firstly, we define an array datatype, named \texttt{Set}, to model sets of elements as follows
\begin{verbatim}
(define-sort Set (T) (Array Int T)) 
\end{verbatim}
where the type parameter \texttt{T} is the type of the elements of the array. By definition of array, each element has an associated integer index that is used to access the corresponding value. Thus, a multivalued attribute is represented by a \Fval\ record with type an instantiated \texttt{Set}, e.g. \texttt{(TValue (Set Int))} is an attribute whose value is a set of integers.  Consequently, we can build the uninterpreted function modelling the constraint operator $\in$. In case of integer sets, we have
\begin{verbatim}
(define-fun inInt 
  ((x (TValue  Int)) (y (TValue (Set Int))))
  (TValue Bool)
  (ite (or (err x)(err y)) 
    (mk-val false false true)
    (ite (or (miss x) (miss y))
      (mk-val false true false)
      (ite (exists ((i Int)) 
             (=  (val x) (select (val y) i)))
           (mk-val true  false false)
           (mk-val false false false)))))
\end{verbatim}
where the command \texttt{(select (val y) i)} takes the value in position \texttt{i} of the set in the field \FvalField\ of the argument \texttt{y}. In addition to the conditional assertions, the function uses the existential quantifier \texttt{exists} for checking if the value of the argument \texttt{x} is contained in the set of the argument \texttt{y}. 

The coding approach we pursue generates, in most of the cases, fully decidable constraints. In fact, since we support non-linear arithmetic, i.e. multiplication, it is possible to define constraints for which a constraint solver is not able to answer. Anyway, modern constraint solvers are actually able to resolve nontrivial nonlinear problems that, for what concerns access control policies, should prevent any undefined evaluation\footnote{It should be noted that if at least one argument of each occurrence of the multiply operator is a numeric constant, then the resulting non-linear arithmetic constraints are decidable.}. Similarly, the quantifier-based constraints are in general not decidable, but solvers still succeed in evaluating complicated quantification assertions due to, e.g., powerful pattern techniques (see, e.g., the documentation of Z3). Notice anyway that if we assume that each expression operator $\x{in}$ (and, consequently, constraint operator $\in$) is applied to at most one attribute name, the quantifications are bounded by the number of literals defining the other operator argument.

Concerning the value types we support, SMT-LIB does not provide a primitive type for $\dateT$. Hence, we use integers to represent its elements. Furthermore, even though SMT-LIB supports the $\stringT$ type, the Z3 solver we use does not. Thus, given a policy as an input, we define an additional datatype, say \texttt{Str}, with as many constants as the string values occurring in the policy. The string equality function is then defined over \Fval\ records instantiated with type \texttt{Str}.

By way of example, the SMT-LIB code for the constraint $\const_{trg1}$ (see Section~\ref{sec:constr_ehealth}) is 
\begin{verbatim}
(define-fun cns_target_Rule1 
  () 
  (TValue Bool)
  (FAnd 
    (equalStr n_sub/role cst_doc) 
    (FAnd (equalStr n_act/id cst_write)  
      (FAnd 
        (inStr cst_permWrite n_sub/perm)
        (inStr cst_permRead n_sub/perm)))))
\end{verbatim}
where identifiers starting with \texttt{n\_} (resp. \texttt{cst\_}) represent attribute names (resp. literals) of the represented expression. The whole SMT-LIB code for Policy~(\ref{policy1}) can be found at 
\url{http://facpl.sf.net/eHealth/}.


\subsection{Automated Properties Verification}
\label{sec:prop_verify}

The SMT-LIB coding permits using SMT solvers to automatically verify the properties formalised in Section~\ref{sec:prop}. In the following, given a \facpl\ policy $\policy$, we denote by $\langle \permit: \smtconst_p\ \deny: \smtconst_d\ \notApp: \smtconst_n\ \indet: \smtconst_i\rangle$ the tuple of SMT-LIB codes representing the formal constraints $\translPol{\policy} = \langle \permit: \const_p \ \deny: \const_d\ \notApp: \const_n\ \indet: \const_i\rangle$. Hereafter, we present first the strategies to follow for verifying the authorisation properties, then those for verifying the structural properties. 

\subsubsection{Authorisation Properties}

The verification of authorisation properties requires: (i) to introduce into the policy constraint of interest, which is chosen according to the property, the SMT-LIB coding of the request defined by the property; (ii) to check the satisfiability (or validity) of the resulting constraint.

Given a request $\req$, the SMT-LIB coding of the request is defined as follows
$$
\begin{array}{l}
\req_{smtlib} \define \\[1pt]
\quad
    \left\{ 
    \begin{array}{l|l@{}}
        \begin{array}{@{}l@{}}
            \mathtt{(assert\ (=\ (val\ } \ \name  \mathtt{)\ \ } \val )) \\
            \mathtt{(assert\ (and\ (not\ (miss\ } \ \name ))\ \\
            \qquad\qquad\qquad\, \, 
            \mathtt{(not\ (err\  \ } \name ))))
        \end{array}
    & 
     \begin{array}{l}
     	\req(n) = v 
     \end{array} 
    \end{array}
    \right\}
\end{array}
$$
Thus, all attribute names $\name$ in $\req$ are asserted to be equal to their value $\val$ and to be neither missing nor erroneous. Furthermore, given a \facpl\ policy $\policy$, we also define the following SMT-LIB coding of the request
$$
\begin{array}{l}
\overline{\req_{smtlib}(p)} \define \\[1pt]
\quad
    \left\{ 
    \begin{array}{l|l}
            \mathtt{(assert\ (miss\ }\ \name )) 
    & 
		\name \in \mathit{Names}(\policy) \ \wedge\  \req(n) = \excpt
    \end{array}
    \right\}
\end{array}
$$
where, as in Section~\ref{sec:prop_sem}, $\mathit{Names}(\policy)$ indicates the set of attribute names occurring in $\policy$. Thus, all the names $\name$ that occur in $\policy$ and are not assigned to a value in $\req$ are asserted to be missing attributes.

By exploiting this SMT-LIB coding of requests, we define the automated verification (i.e., via an SMT solver) of the authorisation properties as follows
$$
\begin{array}{@{}l@{}}
\Sat{\policy} \Eval{\req}{\dec}\ \  \mathtt{iff} \\
\qquad\quad
\smtconst_{\dec} \ \circ\ \req_{smtlib} \ \circ\ \overline{\req_{smtlib}(p)} \quad\ \, \mathtt{is} \ \mathtt{sat}  \\
\Sat{\policy} \May{\req}{\dec} \ \ \, \mathtt{iff} \\
\hfill 
\smtconst_{\dec} \ \circ\ \req_{smtlib} \quad  \mathtt{is} \ \mathtt{sat} \\
\Sat{\policy} \Must{\req}{\dec} \   \mathtt{iff} \\
 \hfill \smtconst_{\dec} \ \circ\ \req_{smtlib} \quad  \mathtt{is} \ \mathtt{valid} \\
\end{array}
$$
where $\circ$ indicates the concatenation of SMT-LIB code\footnote{Notably, checking the satisfiability of the SMT-LIB code resulting from the concatenation of two (sets of) SMT-LIB assertions amounts to check if both the assertions hold at the same time.}\ and $\mathtt{valid}$ means that the corresponding SMT-LIB code is a valid set of assertions. Some comments follow.

The \textit{Evaluate-To} property does not exploit request extensions, hence all attribute names not assigned by the considered request can only assume the special value $\excpt$. This means that the request $\req$ is coded in SMT-LIB with $\req_{smtlib}$ and $\overline{\req_{smtlib}(p)}$. The satisfiability of the property thus corresponds to that of the resulting SMT-LIB code.

To verify the \textit{May-Evaluate-To} property, since it considers request extensions, the request has to be coded only with $\req_{smtlib}$. As before, the satisfiability of the property  corresponds to that of the resulting SMT-LIB code.

Finally, to verify the \textit{Must-Evaluate-To} property, we code again the request with $\req_{smtlib}$ only, but we check the validity of the resulting SMT-LIB code, i.e. that it is satisfied by all the assignments for the attribute names. This amounts to check if the negation of the resulting SMT-LIB code is not satisfiable, in which case the property holds.

\subsubsection{Structural Properties}
The verification of structural properties does not require to modify policy constraints, but rather to check the unsatisfiability of combinations of constraints. It is defined as follows
$$
\begin{array}{@{}l@{}}
\policy\ \mathtt{sat}\ \complete\qquad\,  \mathtt{iff} \quad \ \ \,  \smtconst_{n}  \quad\, \, \mathtt{is} \ \mathtt{unsat}  \\[.1cm]
\policy\ \mathtt{sat}\ \disjoint\ \policy'  \quad\! \mathtt{iff}
\\ 
    \qquad\quad 
    \left\{ 
    \begin{array}{l@{\qquad\qquad }l}
    \smtconst_{p} \ \circ\ \smtconst'_{p} &  \mathtt{is} \ \mathtt{unsat} \\
    \smtconst_{p} \ \circ\ \smtconst'_{d} &  \mathtt{is} \ \mathtt{unsat} \\
    \smtconst_{d} \ \circ\ \smtconst'_{p} &  \mathtt{is} \ \mathtt{unsat} \\
    \smtconst_{d} \ \circ\ \smtconst'_{d} &  \mathtt{is} \ \mathtt{unsat} \\
    \end{array}
    \right.
\\
\mbox{}\\[-.2cm]
\policy\ \mathtt{sat}\ \cover\ \policy'  \qquad\ \ \, \mathtt{iff}
\\
\qquad\quad
    \left\{ \:
    \begin{array}{l@{\qquad\quad\!}l}
    \lnot\ \smtconst_{p} \ \circ\ \smtconst'_{p} &  \mathtt{is} \ \mathtt{unsat} \\
    \lnot\ \smtconst_{d} \ \circ\ \smtconst'_{d} &  \mathtt{is} \ \mathtt{unsat} \\
    \end{array}
    \right.
\\
\end{array}
$$
where $\smtconst'_{\dec}$ refers to the SMT-LIB code modelling decision $\dec$ of policy $\policy'$. Some comments follow.

The trivial case is that of the \emph{completeness} property, which only amounts to check if the constraint modelling the decision $\notApp$ is not satisfiable, \ie if its neg\-ation is valid. If it is, the property holds. 

The \emph{disjointness} of two policies is verified by checking, one at a time, if the conjunctions between the $\permit$ or $\deny$ constraint of the first policy and the $\permit$ or $\deny$ constraint of the second policy are not satisfiable.
If this holds for the four possible combinations of those constraints, the property holds. 

The \emph{coverage} of policy $\policy$ on policy $\policy'$ is verified by checking if the conjunction between the negation of the $\permit$ (resp., $\deny$) constraint of $\policy$ and the $\permit$ (resp., $\deny$) constraint of $\policy'$ is not satisfiable. Intuitively, if the policy $\policy$ does not calculate a $\permit$ or $\deny$ decision (i.e., $\lnot\ \smtconst_{p}$ and $\lnot\ \smtconst_{d}$ hold), policy $\policy'$ cannot do it as well, otherwise the property is not satisfied. If this holds for the two conjunctions separately, the property holds.


\section{The FACPL Toolchain}
\label{sec:tool}

\begin{figure*}[t]
\centering
\includegraphics[scale=.38]{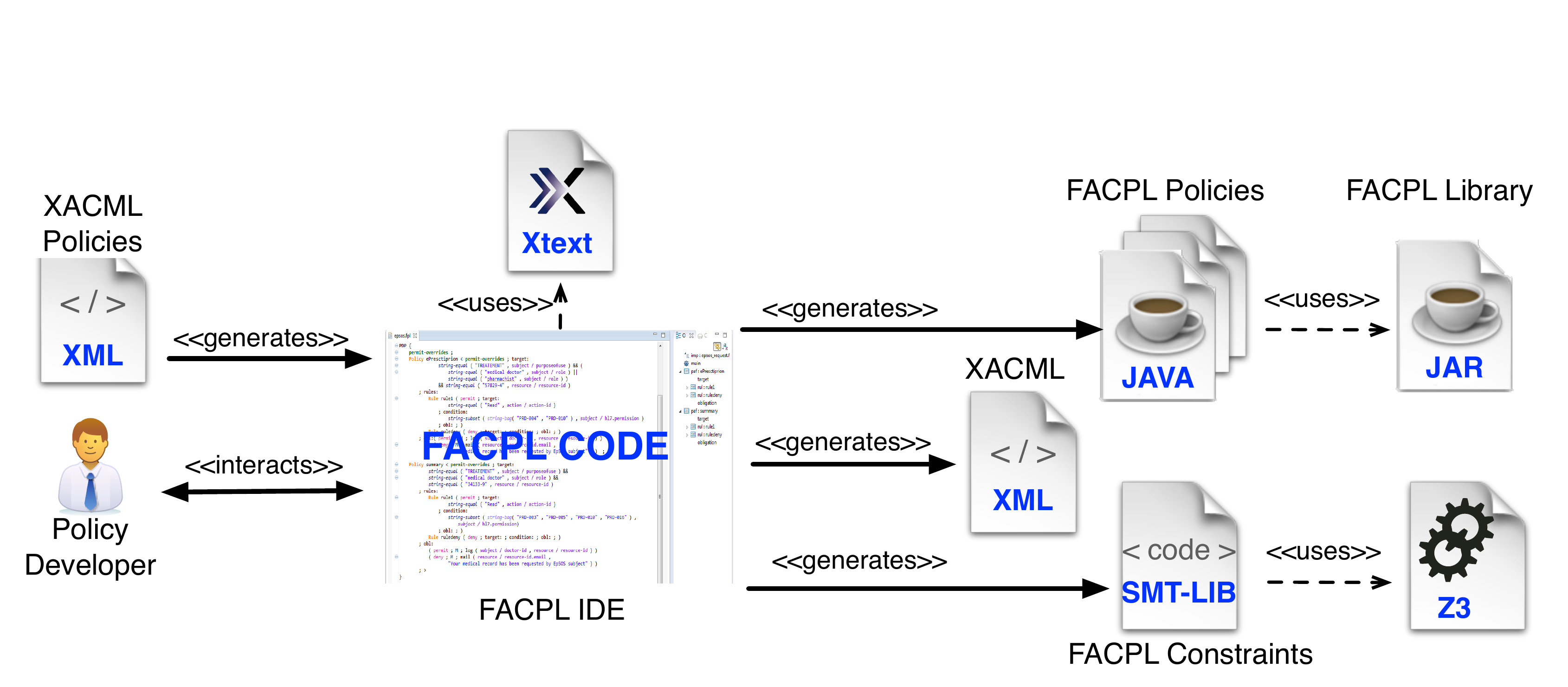}
\vspace*{-.4cm}
\caption{The \facpl\ toolchain}
\label{fig:tools}
\end{figure*}

The coding, analysis and enforcement tasks pursued in the development of \facpl\ specifications are fully supported by a Java-based software toolchain\footnote{The \facpl\ supporting tools are freely available and open-source; binary files, source files, unit tests and documentation can be found at the \facpl\ website \url{http://facpl.sf.net}.}, graphically depicted in Figure~\ref{fig:tools}. The key element of the toolchain is an Eclipse-based IDE that provides features like, e.g., static code checks and automatic generation of runnable Java and SMT-LIB code. An expressly developed Java library is used to compile and execute the Java code, while the analysis of SMT-LIB code exploits the Z3 solver.

To provide interoperability with the 
standard \xacml\ and the various available tools supporting it (e.g., XCREATE~\cite{XCREATE}, Margrave~\cite{FislerKMT05} and Balana~\cite{Balana}), the IDE automatically translates \facpl\ code into \xacml\ one and vice-versa. Because of slightly different expressivity, there are some limitations in \facpl\ and \xacml\ interoperability (see Section~\ref{sec:facplvsxacml} for further details).

Furthermore, to allow newcomer users to directly experiment with \facpl, the web application ``Try FACPL in your Browser'' (reachable from the \facpl\ website) offers an online editor for creating and evaluating \facpl\ policies; the e-Health case study is there reported as a running example. Additionally, the web interface reachable from \url{http://facpl.sf.net/eHealth/demo.html} shows a proof-of-concept demo on how a \facpl -based access control system can be exploited for providing e-Health services. 

In the rest of this section, we detail the \facpl\ Java library and IDE, while Section~\ref{sec:performance} reports performance and functionality comparisons with other similar tools.

\subsection{The \facpl\ Library} 

The Java library we provide aims at representing and evaluating \facpl\ policies, hence at fully implementing the evaluation process formalised in Section~\ref{sec:formal_sem}. To this aim, driven by the formal semantics, we have defined a conformance test-suite that systematically verifies each library unit (e.g., expressions and combining algorithms) with respect to its formal specification.

For each element of the language the library contains an abstract class that provides its evaluation method. In practice, a \facpl\ policy is translated into a Java class that instantiates the corresponding abstract one and adds, by means of specific methods (e.g., $\mathtt{addObligation}$), its forming elements. Similarly, a request corresponds to a Java class containing the request attributes and a reference to a context handler that can be used to dynamically retrieve additional attributes at evaluation-time.

Evaluating requests amounts to invoke the evaluation method of a policy, which coordinates the evaluation of its enclosed elements in compliance with its formal specification. In addition to the authorisation process, the library supports the enforcement process by defining the three enforcement algorithms and a minimal set of pre-defined \pep\ actions, i.e. $\mathtt{log}$, $\mathtt{mailTo}$ and $\mathtt{compress}$. Additional  actions can be dynamically introduced by providing their implementation classes to the \pep\ initialisation method.

By way of example, we report in the following listings (an excerpt of) the Java code of Policy~(\ref{policy1}) introduced in Section~\ref{sec:facplEx}. Besides the specific methods used for adding policy elements, the Java code highlights the use of class references for selecting expression operators and combining algorithms. This design choice, together with the use of best-practices of object-oriented programming, allows the library to be easily extended with, e.g., new expression operators, combining algorithms and enforcement actions. Notice that rules are private inner classes, because they cannot be referred outside the enclosing policy sets. 
\medskip

\ALGOInline{An excerpt of the generate code of the e-Prescription policy}{}{ePre2_short.txt}

Besides the four-valued decisions considered so far, the \facpl\ library also supports the extended indeterminate values used by \xacml, \ie $\indetP$, $\indetD$ and $\indetDP$. They can be used to specify the potential decision ($\permit$, $\deny$ and both, respectively) that should have been returned by the evaluation of a policy if an error would not have occurred. Extended indeterminate values allow the \pdp\ to obtain additional information about policy evaluation, which can be exploited, e.g., during policy debugging for improving the treatment of errors. However, their usage may require additional workload. In fact, it establishes that if the target of a policy set evaluates to $\err$, rather than stopping and returning $\indet$, the evaluation process continues the computation by processing the enclosed policies and using the decision resulting from the application of the combining algorithm to calculate an extended indeterminate value. Thus, e.g., if the combining algorithm returns the decision $\permit$, the evaluation of the policy returns $\indetP$. For all these reasons, we have chosen to support the extended indeterminate values by means of a boolean parameter (of the method $\mathtt{doAuthorisation}$) whose setting can enable or disable their use at each \pdp\ invocation.

\subsection{The FACPL IDE}

\begin{figure*}[t]
\centering
\includegraphics[scale=.18]{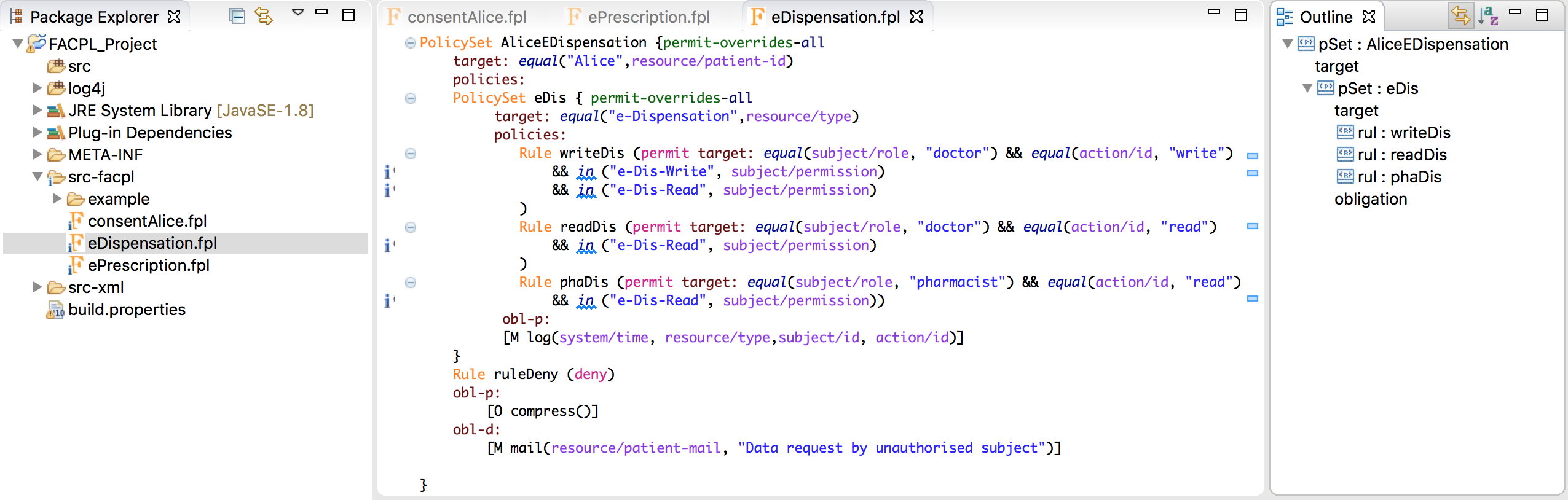}
\vspace*{-.4cm}
\caption{The \facpl\ IDE}
\label{fig:IDE}
\end{figure*}

The \facpl\ IDE (see the screenshot in Figure~\ref{fig:IDE}) is developed as an Eclipse plug-in and aims at bringing together the available functionalities and tools. Indeed, it fully supports writing, evaluating and analysing \facpl\ specifications. The plug-in has been implemented by means of Xtext (\url{http://www.eclipse.org/Xtext/}), that is a framework to design and deploy domain-specific languages.

The plug-in accepts an enriched version of the \facpl\ language, which contains high level features facilitating the coding tasks. In particular, each policy has an identifier that can be used as a reference to include the policy within other policies, while specific linguistic handles enable the definition of new expression operators and combining algorithms. In order to ease the organisation of large policy specifications, the plug-in supports modularisation of files and import commands extending file scopes.

The development environment provided by the plug-in is standard. It offers graphical features (e.g., keywords highlighting, code suggestion and navigation within and among files), static controls on \facpl\ code (e.g., uniqueness of identifiers and type checking), and automatic generation of Java, \xacml, and SMT-LIB code. To configure all the required libraries, a dedicated wizard creates a  \facpl -type project. 

To facilitate the analysis of \facpl\ policies, the plug-in also provides a simple interface allowing policy developers to specify the authorisation and structural properties to be verified on a certain policy. Thus, the plug-in automatically generates the corresponding SMT-LIB files according to the strategies reported in Section~\ref{sec:prop_verify}; an execution script for the Z3 solver is also generated. Of course, the SMT-LIB files can be also evaluated by any other solver accepting SMT-LIB and supporting the theories we use.

As previously pointed out, the Java library is flexible enough to be easily extended. The plug-in facilitates this task by means of dedicated commands. For instance, to define a new expression operator, once a developer has defined the signature of the new function (which is used for type checking and inference), a template of its Java and SMT-LIB implementation is automatically generated. The actual implementation of the Java class, as well as that of the SMT-LIB function, is left to the developer.

In the \facpl\ IDE, all the mentioned functionalities are offered via the customised \facpl\ project created by the dedicated wizard. Once the project is created, the policy developer can write code starting from the basic \facpl\ and \xacml\ examples already provided or from scratch. A \facpl\ file is a generic text file with extension \emph{.fpl},  which has dedicated text editor, outline view and contextual menus. Functionalities supporting code development, \eg code suggestion and auto-completion, are available via the usual Eclipse shortcuts and menus. In particular, from either the toolbar menu or the right-click editor menu, the developer can find a set of pre-defined commands to generate Java, \xacml\ and SMT-LIB code, or to open a step-by-step wizard for the definition of authorisation and structural properties.


\section{Related Work}
\label{sec:relatedWork}

A preliminary version of \facpl\ was introduced in~\cite{ESSOS} with the aim of formalising the semantics of \xacml. The language presented here addresses a wider range of aspects concerning access control. Specifically, the syntax of the language is cleaned up and streamlined (e.g., rule conditions are integrated with rule targets and the policy structure is simplified); at the same time, it is extended with additional combining algorithms, the \pep\ specification, an explicit syntax for expressions, and obligations. This latter extension widens \facpl\ applicability range and expressiveness, as it provides the policy evaluation process with further, powerful means to affect the behaviour of controlled systems (see \eg\ \cite{WS-FM} for a practical example of a policy-based manager for a Cloud platform). Additional important differences concern the definition of the policy semantics: in \cite{ESSOS} it is given in terms of partitions of the set of all possible requests, while here it is defined in a functional fashion with respect to a generic request. The new approach also features the formalisation of combining algorithms in terms of binary operators and instantiation strategies, and the automatic management of missing attributes and evaluation errors throughout the evaluation process. Most of all, the aim of this work is significantly different: we do not only propose a different language, but we provide a complete methodology that encompasses all phases of policy lifecycle, i.e. specification, analysis and enforcement. Concerning the analysis, we define a set of relevant authorisation and structural properties (whose preliminary definition is given in \cite{WWV15}) characterised in terms of sets of requests. We then introduce a constraint-based representation of policies and an SMT-based approach for mechanically verifying properties on top of constraints. To effectively support the functionalities, we provide a fully-integrated software toolchain.


In the rest of this section we survey more closely related work. First, we comment on differences and interoperability of \facpl\ with the already mentioned standard \xacml\ (Section~\ref{sec:facplvsxacml}). Then, we discuss other relevant policy languages (Section~\ref{sec:pol_lang}), and approaches to the analysis of (access control) policies (Section~\ref{sec:rel_analysis}). Finally, we compare supporting tools (Section~\ref{sec:performance}). 

\subsection{FACPL vs XACML}
\label{sec:facplvsxacml}
  
XACML~\cite{XACML3} is a well-established standard for the specification of attribute-based access control policies and requests. It has an XML-based syntax and an evaluation process defined in accordance with~\cite{rfc2753} (hence similar to the \facpl\ one). As a matter of notation, hereafter the words emphasised in sans-serif, e.g. \textsf{Rule}, are XML elements, while element attributes are in italics.


From a merely lexical point of view, \facpl\ allows developers to define each policy element via a lightweight mnemonic syntax and leads to  compact policy specifications. Instead, the XML-based syntax used by \xacml\ ensures cross-platform interoperability, but generates verbose specifications that are hardly of immediate comprehension for developers and are not suitable for formally defining semantics and analysis techniques. Table~\ref{tab:comparison} exemplifies a lexical comparison between the \facpl\ policies for the e-Health case study and the corresponding \xacml\ ones (both groups of policies can be downloaded from 
\url{http://facpl.sf.net/eHealth/}). 

Although \facpl\ and \xacml\ policies have a similar structure, there are quite a number of (semantic) differences. Hereafter we outline the main ones.

\begin{table*}[!t]
\caption{\facpl\ vs. \xacml\ on the e-Health case study} \label{tab:comparison}
\centering
\begin{tabular}{|@{\ }l@{\ }|@{\ }c@{\ }|@{\ }c@{\ }|@{\ }c@{\ }|@{\ }c@{\ }|@{\ }c@{\ }|@{\ }c@{\ }|}
\hline
\multicolumn{1}{|@{\ }c@{\ }|@{\ }}{\textbf{Policy}} & \multicolumn{2}{@{}c@{\ }|@{\ }}{\textbf{Number of lines}} & \textbf{Saved} & \multicolumn{2}{@{}c@{\ }|@{\ }}{\textbf{Number of types}} & \textbf{Saved}\\[.1cm]
 & \xacml & \facpl & \textbf{lines} & \xacml & \facpl & \textbf{types}\\
\hline
e-Prescription  &  239  &  24 & 89,95\% & 10.656 &  894 &   91,61\%\\
\hline
e-Dispensation   &  239 & 24 &  89,95\% & 10.674 & 914  & 91,43\%  \\
\hline
Patient Consent   & 423 & 38  & 91,01\% & 19.195 & 1.558  & 91,88\% \\
\hline
\end{tabular}
\end{table*}

In \facpl, request attributes are referred by structured names. In \xacml, they are referred by either \textsf{AttributeDesignator} or \textsf{AttributeSelector} elements. The former one corresponds to a typed version of a structured name, while the latter one is defined in terms of XPath expressions, which are not supported by \facpl. Anyway, \facpl\ can represent some of them by appropriately using structured names; \eg an \textsf{AttributeSelector} with category $\mathit{subject}$ and an XPath expression like $\mathit{type/id/text()}$ correspond to $\x{subject/type.id}$.

A \xacml\ \textsf{Target} is made of \textsf{Match} elements defining basic comparison functions on request attributes. The elements are then organised in terms of the tag structure \textsf{AnyOf}\textrm{-}\textsf{AllOf}\textrm{-}\textsf{Match}. This structure can be rendered in \facpl\ by means of, respectively, the expression operators $\x{and}\textrm{-}\x{or}\textrm{-}\x{and}$. However, slightly different results can be obtained from target evaluations due to the management of errors and missing attributes. Indeed, when a value is missing, the \xacml\ semantics of \textsf{Match} elements returns $\err$ or $\false$ (according to the setting of the boolean parameter \textsf{MustBePresent}), whereas the \facpl\ semantics of the target elements (depending on the expression operator) can return $\excpt$ possibly until the level of policies is reached, where $\excpt$ is converted to $\false$; the same occurs for evaluation errors. From our point of view, the \facpl\ management of missing attributes and evaluation errors is smoother than the \xacml's one. Indeed \facpl\ can distinguish if a boolean target function has returned $\false$ due to a not satisfied comparison or due to a missing attribute; instead, \xacml\ cannot always do it.

Additionally, the evaluation of \textsf{Match} functions in \xacml\ is iteratively defined on all the retrieved attribute values. To ensure a similar behaviour in \facpl, a \xacml\ expression such as, e.g., an equality comparison must be translated into an operator defined on sets, like \eg $\x{in}$. Clearly, this limits the amount of \xacml\ functions that can be faithfully represented in \facpl. 
 
Except for target evaluation, the semantics of \xacml\ and \facpl\ policies mainly comply with each other. However, the specification approach fostered by \facpl\ is more generic and poses less constraints on the policy structure.  In particular, \xacml\ prescribes a policy structure based on \textsf{Policy}s, i.e. collections of \textsf{Rules}, and \textsf{PolicySet}s, i.e. collections of \textsf{Policy}s and/or \textsf{PolicySet}s. Most of all, \xacml\ forces specific constraints on targets of \textsf{Policy} and \textsf{Policy Set}: they can only contain comparison functions and each comparison can only contain one attribute name. Moreover, \xacml\ supports fewer combining algorithms than \facpl, as well as instantiation strategies (indeed, \xacml\ only provides the $\greedy$ one). Additionally, as previously pointed out, \xacml\ specialises the decision $\indet$ into three sub-decisions: for the sake of presentation, we have not considered them in the formal development of \facpl, but they are fully supported by the \facpl\ library (see Section~\ref{sec:tool}). 

\facpl\ and \xacml\ share the same management of obligations, although in \facpl\ this process is specified in a more precise manner, \ie through the instantiation strategies and the binary combining operators. It is also worth noticing that the \xacml\ `obligation fulfilment' is termed `obligation instantiation' in \facpl, since indeed the evaluation of obligations by the PDP does not carry out any task beyond its mere instantiation.

Also, \xacml\ provides some constructs that do not crucially affect policy expressiveness and evaluation. For instance, \textsf{Variable} elements permit defining pointers to expression declarations. These constructs are not directly supported by \facpl.  

To sum up, except for minor differences on tangled \xacml\ aspects mainly concerning the management of missing attributes and evaluation errors, \facpl\ subsumes \xacml\ policies not containing XML raw data by offering, at the same time, a higher flexibility in the policy specification approach and a richer set of combining algorithms. Most differently from \xacml, \facpl\ provides a formal semantics that supports formally-based analysis techniques.

\begin{table*}[!t]
\caption{Comparison of some relevant policy languages ($\checkmark^\ast$ means that user encoding is required)}\label{tab:PL_comparison}
\centering
\begin{tabular}{|@{\ }c@{\ }|@{\ }c@{\ }|@{\ }c@{\ }|@{\ }c@{\ }|@{\ }c@{\ }|@{\ }c@{\ }|@{\ }c@{\ }|@{\ }c@{\ }|}
\hline
\textbf{Features} & \textbf{\xacml} & \textbf{Ponder} & \textbf{ASL} & \textbf{PTaCL} & \cite{RaoLBLL09} & \cite{ArkoudasCC14} & \textbf{\facpl}\\
\hline
Rule-based & $\checkmark$ & $\checkmark$ & & &  & & $\checkmark$\\
\hline
Logic-based & &  & $\checkmark$ & $\checkmark$ & $\checkmark$ &  $\checkmark$&   \\
\hline
Mnemonic spec. & & \checkmark & & & & & \checkmark\\
\hline
Comb. algorithms & \checkmark & & $\checkmark^\ast$ & $\checkmark^\ast$ & \checkmark &  $\checkmark^\ast$  &\checkmark\\
\hline
Obligations & \checkmark &  \checkmark & & & & & \checkmark \\
\hline
Missing attributes & \checkmark & &  & \checkmark &  & & \checkmark\\
\hline
Error handling &  \checkmark & & &  & & & \checkmark \\
\hline
\end{tabular}
\end{table*}

\subsection{Policy Languages for Access Control}
\label{sec:pol_lang}

Policy languages have recently been the subject of extensive research, both by industry and academia. Indeed, policies permit managing different important aspects of system behaviours, ranging from access control to adaptation and emergency handling. We compare in the following the main policy languages devoted to access control, which is our focus; Table~\ref{tab:PL_comparison} summarises the comparison.

Among the many proposed policy languages, we can identify two main specification approaches: \emph{rule-based}, as \eg the \xacml\ standard and Ponder~\cite{DamianouDLS01,TwidleDLS09}, and \emph{logic-based}, as e.g. ASL~\cite{Jajodia97alogical}, PTaCL~\cite{CramptonM12} and the logical frameworks in~\cite{ArkoudasCC14}. Many other works, as \eg~\cite{LiWQBRLL09,RaoLBLL09,RamliNN14}, study (part of) \xacml\ by formally addressing peculiar features of design and evaluation of access control policies. 

In the rule-based approach, policies are structured into sets of declarative rules. The seminal work~\cite{sloman} introduces two types of policies: authorisations and obligations. Policies of the former type have the aim of establishing if an access can be performed, while those of the latter type are basically Event-Condition-Action rules triggering the enforcement of adaptation actions. This setting is at the basis of Ponder.

Ponder is a strongly-typed policy language that, differently from \facpl, takes authorisation and obligations policies apart. Ponder does not provide explicit strategies to resolve conflictual decisions possibly arising in policy evaluation, rather it relies on abductive reasoning to statically prevent conflicts from occurring, although no implementation or experimental results are presented. On the contrary, \facpl\ provides combining algorithms, as we think they offer higher degrees of freedom to policy developers for managing conflicts. Similarly to Ponder, \facpl\ uses a mnemonic textual specification language and addresses value types, although they are not explicitly reported. Finally, the \facpl\ evaluation process is triggered by requests and not by events as in Ponder. Anyway, the \facpl\ approach is as general as the Ponder one since, by exploiting attributes, requests can represent any event of a system.


The logic-based approach mainly exploits predicate or multi-valued logics. Most of these proposals are based on Datalog~\cite{CeriGT89} (see, e.g., ~\cite{Jajodia97alogical,HashimotoKTT09,DeTreville:2002}), which implies that the access rules are defined as first order logic predicates. In general, these approaches offer valuable means for a low-level design of rules, but the lack of high-level features, e.g. combining algorithms or obligations,  prevent them from representing policies like those of \facpl.

ASL
is one of the firstly defined logic-based languages. It expresses authorisation policies based on user identity credentials and authorisation privileges, and supports hierarchisation and propagation of access rights among roles and groups of users.  Additional predicates enable the definition of (a posteriori) integrity checks on authorisation decisions, e.g. conflict resolution strategies. Differently from ASL, \facpl\ provides high-level constructs and offers by-construction many not straightforward features like, e.g., conflict resolution strategies. A suitable use of policies hierarchisation enables propagation of access rights also in \facpl\ specifications.

PTaCL
follows the logic-based approach as well, but it does not rely on Datalog. It defines two sets of algebraic operators based on a multi-valued logic: one modelling target expressions, the other one defining policy combinations. These operators emphasise the role of missing attributes in policy evaluation, in a way similar to \facpl, but only partially address errors. In fact, combination operators are not defined on error values: it is rather assumed that all target functions are string equalities that never produce errors. Similarly to \facpl, PTaCL permits formalising the \emph{non-monotonicity} and \emph{safety} properties of attribute-based policies introduced in~\cite{TschantzK06}. The PTaCL extension reported in~\cite{CramptonW15} introduces obligations and their instantiation, but it still lacks error handling.

A similar study, but more focussed on the distinguishing features of \xacml, is reported in~\cite{RamliNN14}. It introduces a formalisation of \xacml\ in terms of multi-valued logics, by first considering 4-valued decisions and then 6-valued ones. Most of the \xacml\ combining algorithms are formalised as operators on a partially ordered set of decisions, while the algorithms $\firstAppO{}$ and $\onlyOneAppO{}$ are defined by case analysis. Differently from \facpl, this formalisation does not deal with missing attributes and obligations, which have instead a crucial role in \xacml\ policy evaluation.

Another logic-based language is presented in~\cite{ArkoudasCC14}. In this case, a policy is a list of constraint assertions that are evaluated by means of an SMT solver. The framework supports reasoning about different properties, but any high-level feature, e.g. combining algorithms, has to be encoded `by hand' into low-level assertions. In addition, missing attributes, erroneous values and obligations are not addressed. 

Multi-valued logics and the relative operators have also been exploited to model the behaviour of combining algorithms. For example, the \emph{Fine-Integration Algebra} introduced in~\cite{RaoLBLL09} models the strategies of \xacml\ combining algorithms by means of a set of 3-valued (i.e., $\permit$, $\deny$ and $\notApp$) binary operators. The behaviour of each algorithm is then defined in terms of the iterative application of the operators to the policies of the input sequence.
This approach significantly differs from the \facpl\ one since it does not consider the $\indet$ decision. Instead,~\cite{LiWQBRLL09} explicitly introduces an error handling function that, given two decisions, determines whether their combination produces an error, i.e. an $\indet$ decision. Each (binary) operator is then defined  using such error function. The formalisation of \facpl\ combining algorithms follows a similar approach, but it also deals with obligations and instantiation strategies, which require different iterative applications of operators.

Moreover, in~\cite{LiWQBRLL09} nonlinear constraints are used for the specification of combining algorithms which return a decision $\dec$ if the majority of the input policies return $\dec$. Such algorithms are not usually dealt with in the literature and cannot be expressed in terms of iterative application of some binary operators.

\subsection{Analysis of Access Control Policies}
\label{sec:rel_analysis}

The increasing spread of policy-based specifications has prompted the development of many verification techniques like, e.g., property checking and behavioural characterisations. Such techniques have been implemented by means of different formalisms, ranging from SMT formulae to multi-terminal binary decision diagrams (MTBDD), including different kinds of logics. Hereafter we review the more relevant ones.

The works concerning policy analysis that are closer to our approach are of course those exploiting SMT formulae.
In~\cite{TurkmenHRZ15}, a strategy for representing \xacml\ policies in terms of SMT formulae is introduced. The representation, which is based on an informal semantics of \xacml, supports integers, booleans and reals, while the representation of sets of values and strings is only sketched. The combining algorithms are modelled as conjunctions and disjunctions of formulae representing the policies to be combined, \ie in a form similar to the approach shown in Appendix~\ref{sec:appendixA}. As a design choice, formulae corresponding to the $\notApp$ decision are not generated, because they can be inferred as the complementary of the other ones. Thus, in case of algorithms like $\denyUnlessO{}$, additional workload is required. Moreover, the representation assumes that each attribute name is assigned only to those values that match the implicit type of the attribute, hence the analysis cannot deal with missing attributes or erroneous values. Finally, it does not take into account obligations, which have instead an important role in the evaluation. The SMT-based framework of~\cite{ArkoudasCC14}, mentioned in Section~\ref{sec:pol_lang}, suffers from similar drawbacks.

The only analysis approach that takes missing attributes into account is presented in~\cite{CramptonMZ15}. The analysis is based on a notion of request extension, as we have done in Section~\ref{sec:analysis}. Differently from our approach, this analysis aims at quantifying the impact of possibly missing attributes on policy evaluations. 

The change-impact analysis of XACML policies presented in~\cite{FislerKMT05} aims at studying the consequences of policy modifications. In particular, to verify structural properties among policies by means of automatic tools, this approach relies on an MTBDD-based representation of policies. However, it cannot deal with many of the \xacml\ combining algorithms and, as outlined in~\cite{ArkoudasCC14}, an SMT-based approach like ours scales significantly better than the MTBDD one.

Datalog-based languages, like \eg ASL, only provide limited analysis functionalities, that are anyway significantly less performant than SMT-based approaches. In general, these languages are useful to reason on access control issues at an high abstraction level, but they neglect many of the advanced features of modern access control systems.


Description Logic (DL) is used in~\cite{KolovskiHP07} as a target formalism for representing a part of \xacml. The approach does not take into account many combining algorithms and the decisions $\notApp$ and $\indet$. Thus, it only permits reasoning on a set of properties significantly reduced with respect to that supported by our SMT-based approach. Furthermore, DL reasoners support the verification of structural properties of policies but suffer from the same scalability issues as the MTBDD-based reasoners.

Answer Set Programming (ASP) is used in~\cite{AhnHLM10,RamliNN12} for encoding \xacml\ and enabling verification of structural properties that are similar to the $\complete$ one defined in Section~\ref{sec:struct_prop}. This approach however suffers from some drawbacks due to the nature of ASP. In fact, differently from SMT, ASP does not support quantifiers and multiple theories like datatype and arithmetic. Some seminal extensions of ASP to ``Modulo Theories" have been proposed, but, to the best of our knowledge, no effective solver like Z3 is available. Similarly, the work in~\cite{HughesB08} exploits the SAT-based tool Alloy~\cite{Jackson:2002} to detect inconsistencies in \xacml\ policies. However, as outlined in~\cite{ArkoudasCC14} and~\cite{FislerKMT05}, Alloy is not able to manage even quite small policies and, more importantly, it cannot reason on arithmetic or any additional theory.


Finally, it is worth noticing that various analysis approaches using SAT-based tools have been developed for the Ponder language, see e.g.~\cite{BandaraLR03}. These approaches, however, cannot actually be compared with ours due to the numerous differences among Ponder and \facpl. Furthermore, many other works deal with the analysis of access control policies by using, e.g., process algebra and model checking techniques. However, they consider only a limited part of access control policy aspects and suffer from scalability issues with respect to SMT-based tools.

In summary, all the approaches to the analysis of access control policies mentioned above are deficient in several respects when compared with ours. Those based on SMT formulae do not address relevant aspects like, e.g. missing attributes, while the other ones do not enjoy the benefits of using SMT, \ie support of multiple theories and scalable performance.



\subsection{Supporting Tools: Performance and Functionalities}
\label{sec:performance}

The effectiveness of supporting tools is a crucial point for the usability of a policy language. Therefore, hereafter we examine the performance of both the \facpl\ Java library and the SMT-based automatic analysis, and the functionalities offered by the \facpl\ IDE.

Concerning the library, we conducted two different tests\footnote{Both tests were conducted on a MacBook Pro, 3.1 GHz Intel i7, 16~Gb RAM, running OS X Sierra.}: 
(i) a performance comparison with a state-of-the-art \xacml\ tool on the CONTINUE~\cite{Krishnamurthi03} case study (partially analysed in~\cite{FislerKMT05}); (ii) a performance stress test on a large set of randomly generated policies, thus to analyse the scalability of the library. We present below our test results, focussing on the most relevant ones; the suites of policies and requests, as well as all test results, are available at \url{http://facpl.sf.net/test/}. 

\begin{figure}[!h]
\centering
\vspace*{-.2cm}
\hspace*{-.55cm}\includegraphics[scale=.45]{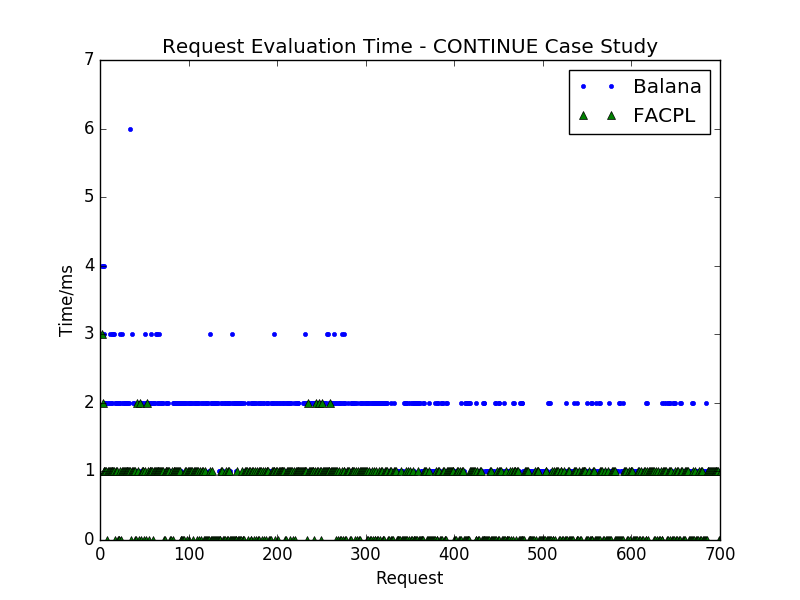}
\vspace*{-,4cm}
\caption{\facpl\ vs. Balana performance evaluation}
\label{fig:performance}
\end{figure}

The \xacml\ standard is by now the point-of-reference for industrial access control. In the authors' knowledge, the most up-to-date, freely available \xacml\ tool is Balana~\cite{Balana}. Differently from our framework that represents \facpl\ policies as Java classes, Balana manages \xacml\ policies directly in XML by exploiting a DOM representation of the XML files and evaluating \xacml\ requests through a visit of the DOM representing the policy. We have compared the evaluation of more than 1.500 requests and obtained the 
results reported in Figure~\ref{fig:performance}; for the sake of readability only 700 requests are reported. The mean request evaluation time is 0.49ms for \facpl\ and 1.27ms for Balana: evaluating a Java class ensures higher performance than navigating the DOM. Additionally, Balana requires an initial set-up time of 770ms to create the DOM.

\begin{table*}[!b]
\begin{minipage}{0.45\textwidth}
\centering
$
\begin{tikzpicture}
\draw(-1.9,0) -- (-1.9,-2.2);
\node at (0,.2) {$\policy(d,w,a)$};
\draw(-2.6,0) -- (2.2,0);
\node at (-2.2,-.35) {$\textrm{\#}w^1$};
\node at (-.1, -.35) {$\policy^1_1 \ \  \ldots \ \  \policy^1_w$};
\node at (-2.2,-.8) {$\textrm{\#}w^2$};
\node at (-1, -.8) {$\policy^2_1  \ldots  \policy^2_w$};
\node at (-.105,-.8) {$\ldots$};
\node at (1, -.8) {$\policy^2_{w+1}  \ldots  \policy^2_{w^2}$};
\node at (-2.2,-1.2) {$\vdots$};
\node at (-.9,-1.2) {$\vdots$};
\node at (-.105,-1.2) {$\vdots$};
\node at (.9,-1.2) {$\vdots$};
\node at (-2.2,-1.8) {$\textrm{\#}w^d$};
\node at (-1.3, -1.8) {$\policy^d_1  \ldots  \policy^d_w$};
\node at (-.105,-1.8) {$\ldots$};
\node at (1.2, -1.8) {$\policy^d_{w^d -w+1}  \ldots  \policy^d_{w^d}$};
\end{tikzpicture}
$
\end{minipage}
\begin{minipage}{0.45\textwidth}
\centering
\quad 
$
\begin{array}{|@{\,\,}c@{\,\,}||@{\,\,}c@{\,\,}|@{\,\,}c@{\,\,}|@{\,\,}c@{\,\,}|@{\,\,}c@{\,\,}|@{\,\,}c@{\,\,}|}

\lfrac{d}{w} & 1 & 2 & 3 & 4 & 5 \\
 \hline
 \hline
1 & 1 &   2& 3 & 4 &5\\
2 & 2 &   6&  12 & 20& 30\\
3 & 3 & 14&   39& 84&155\\
4 & 4 & 30&   120& 340& 780 \\
5 & 5 & 62&   363& 1364& 3905 \\
\hline
\end{array}
$
\end{minipage}
\\[.1cm]
\begin{minipage}{0.45\textwidth}
\centering
(a)
\end{minipage}
\begin{minipage}{0.45\textwidth}
\centering
(b) 
\end{minipage}
\vspace*{-.2cm}
\caption{(a) Structure of each of policy $\policy(d,w,a)$, where $a$ is the number of occurring attributes names; (b) Total number of sub-policies for each combination of $d$ and $w$ }
\label{tab:testbed}
\end{table*}

The CONTINUE case study is by now used as a standard benchmark in the field of access control tools. However it is relatively small: it is made of 24 policies controlling 14 attributes. All policies are combined through the $\firstAppO{}$ algorithm thus, as soon as a policy applies, the evaluation stops. Therefore, for evaluating performance and scalability of the \facpl\ Java library, we have also considered a set of large randomly-generated policies. We generated the policies according to the following criteria: (i) a variable number of occurring \emph{attribute names} (i.e., 10, 100, 1.000 or 10.000); (ii) a variable policy \emph{depth} (i.e., from 1 to 5); (iii) a variable policy \emph{width} (i.e., from 1 to 5); (iv) only the $\all$ instantiation strategy is used (so to always require the evaluation of all the occurring policies). The combinations of the previous criteria give rise to a test-bed of 100 policies, formed by a distinct number of differently structured sub-policies featuring a different number of attribute names. More specifically, given a number of attribute names $a$, depth $d$ and width $w$, the policy $\policy(d,w,a)$ is generated according to the template in Table~\ref{tab:testbed}a. Namely, $d$ corresponds to the nesting levels in the policy hierarchy, while $w$ corresponds to the number of policies that each policy (set) in the hierarchy contains. The total number of sub-policies contained by every policy $\policy(d,w,a)$ is summarised in Table~\ref{tab:testbed}b. For each of the 25 generated policies, there are 4 different versions, one for each value that $a$ can take.

The generated test-bed has been used to perform the stress test on the \facpl\ library. The results, when $a$ is set to $10.000$, are summarised in Figure~\ref{fig:performance_stress}. The graphs show how the performance changes as a function of the policy structure, \ie depending on $d$ and $w$. Better performances are obtained by  structuring policies in terms of larger width values (marked by the blue square), rather than larger depth values (marked by the green triangle). Namely, the average evaluation time increases more significantly by increasing policy depth than width.

\begin{figure}[!t]
\centering
\vspace*{-.2cm}
\includegraphics[scale=.45]{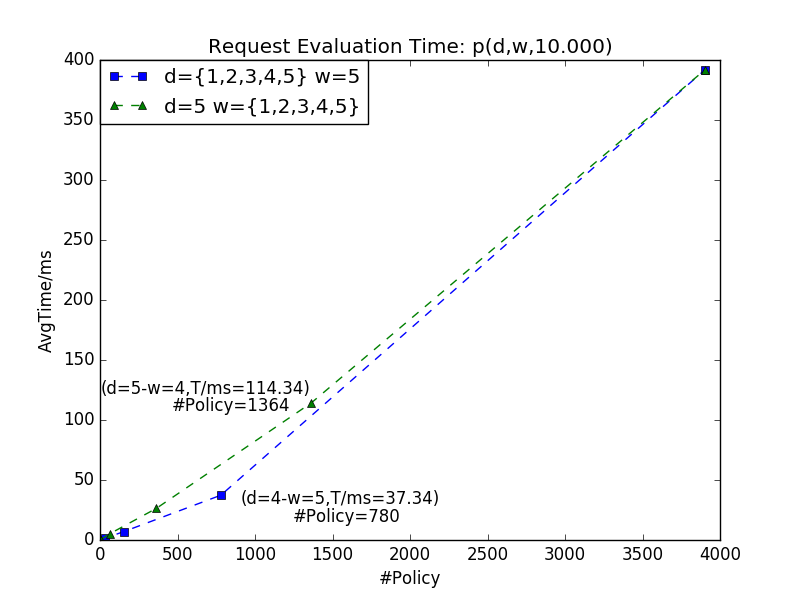}
\vspace*{-,4cm}
\caption{\facpl\ performance stress test (when $a=10.000$)} 
\label{fig:performance_stress}
\end{figure}

Concerning the automatic analysis, as previously pointed out, the tool closer to ours is that of~\cite{ArkoudasCC14}, which relies on the SMT solver Yices~\cite{Yices}. Differently from Z3, Yices does not support datatype theory, which is instead crucial to deal with a wide range of policy aspects, as \eg missing and erroneous attributes. To analyse the completeness of the CONTINUE policies, the Yices-based tool requires around 570ms, while our Z3-based tool requires around 120ms\footnote{The Yices value is taken directly from~\cite{ArkoudasCC14}, since the provided CONTINUE implementation only runs on Windows machines. 
Therefore, we ran this Z3 analysis on an older comparable hardware configuration (with the current configuration it takes only 60ms).}. 
\begin{figure}[!b]
\centering
\vspace*{-.2cm}
\includegraphics[scale=.45]{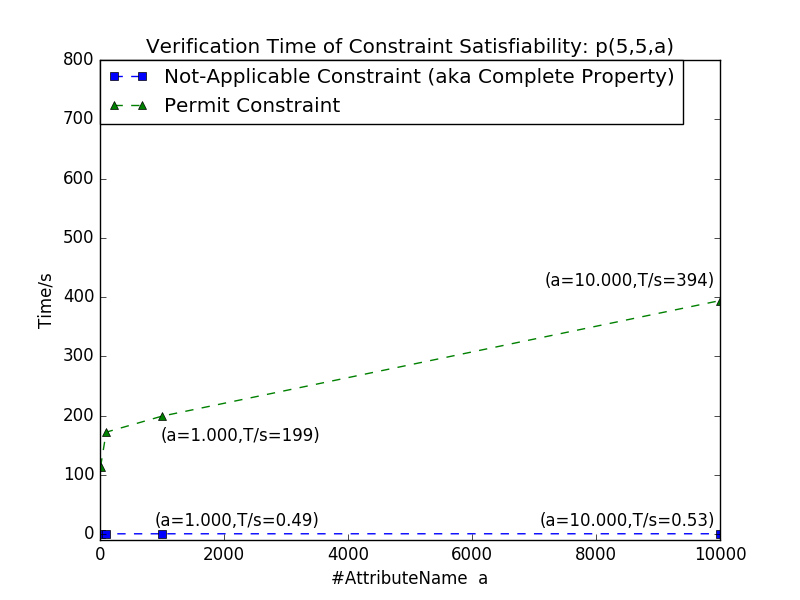}
\caption{\facpl\ analysis performance (when $d=w=5$)} 
\label{fig:analysis_perf}
\end{figure}
Notably, other tools not based on SMT, like e.g. Margrave, have significantly worse performance when policies scale. In fact, as reported in~\cite{ArkoudasCC14}, the increment of the number of possible values for the attributes occurring in the CONTINUE policies prevents Margrave to accomplish the analysis. On the contrary, SMT solvers can also deal with infinite sets of attribute values, as \eg integers. 
To further evaluate the analysis performance, we also report in Figure~\ref{fig:analysis_perf} the time required to verify the satisfiability of the $\notApp$ constraint (i.e., the verification of the $\complete$ property; marked by the blue square), and the $\permit$ constraint (marked by the green triangle) of the policies $\policy(5,5,a)$, i.e. 3905 policies,  with a varying number $a$ of attributes. 
Namely, the $\complete$ property is always verified in less than one second, despite the increasing number of attributes. Instead, the verification time for the satisfiability of the $\permit$ constraint increases by rising the number of attribute names, but the increments are significantly lower than the attribute name variations, i.e. from $199s$ with $1.000$ attribute names, to $394s$ with $10.000$ ones.
The difference between the two cases is due to the policy semantics: as soon as a policy target is $\notApp$, the whole policy is $\notApp$, while a policy evaluates to $\permit$ according to the combination strategies of the chosen combining algorithms.
It is finally worth noticing that the considered policies have a limited number of attribute names representing set values. In fact, a higher number of set attributes would require many satisfiability checks of existential quantifiers along with the 3905 policies (see the definition of the \texttt{inInt} function in Section~\ref{sec:contr_eval}). 
In such a case, the analysis remains feasible, i.e. hours instead of minutes, according to the cases. For example, it lasts two hours for the case of $100$ attribute names and $20$ set attributes, while it cannot complete for the case of $1.000$ attribute names and $200$ set ones.


We conclude by commenting on the most strictly related IDEs. To the best of our knowledge, the only freely available IDEs are the ALFA Eclipse plugin by Axiomatics (\url{http://www.axiomatics.com/alfa-plugin-for-eclipse.html}) and the graphical editor of the Balana-based framework (\url{http://xacmlinfo.org/category/xacml-editor/}). However, differently from our IDE, they only provide a high-level language for writing \xacml\ policies. Additionally, ALFA does not provide any request evaluation engine, since the Axiomatics one is a proprietary software.


\section{Concluding Remarks and Future Work}
\label{sec:conclusions}

We have described a full-fledged framework for the specification, analysis and enforcement of access control policies. Our framework relies on \facpl\ and is built on top of solid formal foundations. The \facpl\ semantics provides a formalisation of complex access control features ---including obligations and missing attributes, which are instead overlooked by many other proposals--- and lays the basis for developing analysis techniques and tools. We have shown that \facpl\ policies can be represented in terms of specific SMT formulae, whose automatic evaluation permits verifying various authorisation and structural properties. We have demonstrated feasibility and effectiveness of our approach by means of a case study from the e-Health application domain, which is currently one of the most critical application domains of access control systems.
We have also shown that the use of SMT solvers provides us with stable and efficient tools, ensuring better performance than many other approaches from the literature. 

In a general perspective, our approach brings together the benefits deriving from using a high-level, mnemonic rule-based language with the rigorous means provided by denotational semantics and constraints. Most of all, the supporting tools we implemented allow access control system developers to use any of the formally-defined functionalities provided by our framework, without the need that they be familiar with formal methods. 

We conclude by first enlighting some distinguishing traits of \facpl\ (Section~\ref{sec:discuss}), then by pointing out some future research directions (Section~\ref{sec:future}).

\subsection{Discussion}
\label{sec:discuss}

We want here to recap and reflect on a few characteristics of \facpl\ (and its framework) and the design choices that underlie them.

\smallskip
\noindent 
\emph{Expressiveness}.\ 
The access control systems expressible by \facpl\ are those expressible by \xacml\ (but not dealing with XML raw data, see Section~\ref{sec:facplvsxacml}), with in addition the possibility of using consensus-based combining algorithms and the $\all$ instantiation strategy for obligations. \facpl\ access control systems are systematically more compact (see Table~\ref{tab:comparison}) and feature a smoother management of errors and missing attributes. This latter characteristic, together with the fact that we decided not to introduce any static check on the type of requests, permits to accurately deal with every access control request. 
Alternatively, we could have defined a type inference system in order to statically check the policies and infer the expected type of any attribute name occurring within. Then, we could have reserved evaluation for only those requests whose attribute names comply with their expected type, while we could have directly returned the $\indet$ decision for all the other requests. By pursuing such an approach, however, we would have lost expressiveness, since policies could not be able anymore to automatically manage errors due to unexpected attribute values and possibly mask them by using operators that combine, according to different strategies, $\indet$ decisions with the others.

Besides the definition of access controls, \facpl\ permits defining obligations, which are a key ingredient to enhance the expensiveness of access control systems. As exemplified in the definition of the e-Health case study (see Section~\ref{sec:facplEx}), the instantiation of obligations permits defining context-dependent actions to be enforced at run-time in the controlled system. Indeed, the side-effects of policy evaluation are not only the enforcement of access decisions, but also the enforcement of dynamically instantiated actions. 
\facpl\ obligations permit enforcing, e.g., resource usage, adaptation and emergency handling strategies. For example,  
in the application of the Ponder language of~\cite{SlomanL10} and in the preliminary version of \facpl\ of~\cite{MargheriPT13,WS-FM}, obligations are used to enforce self-adaptation strategies in autonomic computing systems. Instead, in the context of emergency handling, an obligation-based approach is proposed in \cite{BruckerP09,MarinovicDS14} to enforce the principle known as `break the glass', 
which means that authorisation controls can be bypassed in case of emergency.


We intentionally abstract from the actual syntax of obligations. They are simply intended to be actions executed at run-time.
From time to time, they can be chosen to more adequately express the access control system in hand. We also abstract from the actual semantics obligations. Indeed, the discharging of obligations done by the \pep\ simply refers to the fact that the system has taken charge of their execution, which is intended to finish by the conclusion of the \pep\ enforcement process. However, the possibility of enforcing some obligations after releasing the decision and granting the access is a topic worth to be studied
(it is indeed one of the future research directions we want to pursue).


\smallskip
\noindent
\emph{Validation}.\ Our framework has essentially three constituent elements: (i) the linguistic constructs together with their denotational semantics; (ii) the constraint formalism and the semantic-preserving translation; (iii) the Java-based supporting tools. For each of them we have presented different validation results, both theoretical and empirical. 

The linguistic constructs are validated with respect to their expressiveness. This is done, on the one hand, by modelling a real-world case study (Section~\ref{sec:facplEx}) from the e-Health application domain, on the other hand, by comparing (Section~\ref{sec:facplvsxacml}) \facpl\ with \xacml, \ie the state-of-the-art OASIS standard for attribute-based access control systems. \facpl\ formal semantics is validated according to the so-called \emph{reasonability properties} of~\cite{TschantzK06} that precisely characterise the expressiveness of a policy language. Besides these properties, we show that the semantics is well-defined (Theorem~\ref{theo:deterministic}) and precisely characterise the attributes that are relevant for policy evaluation (Lemma~\ref{lemma1}); this important result, as pointed out in Section~\ref{sec:towautverif}, underlies the automatic property verification. All the results are presented in Section~\ref{sec:prop_sem}, while their proofs are relegated to Appendix~\ref{app:formal_sem}. 

Similarly, the constraint formalism and the semantic-preserving translation of \facpl\ policies into SMT formulae are validated by the theoretical results presented in Section~\ref{sec:prop_transl} and proved in Appendix~\ref{sec:appendix4}. All together these results ensure that the approach to the analysis of \facpl\ policies presented in Section~\ref{sec:analysis} is sound. 

The software tools are validated by empirically examining their performance and functionalities. The obtained results are reported in Section~\ref{sec:performance}.

\smallskip
\noindent
\emph{Exploitation}.\ 
The \facpl\ framework is a production-level software that is also used in industry. Indeed, since its preliminary version, \facpl\ has been used by Tiani Spirit (\url{http://www.tiani-spirit.com/}) instead of \xacml\
to carry out design and automatic analysis of access control policies. In particular, the \facpl\ access control engine has been used as XACML reference implementation in several projects. 
Furthermore, \facpl\ was used for team works in a PhD school on engineering Collective Autonomic Systems (\url{http://www.ascens-ist.eu/springschool}) and has been used in many bachelor and master thesis (further details can be found at the \facpl\ web-site). These practical exploitations have highlighted that its compact, mnemonic syntax requires very short learning time, even to undergraduate students. The users have also appreciated the flexibility of the IDE, which can be smoothly integrated within 
other development environments.

\smallskip
\noindent 
\emph{Extendability}.\ The proposed framework offers a variegated set of constructs, ranging from expression operators to combining algorithms,  for defining access controls. Anyway, to better suit any need, as reported in Section~\ref{sec:tool}, both the Java library and the IDE can be easily extended with the introduction of, e.g., new expression operators. This approach supports writing customised \facpl\  specifications. These specifications can then be translated, in accordance with the user's definition of the added constructs, to Java and SMT-LIB code that can be still evaluated and analysed, respectively. The formal assurance of semantic preservation (Theorem~\ref{thr:constr_sem}) can be easily tailored for encompassing the user's extensions. For instance, in case of addition of new expression operators, it only requires devising a constraint operator (or a combination thereof) that faithfully represents the semantics of the new operator.

\subsection{Future Work}
\label{sec:future}

In the next future, we plan to address the issue of controlling the accesses while they are in progress. In this sort of `continuative' access control, the challenge is to ensure guarantees on how granted accesses are used. This model is usually referred to as Usage Control~\cite{LazouskiMM10} in the literature and has been recently studied by various researchers. 
To deal with usage control, temporal aspects are of paramount importance, both to refer to ongoing accesses and to enforce obligations after releasing access decisions. To this aim, we will provide a \facpl -based solution for usage control that, by relying on the already available context-dependent authorisation process, can control ongoing accesses and instantiate temporal obligations. To actually enforce these obligations and, consequently, reason on them, we will refine the \pep\ semantics by appropriately instantiating the predicate $\Downarrow\!\x{ok}$ introduced in  Section~\ref{sec:sem_enfAlg}.

We also plan to provide a formally-based analysis technique that system developers can exploit to verify, e.g., history-dependent properties like dynamic separation of duty. To this aim, besides formalising new history-dependent authorisation properties, we want to define and verify properties on conflicts and dependencies among obligations.


\bibliographystyle{plain}
\bibliography{biblio}

\newpage 

\appendix

\begin{table*}[!p]
\centering
$
\footnotesize
\begin{array}{c}
\mbox{}\\[.4cm]
    \begin{array}{l||c|c|c|c|}
    \algOpAlg{\permitOverO{}} 
    	    & \langle \permit\ \ \mathit{FO}_2 \rangle & \langle \deny\ \ \mathit{FO}_2 \rangle & \notApp & \indet \\[.08cm]
                \hline\hline
                \langle \permit\ \ \mathit{FO}_1 \rangle  & \langle \permit\ \ \mathit{FO}_1\concat\mathit{FO}_2 \rangle & \langle \permit\ \ \mathit{FO}_1 \rangle & \langle \permit\ \ \mathit{FO}_1 \rangle & \langle \permit\ \ \mathit{FO}_1 \rangle \\
                \langle \deny\ \ \mathit{FO}_1 \rangle & \langle \permit\ \ \mathit{FO}_2 \rangle & \langle \deny\ \ \mathit{FO}_1\concat\mathit{FO}_2 \rangle & \langle \deny\ \ \mathit{FO}_1 \rangle & \indet \\
                \notApp & \langle \permit\ \ \mathit{FO}_2 \rangle & \langle \deny\ \ \mathit{FO}_2 \rangle & \notApp & \indet \\
                \indet & \langle \permit\ \ \mathit{FO}_2 \rangle & \indet & \indet & \indet \\
                \hline
    \end{array}
\\
\mbox{}\\[.4cm]
    \begin{array}{l||c|c|c|c|}
    \algOpAlg{\denyOverO{}} 
    & \langle \permit\ \ \mathit{FO}_2 \rangle  & \langle \deny\ \ \mathit{FO}_2 \rangle & \notApp & \indet \\
            \hline\hline
            \langle \permit\ \ \mathit{FO}_1 \rangle  & \langle \permit\ \ \mathit{FO}_1\concat\mathit{FO}_2 \rangle &  \langle \deny\ \ \mathit{FO}_2 \rangle  & \langle \permit\ \ \mathit{FO}_1\rangle  & \indet \\
            \langle \deny\ \ \mathit{FO}_1 \rangle  & \langle \deny\ \ \mathit{FO}_1 \rangle & \langle \deny\ \ \mathit{FO}_1\concat\mathit{FO}_2 \rangle & \langle \deny\ \ \mathit{FO}_1 \rangle & \langle \deny\ \ \mathit{FO}_1 \rangle \\
            \notApp &  \langle \permit\ \ \mathit{FO}_2 \rangle & \langle \deny\ \ \mathit{FO}_2 \rangle & \notApp & \indet\\
            \indet &  \indet & \langle \deny\ \ \mathit{FO}_2 \rangle & \indet & \indet \\
            \hline
	\end{array}
\\
\mbox{}\\[.4cm]		
  \begin{array}{l||c|c|c|c|}
    \algOpAlg{\denyUnlessO{}} 
    	    & \langle \permit\ \ \mathit{FO}_2 \rangle & \langle \deny\ \ \mathit{FO}_2 \rangle & \notApp & \indet \\[.08cm]
                \hline\hline
                \langle \permit\ \ \mathit{FO}_1 \rangle  & \langle \permit\ \ \mathit{FO}_1\concat\mathit{FO}_2 \rangle & \langle \permit\ \ \mathit{FO}_1 \rangle & \langle \permit\ \ \mathit{FO}_1 \rangle & \langle \permit\ \ \mathit{FO}_1 \rangle \\
                \langle \deny\ \ \mathit{FO}_1 \rangle & \langle \permit\ \ \mathit{FO}_2 \rangle & \langle \deny\ \ \mathit{FO}_1\concat\mathit{FO}_2 \rangle & \langle \deny\ \ \mathit{FO}_1 \rangle & \langle \deny\ \  \mathit{FO}_1 \rangle \\
                \notApp & \langle \permit\ \ \mathit{FO}_2 \rangle & \langle \deny\ \ \mathit{FO}_2 \rangle &  \langle \deny\ \  \epsilon \rangle &  \langle \deny\ \  \epsilon \rangle \\
                \indet & \langle \permit\ \ \mathit{FO}_2 \rangle & \langle \deny\ \  \mathit{FO}_2 \rangle & \langle \deny\ \  \epsilon \rangle &  \langle \deny\ \  \epsilon \rangle \\
                \hline
    \end{array}
\\
\mbox{}\\[.4cm]
    \begin{array}{l||c|c|c|c|}
    \algOpAlg{\permitUnlessO{}} 
    & \langle \permit\ \ \mathit{FO}_2 \rangle  & \langle \deny\ \ \mathit{FO}_2 \rangle & \notApp & \indet \\
            \hline\hline
            \langle \permit\ \ \mathit{FO}_1 \rangle  & \langle \permit\ \ \mathit{FO}_1\concat\mathit{FO}_2 \rangle &  \langle \deny\ \ \mathit{FO}_2 \rangle  & \langle \permit\ \ \mathit{FO}_1\rangle  & \langle \permit\ \ \mathit{FO}_1\rangle \\
            \langle \deny\ \ \mathit{FO}_1 \rangle  & \langle \deny\ \ \mathit{FO}_1 \rangle & \langle \deny\ \ \mathit{FO}_1\concat\mathit{FO}_2 \rangle & \langle \deny\ \ \mathit{FO}_1 \rangle & \langle \deny\ \ \mathit{FO}_1 \rangle \\
            \notApp &  \langle \permit\ \ \mathit{FO}_2 \rangle & \langle \deny\ \ \mathit{FO}_2 \rangle &  \langle \permit\ \ \epsilon \rangle  &  \langle \permit\ \ \epsilon \rangle \\
            \indet &   \langle \permit\ \ \mathit{FO}_2 \rangle  & \langle \deny\ \ \mathit{FO}_2 \rangle &  \langle \permit\ \ \epsilon \rangle  &  \langle \permit\ \ \epsilon \rangle  \\
            \hline
	\end{array}
\\
\mbox{}\\[.4cm]
    \begin{array}{l||c|c|c|c|}
    \algOpAlg{\firstAppO{}}  
    & \langle \permit\ \ \mathit{FO}_2 \rangle  & \langle \deny\ \ \mathit{FO}_2 \rangle & \notApp & \indet \\
           \hline\hline
            \langle \permit\ \  \mathit{FO}_1 \rangle &  \langle \permit\ \  \mathit{FO}_1 \rangle &   \langle \permit\ \  \mathit{FO}_1 \rangle  &   \langle \permit\ \  \mathit{FO}_1 \rangle  &    \langle \permit\ \  \mathit{FO}_1 \rangle   \\
             \langle \deny\ \  \mathit{FO}_1 \rangle  &  \langle \deny\ \  \mathit{FO}_1 \rangle &  \langle \deny\ \  \mathit{FO}_1 \rangle &  \langle \deny\ \  \mathit{FO}_1 \rangle &   \langle \deny\ \  \mathit{FO}_1 \rangle  \\
            \notApp &  \langle \permit\ \  \mathit{FO}_2 \rangle &  \langle \deny\ \  \mathit{FO}_2 \rangle & \notApp & \indet \\
            \indet & \indet & \indet & \indet & \indet \\
            \hline
    \end{array}
\\ 
\mbox{}\\[.4cm]
    \begin{array}{l||c|c|c|c|}
    \algOpAlg{\onlyOneAppO{}} 
    & \langle \permit\ \ \mathit{FO}_2 \rangle  & \langle \deny\ \ \mathit{FO}_2 \rangle & \notApp & \indet \\
           \hline\hline
            \langle \permit\ \  \mathit{FO}_1 \rangle &  \indet &   \indet  &   \langle \permit\ \  \mathit{FO}_1 \rangle  & \indet  \\
             \langle \deny\ \  \mathit{FO}_1 \rangle  &  \indet & \indet &  \langle \deny\ \  \mathit{FO}_1 \rangle & \indet \\
            \notApp &  \langle \permit\ \  \mathit{FO}_2 \rangle &  \langle \deny\ \  \mathit{FO}_2 \rangle & \notApp & \indet \\
            \indet & \indet & \indet & \indet & \indet \\
            \hline
    \end{array}
\\ 
\mbox{}\\[.4cm]
    \begin{array}{l||c|c|c|c|}
    \algOpAlg{\weakConO{}}  
    & \langle \permit\ \ \mathit{FO}_2 \rangle  & \langle \deny\ \ \mathit{FO}_2 \rangle & \notApp & \indet \\
           \hline\hline
            \langle \permit\ \  \mathit{FO}_1 \rangle &  \langle \permit\ \ \mathit{FO}_1\concat\mathit{FO}_2 \rangle &   \indet  &   \langle \permit\ \  \mathit{FO}_1 \rangle  & \indet \\
             \langle \deny\ \  \mathit{FO}_1 \rangle  &  \indet & \langle \deny\ \ \mathit{FO}_1\concat\mathit{FO}_2 \rangle &  \langle \deny\ \  \mathit{FO}_1 \rangle &  \indet\\
            \notApp &  \langle \permit\ \  \mathit{FO}_2 \rangle &  \langle \deny\ \  \mathit{FO}_2 \rangle & \notApp & \indet \\
            \indet & \indet & \indet & \indet & \indet \\
            \hline
    \end{array}
\\ 
\mbox{}\\[.4cm]
   \begin{array}{l||c|c|c|c|}
    \algOpAlg{\strongConO{}}  
    & \langle \permit\ \ \mathit{FO}_2 \rangle  & \langle \deny\ \ \mathit{FO}_2 \rangle & \notApp & \indet \\
           \hline\hline
            \langle \permit\ \  \mathit{FO}_1 \rangle &  \langle \permit\ \ \mathit{FO}_1\concat\mathit{FO}_2 \rangle &   \indet  &   \indet  &   \indet  \\
             \langle \deny\ \  \mathit{FO}_1 \rangle  &  \indet & \langle \deny\ \ \mathit{FO}_1\concat\mathit{FO}_2 \rangle & \indet & \indet \\
            \notApp & \indet &  \indet & \notApp & \indet \\
            \indet & \indet & \indet & \indet & \indet \\
            \hline
    \end{array}
\\[1cm]
\end{array}
$
\caption{Combination matrices for the binary operators $\algOpAlg{\alg{}}$}
\label{tab:all_combmatrices}
\end{table*}

\section{Definitions for Combining Algorithms}
\label{sec:appendixA}


%
%

In this section we report all the definitions regarding the combining algorithms. Table~\ref{tab:all_combmatrices} shows all the combination matrices defining the binary operators $\algOpAlg{\alg{}}$ for each algorithm $\alg{}$. Hereafter we report the constraint resulting from the combination of two constraint tuples, say $A$ and $B$, defined according to the various combining algorithms.
$$
\footnotesize
\begin{array}{@{}l@{}}
    	 \permitOverO{}(A,  B)  =  \\[.1cm]
	\ \ 
	 \begin{array}{@{}l@{}}  
	  \langle\, 
	 	\permit : A \proj{p} \vee\ B \proj{p} \\
    	\ \ \deny: (A \proj{d} \wedge\ B \proj{d}) \vee\ (A \proj{d} \wedge\ B \proj{n}) \vee\ (A \proj{n} \wedge\ B \proj{d}) \\
    	\ \ \notApp: A \proj{n} \wedge\ B \proj{n} \\
    	\ \ \indet : (A \proj{i} \wedge\ \lnot B \proj{p}) \vee\ (\lnot A \proj{p} \wedge\ B \proj{i})  \rangle 
	\end{array}
\\
\mbox{}\\[-.2cm]
    	 \denyOverO{}(A, B) = \\[.1cm]
	\ \ 
	 \begin{array}{@{}l@{}}  
	   \langle\, 
	 	\permit : (A \proj{p} \wedge\ B \proj{p}) \vee\ (A \proj{p} \wedge\ B \proj{n}) \vee (A \proj{n} \wedge\ B \proj{p})\\
    	\ \ \deny:   A \proj{d} \vee\  B \proj{d} \\
    	\ \ \notApp: A \proj{n} \wedge\ B \proj{n} \\
    	\ \ \indet : (A \proj{i} \wedge\ \lnot B \proj{d})  \vee\ (\lnot A \proj{d} \wedge\ B \proj{i})  \rangle 
	\end{array}
\\
\mbox{}\\[-.2cm]
    	 \denyUnlessO{}(A, B) =  \\[.1cm]
	\ \
	 \begin{array}{@{}l@{}}  
	  \langle\, 
	 	\permit : A \proj{p} \vee\  B \proj{p} \\
    	\ \ \deny:  \lnot A \proj{p} \wedge\ \lnot B \proj{p} \wedge\ (A \proj{d} \vee\ A \proj{n} \vee\ A \proj{i}) \\
	\qquad\qquad\quad
	\wedge (B \proj{d} \vee\ B \proj{n} \vee\ B \proj{i}) \\
    	\ \ \notApp: \false \\
    	\ \ \indet : \false  \rangle 
	\end{array}
\\
\mbox{}\\[-.2cm]
    	 \permitUnlessO{}(A, B) = \\[.1cm]
	\ \
	 \begin{array}{@{}l@{}}  
	 \langle\, 
	\permit :  \lnot A \proj{d} \wedge\ \lnot B \proj{d} \wedge\ (A \proj{p} \vee\ A \proj{n} \vee\ A \proj{i}) \\
		\qquad\qquad\quad
	\wedge (B \proj{p} \vee\ B \proj{n} \vee\ B \proj{i}) \\
    	\ \ \deny:  A \proj{d} \vee\  B \proj{d}\\
    	\ \ \notApp: \false \\
    	\ \ \indet : \false  \rangle 
	\end{array}
\\
\mbox{}\\[-.2cm]
    	 \firstAppO{}(A, B) = \\[.1cm]
	\ \ 
	 \begin{array}{@{}l@{}}  
	 \langle\, 
	\permit :  A\proj{p} \vee\ (B \proj{p} \wedge\ A \proj{n}) \\
    	\ \ \deny: A\proj{d} \vee\ (B \proj{d} \wedge\ A \proj{n}) \\
    	\ \ \notApp: A\proj{n} \wedge\ B\proj{n} \\
    	\ \ \indet : A \proj{i} \vee\ (A\proj{n} \wedge\ B\proj{i})  \rangle 
	\end{array}
\\
\mbox{}\\[-.2cm]
    	 \onlyOneAppO{}(A, B) =  \\[.1cm]
	\ \ 
	 \begin{array}{@{}l@{}}  
	 	 \langle\, 
	\permit :  (A\proj{p} \wedge\ B \proj{n}) \vee\ (A \proj{n} \wedge\ B \proj{p}) \\
    	\ \ \deny: (A\proj{d} \wedge\ B \proj{n}) \vee\ (A \proj{n} \wedge\ B \proj{d}) \\
    	\ \ \notApp: A\proj{n} \wedge\ B\proj{n} \\
    	\ \ \indet : A \proj{i} \vee\ B\proj{i} \vee\ ((A\proj{p} \vee\ A\proj{d}) \wedge\ (B\proj{p} \vee\ B\proj{d}))  \rangle 
	\end{array}
\\
\mbox{}\\[-.2cm]
    	 \weakConO{}(A, B) \\[.1cm]
	\ \ 
	 \begin{array}{@{}l@{}}  
	  \langle\, 
	\permit :  (A\proj{p} \wedge\ B \proj{p}) \vee\ (A \proj{p} \wedge\ \lnot B \proj{d}) \vee\ (\lnot A \proj{d} \wedge\  B \proj{p})\\ 
    	\ \ \deny: (A\proj{d} \wedge\ B \proj{d}) \vee\ (A \proj{d} \wedge\ \lnot B \proj{p}) \vee\ (\lnot A \proj{p} \wedge\  B \proj{d}) \\
    	\ \ \notApp: A\proj{n} \wedge\ B\proj{n} \\
    	\ \ \indet : (A\proj{p} \wedge\ B \proj{d}) \vee\ (A\proj{d} \wedge\ B\proj{p}) \vee\ A \proj{i} \vee\ B \proj{i} 
	\end{array}
\\
\mbox{}\\[-.2cm]
    	 \strongConO{}(A, B) = \\[.1cm]
	\ \ 
	 \begin{array}{@{}l@{}}  
	 \langle\, 
	\permit :  A\proj{p} \wedge\ B \proj{p}\\ 
    	\ \ \deny: A\proj{d} \wedge\ B \proj{d} \\
    	\ \ \notApp: A\proj{n} \wedge\ B\proj{n} \\
    	\ \ \indet : A\proj{i} \vee\ B\proj{i} \vee\ (A\proj{n} \wedge\ \lnot B\proj{n}) \vee\ (\lnot A\proj{n} \wedge\ B\proj{n}) \\
	 \qquad\qquad\ \vee\ (A\proj{p} \wedge\ B\proj{d}) \vee\ (A\proj{d} \wedge\ B\proj{p})\rangle \\
	\end{array}
\end{array}
$$

\section{Proofs of the Results}
\label{sec:appendix3}

Some of the proofs proceed by induction on the \emph{depth} of policies, \ie the number of their nesting levels, which is defined by induction on the syntax of policies as follows
$$
\begin{array}{@{}l}
\mathit{depth}(\ruleOpt{\effect\ \ \x{target:} \, \expr\ \ \x{obl:} \, \ob^{*} \,}) = 0\\
\mathit{depth}(\\
\{ \algSyntax\ \ \x{target:} \, \expr\ \ \x{policies:} \, \policy^{+}\, {\x{\oblp:} \, \ob_{p}^{*}\ \ \x{\obld:} \, \ob_{d}^{*}} \, \})= \\
\hfill1 + \mathit{max}( \{ \mathit{depth}(p) \mid \policy \in \policy^{+} \} )
\end{array}
$$
Policies with depth 0 are rules, the other ones are policies containing other policies. Notationally, we will use $\policy^i$ to mean that policy $\policy$ has depth $i$ and $(\policy^+)^i$ to mean that at least a policy in the sequence $\policy^+$ has depth $i$ and the others have depth at most $i$.

\subsection{Proofs of Results in Section~5}
\label{app:formal_sem}

\medskip

\textsc{Theorem}~\ref{theo:deterministic} (Total and Deterministic Semantics).
\begin{enumerate}
\item For all $pas \in \mathit{PAS}$ and $\rSyntax\in\mathit{Request}$, there exists a $\dec \in\mathit{Decision}$, such that $\pasSem{pas,\rSyntax} = \dec$. 
\smallskip
\item For all $pas \in \mathit{PAS}$, $\rSyntax\in\mathit{Request}$ and $\dec,\dec' \in\mathit{Decision}$, it holds that\\[.1cm]
$
\begin{array}{l}
\quad \pasSem{pas,\rSyntax} = \dec \ \ \wedge\ \ \pasSem{pas,\rSyntax} = \dec'\\ 
\quad \Rightarrow\ \ \dec=\dec'\,.
\end{array}
$ 
\end{enumerate}
\begin{proof}
The goal of the proof is to show that $\denSemF{P}as$ is a total and deterministic function, \ie 
it is defined for all possible input pairs and always returns the same decision any time it is applied to a specific pair.
If we let $pas$ be $\{ \,  \x{pep:} \, \enfAlg\ \ \x{pdp:}\, \pdpSyntax \, \}$
then, from the clause~(\ref{sem:pas}), we have that
$$
\begin{array}{l}
\pasSem{\{ \,  \x{pep:} \, \enfAlg\ \ \x{pdp:}\, \pdpSyntax \, \},\rSyntax} = \\
\qquad\qquad\qquad\qquad
\pepSem{\enfAlg}{(\pdpSem{\pdpSyntax}{(\reqSem{\rSyntax}{})})}
\end{array}
$$
Thus, since  
function composition preserves totality and determinism, we are left to prove that $\denSemF{R}$, $\denSemF{P}dp$ and $\denSemF{E}A$ are total and deterministic functions. Due to their inductive definition (given in Section~\ref{sec:formal_sem}), the proof proceeds by inspecting their defining clauses with the aim of checking that they satisfy the two requirements below
\begin{enumerate}
\item[R1:] there is one, and only one, clause that applies to each syntactic domain element (this usually follows since the definition is syntax-driven and considers all the syntactic forms that the input can assume);
\item[R2:] for each defining clause,
\begin{itemize}
\item the conditions in the right hand side are mutually exclusive (from the systematic use of the $\mathtt{otherwise}$ condition, it directly follows that they cover all the possible cases for the syntactic domain elements of the form occurring in the left hand side),
\item the values assigned in each case of the right hand side are obtained by only using 
total and deterministic functions/operators/predicates.
\end{itemize}
\end{enumerate}
\smallskip
\begin{description}
\item [Case $\denSemF{R}$.] 
From its defining clauses~(\ref{sem:req}) we get that $\denSemF{R}$ is defined on all non-empty sequences of attributes, \ie all requests. Moreover, the conditions in the right hand side of each clause are mutually exclusive and the operator $\Cup$ is total and deterministic by definition. Thus R1 and R2 hold, which means that $\denSemF{R}$ is a total and deterministic function. 
\smallskip
\item [Case $\denSemF{P}dp$.] 
To prove this case, we first prove that $\denSemF{E}$, $\denSemF{O}$, $\denSemF{A}$ and $\denSemF{P}$ are total functions. 
\begin{description}
\item [Case $\denSemF{E}$.] 
By an easy inspection of the clauses defining $\denSemF{E}$, an excerpt of which are in Table~\ref{tab:sem_expression}, it is not hard to believe that they satisfy R1 (since the application of the clauses is driven by the syntactic form of the input expression) and R2 above, hence $\denSemF{E}$ is a total function. Moreover, since the operator $\concat$ is total and deterministic, from the clauses~(\ref{sem:exp_con}) it follows that $\denSemF{E}$ remains a total and deterministic function also when extended to sequences of expressions.
\smallskip
\item [Case $\denSemF{O}$.] 
Since $\denSemF{E}$ is a total and deterministic function also on sequences of expressions, from the clauses~(\ref{sem:obl}) and (\ref{sem:obl2}) it follows that R1 and R2 hold, thus $\denSemF{O}$ is a total and deterministic function both on single obligations and on sequences of obligations.
\smallskip
\item [Cases $\denSemF{A}$ and $\denSemF{P}$.] 
The definitions of $\denSemF{P}$ and $\denSemF{A}$ are syntax-driven and consider all the syntactic forms that the input can assume, thus R1 is satisfied. Now, since $\denSemF{P}$ and $\denSemF{A}$ are mutually recursive, we prove by induction on the depth of their arguments that their defining clauses satisfy R2 for all input policies.
\smallskip
\begin{description}
\item [Base Case ($i=0$).] 
Let us start from $\denSemF{P}$. $\policy^0$ is of the form $(\effect\ \x{target}: \expr\ \x{obl}: \ob^*)$. We have hence to prove that the clause~(\ref{sem:rule}), which is the defining clause of $\denSemF{P}$ that applies to $\policy^0$, satisfies R2. This directly follows from the fact that $\denSemF{E}$ and $\denSemF{O}$ are total and deterministic functions.
Now, let us consider $\denSemF{A}$ and proceed by case analysis on $\algSyntax$.
\begin{description}
\item [$(a = \alg{\all}$ for any $\alg{})$.] 
Since the clause~(\ref{sem:rule}) satisfies (R1 and) R2, for each $\policy_j^0$ in $(\policy^+)^0$, $\policySem{\policy_j^0}{\req}$ is uniquely defined. Thus, since each operator $\algOp$ is total and deterministic by construction, the clause~(\ref{sem:algA}), to be used since the form of $a$, satisfies R2 (when all the input policies have depth 0).
\smallskip
\item [$(a = \alg{\greedy}$ for any $\alg{})$.]
This case is similar to the previous one, but involves the clause~(\ref{sem:algB}). It satisfies R2 (when all the input policies have depth 0) since its conditions in the right hand side are mutually exclusive by construction (indeed, each predicate $\isFinalPred{\x{alg}}$ and each operator $\algOp$ is total and deterministic).
\smallskip
\end{description}
\smallskip
\item [Inductive Case ($i = n +1$).] 
Let us start from $\denSemF{P}$. $\policy^{n+1}$ is of the form $\{ \algSyntax\ \, \x{target\!:} \expr\ \, \x{policies\!:} (\policy^+)^{n}$ ${\x{\oblp:} \, \ob_{p}^{*}\ \ \x{\obld:} \, \ob_{d}^{*}}  \} $. By the induction hypothesis, for any $\req$, $\algSyntax$ and $\policy^{k}_j$ in $(\policy^+)^{n}$, with $k\leq n$, the clauses defining $\denSemF{P}$ and $\denSemF{A}$ satisfy (R1 and) R2, that is $\policySem{\policy^{k}_j}{\req}$ and $\algSem{\algSyntax, (\policy^+)^{n}}{\req}$ are uniquely defined. Hence, the clause~(\ref{sem:pol}), to be used since the form of $p^{n+1}$, satisfies R2 as well. For $\denSemF{A}$, we can reason like in the base case by exploiting the induction hypothesis. We can thus conclude that both the clauses~(\ref{sem:algA}) and~(\ref{sem:algB}) satisfy R2 (for any input policy).
\end{description}
\smallskip
Therefore, $\denSemF{P}$ and $\denSemF{A}$ are total and deterministic functions.
\end{description}
\smallskip
Now, that $\denSemF{P}dp$ is a total and deterministic function directly follows from its defining clause~(\ref{sem:pdp}).
\smallskip
\item [Case $\denSemF{E}A$.]  
The requirement R1 is satisfied by definition. Moreover, since the predicate $\Downarrow \x{ok}$ is total and deterministic, the same holds for the function $\pepSemR{\ }$. Therefore, also R2 is satisfied by each defining clause (the conditions on $\pdpRes.\dec$ are trivially mutually exclusive). Hence, $\denSemF{E}A$ is a total and deterministic function.
\end{description}
\vspace*{-.5cm}
\end{proof}

\medskip

\noindent
\textsc{Lemma}~\ref{lemma1} (Policy relevant attributes).
For all $\policy \in \mathit{Policy}$ and $\req,\req' \in R$ such that $\req(\name) = \req'(\name)$ for all $\name \in \mathit{Names}(\policy)$ it holds that 
$\policySem{\policy}{\req} = \policySem{\policy}{\req'}$.
\begin{proof}
The statement is based on an analogous result concerning expressions
\[
\begin{array}{l}
\mbox{for all $\expr \in \mathit{Expr}$ and $\req_1,\req_1' \in R$ such that }\\
\mbox{$\req_1(\name) = \req_1'(\name)$ for all $\name \in \mathit{Names}(\expr)$,}\\
\mbox{it holds that $\exprSem{\expr}{\req_1} = \exprSem{\expr}{\req_1'}$}
\end{array}
\hspace*{-.3cm}
\tag{R}
\label{res:aux}
\]
which can be easily proven by structural induction on the syntax of expressions. Functions $\req_1$ and $\req_1'$ are only exploited in the base case when evaluating a name $\name \in \mathit{Names}(\expr)$ for which, by definition and hypothesis, we have $\exprSem{\name}{\req_1} = \req_1(\name) = \req_1'(\name) = \exprSem{\name}{\req_1'}$.
Since for any $\expr$ occurring in $\policy$, we have that $\mathit{Names}(\expr) \subseteq \mathit{Names}(\policy)$, from~(\ref{res:aux}), by taking $r_1=r$ and $r'_1=r'$, it follows that
\[
\begin{array}{@{}l}
\mbox{for all $\expr$ occurring in $\policy$, $\exprSem{\expr}{\req} = \exprSem{\expr}{\req'}$}
\end{array}
\tag{R-E}
\label{res:expr}
\]
From~(\ref{res:expr}), it also immediately follows that 
\[
\begin{array}{l}
\mbox{for all $\ob$ occurring in $\policy$, $\oblSem{\ob}{\req} = \oblSem{\ob}{\req'}$}
\end{array}
\tag{R-O}
\label{res:obl}
\]
Now we can prove the main statement by induction on the depth $i$ of $\policy$.
\begin{description}
\item [Base Case ($i=0$).] 
$\policy^0$ has the form $(\effect\ \x{target}: \expr\ \x{obl}: \ob^*)$, thus the clause~(\ref{sem:rule}) is used to determine $\policySem{\policy}{\req}$. The thesis then trivially follows from~(\ref{res:expr}) and~(\ref{res:obl}). 
\smallskip
\item [Inductive Case ($i = n +1$).] 
$\policy^{n+1}$ is of the form $\{ \algSyntax\ \ \x{target:} \, \expr\ \x{policies:} \, (\policy^+)^{n}\, {\x{\oblp:} \, \ob_{p}^{*}\ \ \x{\obld:} \, \ob_{d}^{*}} \}$, thus the clause~(\ref{sem:pol}) is used to determine $\policySem{\policy}{\req}$. By the induction hypothesis, for any $\policy^{k}_j$ in $(\policy^+)^{n}$, with $k\leq n$, it holds that $\policySem{\policy^{k}_j}{\req} = \policySem{\policy^{k}_j}{\req'}$. This, due to the clauses~(\ref{sem:algA}) and~(\ref{sem:algB}), implies that $\algSem{a, (\policy^+)^{n}}{\req} = \algSem{a, (\policy^+)^{n}}{\req'}$, for any algorithm $a$. The thesis then follows from this fact and  from~(\ref{res:expr}) and~(\ref{res:obl}).
\end{description}
\vspace*{-.5cm}
\end{proof}

\subsection{Proofs of results in Section~6}
\label{sec:appendix4}

%
%
%

\noindent
\textsc{Theorem}~\ref{thr:constr_fun} (Total and Deterministic Constraint Semantics).
\begin{enumerate}
\item For all $\const \in \mathit{Constr}$ and $\req \in R$, there exists an $\mathit{el} \in (\mathit{Value}\ \cup\ 2^{\mathit{Value}}\cup \{\err, \excpt \})$, such that $\cs{\const} = \mathit{el}$. 
\smallskip
\item For all $\const \in \mathit{Constr}$, $\req \in R$ and $\mathit{el}, \mathit{el}' \in (\mathit{Value}\ \cup\ 2^{\mathit{Value}}\cup \{\err, \excpt \})$, it holds that
$$
\cs{\const} = \mathit{el} \ \  \wedge \ \ \cs{\const} = \mathit{el}' \ \ \Rightarrow \ \ \mathit{el} = \mathit{el}'\,.
$$
\end{enumerate}
\begin{proof}
Similarly to the proof of Theorem~\ref{theo:deterministic}, the proof reduces to showing that $\denSemF{C}$ is a total and deterministic function. We proceed by structural induction on the syntax of $\const$.

\begin{description}
%
%
\item [Base Case.] If $\const=v$, the thesis immediately follows since $\cs{v} = v$; otherwise, \ie $\const=n$, we have $\cs{n} = \req(n)$ and the thesis follows because $\req$ is a total and deterministic function.
\smallskip

%
%
%
%
\item [Inductive Case.] It is not hard to believe that all the defining clauses of $\denSemF{C}$ are such that the conditions in the right hand side are mutually exclusive and cover all the necessary cases. For each different form that $\const$ can assume, the thesis then directly follows by the induction hypothesis.
\end{description}
\vspace*{-.5cm}
\end{proof}

\medskip

The proof of Theorem~\ref{thr:constr_sem} relies on the following three auxiliary results. 

\begin{lemma}
\label{lemma:expr}
For all $\expr \in \mathit{Expr}$ and $\req \in R$, it holds that
$$
\exprSem{\expr}{\req} = \cs{\translExpr{\expr}}
$$
\end{lemma}
\begin{proof}
We proceed by structural induction on the syntax of $\expr$ according to the translation rules of the clause~(\ref{cstr:expr}).
\begin{description}
\item [$(\expr = \name)$.] 
Since $\translExpr{\name} = \name$, the thesis follows because $\exprSem{\name}{\req} = \req(\name) = \cs{\name}$.
\smallskip
\item [$(\expr = \val)$.] 
Since $\translExpr{\val} = \val$, the thesis follows because $\exprSem{\val}{\req} = \val = \cs{\val}$.
\smallskip
\item [$(\expr = \x{not}(\expr_1))$.] 
Since $\translExpr{\expr} = \fnot\: \translExpr{\expr_1}$ and, by the induction hypothesis, $\exprSem{\expr_1}{\req} = \cs{\translExpr{\expr_1}}$, the thesis follows due to the correspondence of the semantic clause of the operator $\fnot$ in Table~\ref{tab:constr_sem} and that of the operator $\x{not}$ in Table~\ref{tab:sem_expression}.
\item [$(\expr = \exprOperator (\expr_1, \expr_2))$.] 
Since $\translExpr{\expr} = \translExpr{\expr_1}$ $\x{getOp}(\exprOperator)\ \translExpr{\expr_2}$ and, by the induction hypothesis, $\exprSem{\expr_1}{\req} = \cs{\translExpr{\expr_1}}$ and  $\exprSem{\expr_2}{\req} = \cs{\translExpr{\expr_2}}$, the thesis follows due to the correspondence of the semantic clause of the expression operator $\exprOperator$ in Table~\ref{tab:constr_sem} and that of the constraint operator $\x{getOp}(\exprOperator)$ in Table~\ref{tab:sem_expression}.
\end{description}
\vspace*{-.5cm}
\end{proof}

\medskip

\begin{lemma}
\label{lemma:obl}
For all $\ob \in \mathit{Obligation}$ and $\req \in R$ it holds that
$$
\oblSem{\ob}{\req} = \fo \ \ \Leftrightarrow \ \ \cs{\translObl{\ob}} = \true
$$
and
$$
\oblSem{\ob}{\req} = \err \ \ \Leftrightarrow \ \ \cs{\translObl{\ob}} = \false
$$
\end{lemma}
\begin{proof}
We only prove the ($\Rightarrow$) implication as the proof for the other direction proceeds in a specular way. Let $\ob\ = [\obType\ \mathit{pepAct}(\expr^*)]$
with $\expr^* = \expr_1 \ldots \expr_n$. By the clause~(\ref{cstr:obl}), it is translated into the constraint 
\begin{center}
$
c\ =\ 
\bigwedge_{\expr_j \in \expr^*} \lnot\isBot{\translExpr{\expr_j}} \wedge \lnot\isErr{\translExpr{\expr_j}} 
$
\end{center} 
We now proceed by case analysis on $\oblSem{\ob}{\req}$.
\begin{description}
\item [$(\oblSem{\ob}{\req} = \fo)$.] 
We have to prove that $\cs{c} = \true$. By the definition of $\denSemF{C}$, $\cs{c} = \true $ corresponds to 
$$
\begin{array}{l}
\forall j \in \{1,\ldots,n\}\  : \\
\qquad
\cs{\lnot\ \isBot{\translExpr{\expr_j}}} = \true\ \\
\qquad\quad 
\wedge\ \cs{\lnot\ \isErr{\translExpr{\expr_j}}}= \true 
\end{array}
$$
According to the constraint semantics of $\lnot$, $\mathtt{isMiss}$ and $\mathtt{isErr}$, this corresponds to
 $$
 \begin{array}{l}
\forall j \in \{1,\ldots,n\}\  : \\
\qquad
\cs{\translExpr{\expr_j}} \neq \excpt\\
\qquad \ \wedge\ \ \cs{\translExpr{\expr_j}} \neq \err
\end{array}
$$ 
By the hypothesis $\oblSem{\ob}{\req} = \fo$ and the clauses~(\ref{sem:obl}) and (\ref{sem:exp_con}), we have
$$
\exprSem{\expr^*}{\req} = \exprSem{\expr_1}{\req}\ \concat\: \ldots\ \concat\: \exprSem{\expr_n}{\req}  = \extVal_1\: \ldots \extVal_n
$$ 
where $\extVal_j$ stands for a literal value or a set of values. Thus, by Lemma~\ref{lemma:expr}, 
we get that 
$$\begin{array}{l}
\forall j \in \{1,\ldots,n\}\ : \\
\qquad  \cs{\translExpr{\expr_j}} = \extVal_j \not\in \{\excpt, \err\}
\end{array}
$$ 
which proves the thesis.
\smallskip

\item [$(\oblSem{\ob}{\req} = \err)$.] 
We have to prove that $\cs{c} = \false$. By the definition of $\denSemF{C}$, $\cs{c} = \false$ corresponds to 
$$
\begin{array}{l}
\exists j \in \{1,\ldots,n\}\  :\\
\qquad
\cs{\lnot\ \isBot{\translExpr{\expr_j}}} = \false\\
\qquad\quad
 \vee\ \cs{\lnot\ \isErr{\translExpr{\expr_j}}}= \false 
\end{array}
$$
According to the constraint semantics of $\lnot$, $\mathtt{isMiss}$ and $\mathtt{isErr}$, this corresponds to
$$
\begin{array}{@{}l}
\exists j \in \{1,\ldots,n\}\  :\\
\quad
\cs{\translExpr{\expr_j}} = \excpt\ \vee\ \ \cs{\translExpr{\expr_j}} = \err
\end{array}
$$ 
By the hypothesis $\oblSem{\ob}{\req} = \err$ and the clauses~(\ref{sem:obl}) and (\ref{sem:exp_con}), we have
$$
\begin{array}{l}
\exprSem{\expr^*}{\req} = \exprSem{\expr_1}{\req}\ \concat\: \ldots\ \concat\ \exprSem{\expr_n}{\req}  \neq \extVal^* \\
\qquad \Rightarrow\ \exists j \in \{1, \ldots,n\} : \exprSem{\expr_j}{\req} \in \{\excpt,\err\}
\end{array}
$$ 
Thus, by Lemma~\ref{lemma:expr}, we obtain that 
$$
\exists j \in \{1, \ldots,n\} \ : \ \cs{\translExpr{\expr_j}} \in \{\excpt,\err\}
$$ 
which proves the thesis.
\end{description}
\vspace*{-.5cm}
\end{proof}

\medskip

\begin{lemma}
\label{lemma:alg}
For all $\alg{\all} \in Alg$, $\req \in R$ and policies $\policy_1,\ldots,\policy_s \in \mathit{Policy}$ such that $\forall\,i\in\{1,\ldots,s\}\ : \ \policySem{\policy_i}{\req} = \langle \dec_i \ \fo_i^* \rangle \ \Leftrightarrow \ \cs{\translPol{\policy_i}\proj{\dec_i}} =  \true\,$, it holds that 
$$
\begin{array}{c}
\algSem{\alg{\all}, \policy_1\  \ldots\ \policy_s}{\req} = \langle \dec \ \fo^* \rangle \ \Leftrightarrow\ \\[.1cm] 
\cs{\translAlg{\alg{\all}, \policy_1\ \ldots \ \policy_s}\proj{\dec}} = \true
\end{array}
$$
\end{lemma}
\begin{proof}
Since the considered algorithms use the $\all$ instantiation strategy, by the hypothesis and the clauses~(\ref{sem:algA}) and~(\ref{cstr:alg}), the thesis is equivalent to prove that
$$
\begin{array}{@{}c}
\algOp (\ldots \algOp (\langle \dec_1\ \fo_1^* \rangle, \langle \dec_2 \ \fo_2^* \rangle),\ldots, \langle \dec_s \ \fo_s^* \rangle)\hfill\\
\hfill
= \langle \dec\ \fo^*\rangle\\[2pt]
\Longleftrightarrow\\[2pt]
\cs{\alg{}(\ldots \alg{}(\translPol{\policy_1}, \translPol{\policy_2}),\ldots,\translPol{\policy_s})\proj{\dec}} \\
\hfill
= \true
\end{array}
$$
The proof proceeds by case analysis on $\alg{}$.
In what follows, we only report the case of the $\permitOverO{}$ 
algorithm, as the other ones are similar and derive directly from the definitions in Appendix~\ref{sec:appendixA}
\smallskip

When $s=1$, we have $
\algOpAlg{\permitOverO{}} (\policySem{\policy_1}{\req})=\policySem{\policy_1}{\req}$ and $\permitOverO{}(\translPol{\policy_1}) = \translPol{\policy_1}$ by definition, hence the thesis directly follows from the hypothesis that $\policySem{\policy_1}{\req} = \langle \dec_1 \ \fo_1^* \rangle \ \Leftrightarrow \ \cs{\translPol{\policy_1}\proj{\dec_1}} =  \true\,$. 
For the remaining cases, we proceed by induction on the number $s$ of policies to combine.  
\begin{description}
\item [Base Case $(s=2)$.] 
We must prove that
$$
\begin{array}{c}
\algOpAlg{\permitOverO{}} (\langle \dec_1\ \fo_1^* \rangle, \langle \dec_2 \ \fo_2^* \rangle) = \langle \dec\ \fo^*\rangle\\[2pt]
\Leftrightarrow\\[2pt]
\cs{\permitOverO{}(\translPol{\policy_1}, \translPol{\policy_2})\proj{\dec}} = \true\,.
\end{array}
$$
For the sake of simplicity, in the following we omit the sequences of instantiated obligations, as their combination does not affect the decision $\dec$ returned by $\algOpAlg{\permitOverO{}}$.
We proceed by case analysis on the decision $\dec$.
\begin{description}
\item [$(\dec = \permit)$.] 
It follows that $\dec_1 = \permit$ or $\dec_2 = \permit$. Moreover, by definition we have $\permitOverO{}(\translPol{\policy_1}, \translPol{\policy_2})\proj{p} = \translPol{\policy_1} \proj{p}\ \vee\ \translPol{\policy_2} \proj{p}$.
\smallskip

\item [$(\dec = \deny)$.] 
It follows that $\dec_1, \dec_2 \in \{\deny, \notApp\}$. Moreover, by definition we have $\permitOverO{}(\translPol{\policy_1}, \translPol{\policy_2})\proj{d} = (\translPol{\policy_1} \proj{d} \wedge\, \translPol{\policy_2} \proj{d}) \vee\ \mbox{$(\translPol{\policy_1} \proj{d} \wedge\ \translPol{\policy_2} \proj{n})$} \vee\, (\translPol{\policy_1} \proj{n} \wedge\ \translPol{\policy_2} \proj{d})$.
\smallskip

\item [$(\dec = \notApp)$.] 
It follows that $\dec_1 = \dec_2 = \notApp$. Moreover, by definition we have $\permitOverO{}(\translPol{\policy_1}, \translPol{\policy_2})\proj{n} = \translPol{\policy_1} \proj{n}\, \wedge\ \translPol{\policy_2} \proj{n}$.
\smallskip

\item [$(\dec = \indet)$.] 
It follows that $\dec_1 = \indet$ or $\dec_2 = \indet$ and $\dec_1, \dec_2 \neq \permit$. Moreover, by definition we have $\permitOverO{}(\translPol{\policy_1}, \translPol{\policy_2})\proj{i} = (\translPol{\policy_1} \proj{i} \wedge\, \lnot \translPol{\policy_2} \proj{p})$ $\vee\, (\lnot \translPol{\policy_1} \proj{p} \wedge\ \translPol{\policy_2} \proj{i})$.
\smallskip

\end{description}
In any case, thesis follows from the hypothesis on $\translPol{\policy_i}$ and the definition of $\denSemF{C}$.

\smallskip

\item [Inductive Case $(s=k+1)$.] 
By the induction hypothesis the thesis holds for $k$ policies, that is
$$
\begin{array}{@{}c}
\algOp (\ldots \algOp (\langle \dec_1\ \fo_1^* \rangle, \langle \dec_2 \ \fo_2^* \rangle),\ldots, \langle \dec_k \ \fo_k^* \rangle)
\\
\hfill
= \langle \dec\ \fo^*\rangle\\[2pt]
\Longleftrightarrow\\[2pt]
\cs{\alg{}(\ldots \alg{}(\translPol{\policy_1}, \translPol{\policy_2}),\ldots,\translPol{\policy_k})\proj{\dec}} \\
\hfill
= \true
\end{array}
$$
The thesis then follows by repeating the case analysis on decision $\dec$ of the `Base Case' once we replace $\langle \dec_1\ \fo_1^*\rangle$, $\langle \dec_2\ \fo_2^*\rangle$, $\translPol{\policy_1}$ and $\translPol{\policy_2}$ by $\langle \dec'\ \fo'^*\rangle$, $\langle \dec_s\ \fo_s^*\rangle$, $\permitOverO{}(\ldots \permitOverO{}(\translPol{\policy_1}, \translPol{\policy_2}),\ldots,\translPol{\policy_k})$ and $\translPol{\policy_{s}}$, respectively. 
\end{description}
\vspace*{-.5cm}
\end{proof}

\medskip

\noindent
\textsc{Theorem}~\ref{thr:constr_sem} [Policy Semantic Correspondence]
For all $\policy \in \mathit{Policy}$ enclosing combining algorithms only using $\all$ as instantiation strategy, and $\req \in R$, it holds that
$$
\policySem{\policy}{\req}= \langle \dec \ \fo^* \rangle
\ \ \Leftrightarrow \ \ 
\cs{\translPol{\policy}\proj{\dec}} =  \true
$$
\begin{proof}
The proof proceeds by induction on the depth $i$ of $\policy$. 
\begin{description}

\item [Base Case $(i=0)$.] This means that $\policy$ is of the form $\ruleOpt{\effect\ \ \x{target:} \, \expr\ \ \x{obl:} \, \ob^{*} \,}$.
We proceed by case analysis on $\dec$. 
\begin{description}
\item [$(\dec = \permit)$.] By the clause~(\ref{sem:rule}), it follows that
$$
\exprSem{\expr}{\req} = \true\ \wedge\ \oblSem{\ob^{*}}{\req} = \fo^*
$$
Thus, by Lemma~\ref{lemma:expr}, it follows that
$$
\cs{\translExpr{\expr}} = \true
$$
and, by Lemma~\ref{lemma:obl} and the clause~(\ref{cstr:obl}), it follows that
$$
\cs{\translObl{\ob^{*}}} = \true
$$
On the other hand, by the clause~(\ref{cstr:rule}), we have that
$$
\begin{array}{l}
\translPol{\ruleOpt{\permit\ \ \x{target:} \, \expr\ \ \x{obl:} \, \ob^{*} \,}}\proj{p} = \\
\qquad\qquad\qquad\qquad
\translExpr{\expr} \wedge \translObl{\ob^{*}}
\end{array}
$$
Hence, by the definition of $\denSemF{C}$, we can conclude that
$$
\begin{array}{@{}r@{}c@{}l@{}}
\cs{\translPol{\ruleOpt{\permit\ \ \x{target:} \, \expr\ \ \x{obl:} \, \ob^{*} \,}}\proj{p}} & = &\\[.1cm] \cs{\translExpr{\expr}} \wedge \cs{\translObl{\ob^{*}}} & = &\\ 
\true \wedge \true & = & \true
\end{array}
$$
which proves the thesis.
\smallskip

\item [$(\dec = \deny)$.] We omit the proof since it proceeds like the previous case.
\smallskip

\item [$(\dec = \notApp)$.] By the clause~(\ref{sem:rule}), it follows that
$$
\exprSem{\expr}{\req} = \false\ \vee\  \exprSem{\expr}{\req} = \excpt
$$ 
By the clause~(\ref{cstr:rule}), we have that
$$
\translPol{\ruleOpt{\effect\ \ \x{target:} \, \expr\ \ \x{obl:} \, \ob^{*} \,}}\proj{n} = \lnot\: \translExpr{\expr}
$$ 
Hence, the thesis directly follows by Lemma~\ref{lemma:expr} and the definition of $\denSemF{C}$.
\smallskip 

\item [$(\dec = \indet)$.] By the clause~(\ref{sem:rule}), the $\mathtt{otherwise}$ condition holds, that is 
$$
\begin{array}{l}
\lnot (\exprSem{\expr}{\req} = \true\ \wedge\ \oblSem{\ob^{*}}{\req} = \fo^*) \\
\wedge\ \lnot (\exprSem{\expr}{\req} = \false\ \vee\  \exprSem{\expr}{\req} = \excpt)
\end{array}
$$
By applying standard boolean laws and reasoning on function codomains, this condition can be rewritten as follows
$$
\begin{array}{@{}l}
\lnot (\exprSem{\expr}{\req} = \true\ \wedge\ \oblSem{\ob^{*}}{\req} = \fo^*) \\
\quad
\wedge\ \lnot (\exprSem{\expr}{\req} = \false\ \vee\  \exprSem{\expr}{\req} = \excpt)\\
= (\exprSem{\expr}{\req} \neq \true\, \vee\, \oblSem{\ob^{*}}{\req} = \err) \\
\quad
 \wedge \, (\exprSem{\expr}{\req} \not\in \{ \false, \excpt\} )\\ 
= \exprSem{\expr}{\req} \not\in \{\true, \false, \excpt\} 
\vee (\exprSem{\expr}{\req} \not\in \{\false, \excpt\}\, \\
\quad
\wedge\, \oblSem{\ob^{*}}{\req} = \err)\\
= \exprSem{\expr}{\req} \not\in \{\true, \false, \excpt\}\ \vee\\
\quad (\exprSem{\expr}{\req} \not\in \{\true, \false, \excpt\}\, \wedge\, \oblSem{\ob^{*}}{\req} = \err)\ \vee\\
\quad (\exprSem{\expr}{\req} = \true \, \wedge\, \oblSem{\ob^{*}}{\req} = \err)\\
= \exprSem{\expr}{\req} \not\in \{\true, \false, \excpt\}  \\
\quad
\vee (\exprSem{\expr}{\req} = \true \, \wedge\, \oblSem{\ob^{*}}{\req} = \err)
\end{array}
$$
On the other hand, by the clause~(\ref{cstr:rule}), we have that
$$
\begin{array}{l}
\translPol{\ruleOpt{\effect\ \ \x{target:} \, \expr\ \ \x{obl:} \, \ob^{*} \,}}\proj{i} = 
\\[.1cm]
\quad \lnot\ (\isBool{\translExpr{\expr}}\, \vee\,\isBot{\translExpr{\expr}})
\\ 
\quad
\vee\ (\translExpr{\expr} \wedge\, \lnot\, \translObl{\ob^{*}})
\end{array}
$$
The thesis then follows by Lemmas~\ref{lemma:expr} and~\ref{lemma:obl} and the definition of $\denSemF{C}$.
\end{description}
\smallskip 


\item [Inductive Case $(i=k+1)$.] $\policy$ is of the form $\{ \alg{\all}\ \ \x{target:} \, \expr\ \ \x{policies:} \, (\policy^{+})^k \, {\x{\oblp:} \, \ob_{p}^{*}\ \ \x{\obld:} \, \ob_{d}^{*}} \, \}$.
We proceed by case analysis on $\dec$. 
\begin{description}
\item [$(\dec = \permit)$.] 
By the clause~(\ref{sem:pol}), it follows that
$$
\begin{array}{l}
\exprSem{\expr}{\req} = \true \\
 \wedge\ \algSem{\alg{\all}, (\policy^{+})^k}{\req} = \langle \permit\ \ \fo_1^*\rangle\\
\wedge\ \oblSem{{\ob_{p}^{*}}}{\req} = \fo_2^*
\end{array}
$$ 
Thus, by Lemma~\ref{lemma:expr}, it follows that 
$$
\exprSem{\expr}{\req}  =\ \cs{\translExpr{\expr}} = \true
$$
and, by Lemma~\ref{lemma:obl} and the clause~(\ref{cstr:obl}), it follows that
$$
\cs{\translObl{{\ob_{p}^{*}}}} = \true
$$
Since by the induction hypothesis, for all $\policy^h_i$ in $(\policy^+)^k$ with $h \leq k$, it holds that
$$
\policySem{\policy_i^h}{\req} = \langle \dec_i \ \fo^* \rangle \ \ \Leftrightarrow \ \ \cs{\translPol{\policy_i^h}\proj{\dec_i}} =  \true
$$ 
then, by Lemma~\ref{lemma:alg}, it follows that 
$$
\translAlg{\alg{\all},(\policy^+)^k}\proj{p} = \true
$$
On the other hand, by the clause~(\ref{cstr:pol}), we have that
$$
\begin{array}{@{}l@{}}
\translSymbol_{P}\{\!|
\{ \alg{\all}\ \x{target:} \, \expr\ \x{policies:} \, (\policy^{+})^k \, \\
\qquad\qquad\qquad\qquad\qquad\ \ {\x{\oblp:} \, \ob_{p}^{*}\ \ \x{\obld:} \, \ob_{d}^{*}} \, \}
|\!\}
\proj{p}\\[.1cm]
= \translExpr{\expr} \wedge\ \translAlg{\alg{\all},(\policy^+)^k}\proj{p} \wedge\ \translObl{{\ob_{p}^{*}}} 
\end{array}
$$
Hence, by the definition of $\denSemF{C}$, we can conclude that
$$
\begin{array}{@{}l@{}}
\constrSem[\![\translSymbol_{P}\{\!|
\{ \alg{\all}\ \x{target:} \, \expr\ \x{policies:} \, (\policy^{+})^k \, \\
\qquad\qquad\qquad\qquad\qquad\quad\ \ {\x{\oblp:} \, \ob_{p}^{*}\ \ \x{\obld:} \, \ob_{d}^{*}} \, \}
|\!\}
\proj{p}
]\!]{\req}\\[.1cm] 
= \cs{\translExpr{\expr}} \wedge \cs{\translAlg{\alg{\all},(\policy^+)^k}\proj{p}}\\
\quad \wedge\ \cs{\translObl{{\ob_{p}^{*}}}}\\[.1cm] 
= \true \wedge \true \wedge \true  =  \true
\end{array}
$$
which proves the thesis.
\smallskip

\item [$(\dec = \deny)$.] We omit the proof since it proceeds like the previous case.
\smallskip

\item [$(\dec = \notApp)$.] By the clause~(\ref{sem:pol}), it follows that
$$
\begin{array}{@{}l@{}}
\exprSem{\expr}{\req} = \false \, \vee\, \exprSem{\expr}{\req} = \excpt \\
\vee\ (\exprSem{\expr}{\req} = \true\ \wedge\ \algSem{\alg{\all}, (\policy^{+})^k}{\req} = \notApp)
\end{array}
$$
By the clause~(\ref{cstr:pol}), we have that
$$
\begin{array}{@{}l@{}}
\translSymbol_{P}\{\!|
\{ \alg{\all}\ \x{target:} \, \expr\ \x{policies:} \, (\policy^{+})^k \, \\
\qquad\qquad\qquad\qquad\qquad\ \ {\x{\oblp:} \, \ob_{p}^{*}\ \ \x{\obld:} \, \ob_{d}^{*}} \, \}
|\!\}
\proj{n}\\[.1cm]
= \lnot\ \translExpr{\expr} \vee (\translExpr{\expr} \wedge \translAlg{\alg{\all},(\policy^+)^k}\proj{n})
\end{array}
$$
The thesis then directly follows by Lemmas~\ref{lemma:expr} and~\ref{lemma:alg}, due to the induction hypothesis and the definition of $\denSemF{C}$.
\smallskip

\item [$(\dec = \indet)$.] 
By the clause~(\ref{sem:pol}), the $\mathtt{otherwise}$ condition holds, that is 
\[
\begin{array}{@{\!\!\!\!\!\!}l@{}}
\lnot (\exprSem{\expr}{\req}=\true\ \wedge\ \algSem{\alg{\all}, (\policy^{+})^k}{\req}= \langle \effect\ \  \foS_1 \rangle 
 \\[.1cm]
\quad
\wedge\  (\algSem{\alg{\all}, (\policy^{+})^k}{\req}= \langle \effect\ \  \foS_1 \rangle \Longrightarrow \oblSem{{\ob_{\effect}^{*}}}{\req} = \foS_2) \
 \\[.2cm]
\wedge\ \lnot (\exprSem{\expr}{\req}=\false \vee\, \exprSem{\expr}{\req}=\ \excpt  
 \\[.1cm]
\quad
\vee\, (\exprSem{\expr}{\req}=\true \ \wedge \ \algSem{\alg{\all}, (\policy^{+})^k}{\req}=\notApp))
\\
\mbox{ }
 \\[-.2cm]
\end{array}
\hspace*{-1cm}
\tag{C}
\label{othcondition}
\]
where,
to recall the connection between the effect returned by the combining algorithm and the sequence of obligations that is instantiated, we exploit the tautology
$$
\begin{array}{@{\!\!}l@{}}
\algSem{\alg{\all}, (\policy^{+})^k}{\req}= \langle \effect\ \ \foS_1 \rangle
\wedge
\oblSem{{\ob_{\effect}^{*}}}{\req} = \foS_2\\
= \\[.1cm]
\algSem{\alg{\all}, (\policy^{+})^k}{\req}= \langle \effect\ \ \foS_1 \rangle\\
\wedge (\algSem{\alg{\all}, (\policy^{+})^k}{\req}= \langle \effect\ \ \foS_1 \rangle \Longrightarrow \oblSem{{\ob_{\effect}^{*}}}{\req} = \foS_2)
\end{array}
$$
By applying standard boolean laws and reasoning on function codomains, the Condition~(\ref{othcondition}) above can be rewritten as follows
$$
\begin{array}{@{}l@{}}
\lnot (\exprSem{\expr}{\req}=\true\ \wedge\ \algSem{\alg{\all}, (\policy^{+})^k}{\req}= \langle \effect\ \  \foS_1 \rangle \\
\quad
\wedge\  (\algSem{\alg{\all}, (\policy^{+})^k}{\req}= \langle \effect\ \  \foS_1 \rangle \Longrightarrow \oblSem{{\ob_{\effect}^{*}}}{\req} = \foS_2) \
 \\
\wedge\ \lnot (\exprSem{\expr}{\req}=\false \vee\, \exprSem{\expr}{\req}=\ \excpt \\
\quad \vee\, (\exprSem{\expr}{\req}=\true \ \wedge \ \algSem{\alg{\all}, (\policy^{+})^k}{\req}=\notApp))\\[.1cm]
=\\
( \exprSem{\expr}{\req} \neq \true\, \vee\, \algSem{\alg{\all}, (\policy^{+})^k}{\req} \in \{ \notApp, \indet\}\\
\quad \vee\ (\algSem{\alg{\all}, (\policy^{+})^k}{\req}= \langle \effect\ \  \foS_1 \rangle \wedge \oblSem{{\ob_{\effect}^{*}}}{\req} = \err) )\\
\wedge\ (\exprSem{\expr}{\req} \not\in \{ \false, \excpt\} \\
\quad \wedge ( \exprSem{\expr}{\req} \neq \true\, \vee\,
\algSem{\alg{\all}, (\policy^{+})^k}{\req} 
\neq \notApp))\\[.1cm]
=\\[.1cm]
( \exprSem{\expr}{\req} \neq \true\, \vee\, \algSem{\alg{\all}, (\policy^{+})^k}{\req} \in \{ \notApp, \indet\}\\
\quad \vee\ (\algSem{\alg{\all}, (\policy^{+})^k}{\req}= \langle \effect\ \  \foS_1 \rangle \wedge \oblSem{{\ob_{\effect}^{*}}}{\req} = \err) )\\
\wedge 
(\exprSem{\expr}{\req} \not\in \{ \true, \false, \excpt\}\\
\quad \vee ( \exprSem{\expr}{\req} \not\in \{ \false, \excpt\} \wedge\\
\quad
\algSem{\alg{\all}, (\policy^{+})^k}{\req} 
\neq \notApp))\\[.1cm]
=\\[.1cm]
\quad \exprSem{\expr}{\req} \not\in \{ \true, \false, \excpt\}\\
\vee (\exprSem{\expr}{\req} \not\in \{ \true, \false, \excpt\} \\
\quad \wedge\ \algSem{\alg{\all}, (\policy^{+})^k}{\req} \neq \notApp )\\
\vee (\exprSem{\expr}{\req} \not\in \{ \true, \false, \excpt\} \\
\quad \wedge\ \algSem{\alg{\all}, (\policy^{+})^k}{\req} \in \{ \notApp, \indet\} )\\
\vee (\exprSem{\expr}{\req} \not\in \{ \false, \excpt\} \\
\quad \wedge\ \algSem{\alg{\all}, (\policy^{+})^k}{\req} \neq \notApp \\
\quad \wedge\ \algSem{\alg{\all}, (\policy^{+})^k}{\req} \in \{ \notApp, \indet\} )\\
\vee (\exprSem{\expr}{\req} \not\in \{ \true, \false, \excpt\} \wedge \algSem{\alg{\all}, (\policy^{+})^k}{\req}= \langle \effect\ \  \foS_1 \rangle\\
\quad \wedge\ \oblSem{{\ob_{\effect}^{*}}}{\req} = \err )\\
\vee (\exprSem{\expr}{\req} \not\in \{ \false, \excpt\} \wedge \algSem{\alg{\all}, (\policy^{+})^k}{\req} \neq \notApp\\ 
\quad \wedge\ \algSem{\alg{\all}, (\policy^{+})^k}{\req} = \langle \effect\ \  \foS_1 \rangle \wedge \oblSem{{\ob_{\effect}^{*}}}{\req} = \err )\\[.1cm]
=\\[.1cm]
\quad \exprSem{\expr}{\req} \not\in \{ \true, \false, \excpt\}\\
\vee (\exprSem{\expr}{\req} \not\in \{ \false, \excpt\} \wedge \algSem{\alg{\all}, (\policy^{+})^k}{\req} = \indet )\\
\vee (\exprSem{\expr}{\req} \not\in \{ \false, \excpt\} \\
\quad\ \wedge\ \algSem{\alg{\all}, (\policy^{+})^k}{\req} = \langle \effect\ \  \foS_1 \rangle \wedge \oblSem{{\ob_{\effect}^{*}}}{\req} = \err )\\[.1cm]
=\\[.1cm]
\quad \exprSem{\expr}{\req} \not\in \{ \true, \false, \excpt\}\\
\vee 
(\exprSem{\expr}{\req} \not\in \{ \true, \false, \excpt\} \wedge \algSem{\alg{\all}, (\policy^{+})^k}{\req} = \indet )\\
\vee 
(\exprSem{\expr}{\req} = \true \wedge \algSem{\alg{\all}, (\policy^{+})^k}{\req} = \indet )\\
\vee 
(\exprSem{\expr}{\req} \not\in \{ \true, \false, \excpt\} \\
\quad \wedge \algSem{\alg{\all}, (\policy^{+})^k}{\req} = \langle \effect\ \  \foS_1 \rangle \wedge \oblSem{{\ob_{\effect}^{*}}}{\req} = \err )\\
\vee 
(\exprSem{\expr}{\req} = \true \\
\quad 
\wedge \algSem{\alg{\all}, (\policy^{+})^k}{\req} = \langle \effect\ \  \foS_1 \rangle \wedge \oblSem{{\ob_{\effect}^{*}}}{\req} = \err )\\[.1cm]
=\\[.1cm]
\end{array}
$$
$$
\begin{array}{@{}l@{}}
\quad \exprSem{\expr}{\req} \not\in \{ \true, \false, \excpt\}\\
\vee 
(\exprSem{\expr}{\req} = \true \wedge \algSem{\alg{\all}, (\policy^{+})^k}{\req} = \indet )\\
\vee 
(\exprSem{\expr}{\req} = \true \wedge \algSem{\alg{\all}, (\policy^{+})^k}{\req} = \langle \effect\ \  \foS_1 \rangle \\
\quad \wedge \oblSem{{\ob_{\effect}^{*}}}{\req} = \err )\\
=\\[.1cm]
\quad \exprSem{\expr}{\req} \not\in \{\true, \false, \excpt\} \\
\vee (\exprSem{\expr}{\req} = \true\: \wedge\: \algSem{\alg{\all}, (\policy^{+})^k}{\req} = \indet\ )\\
\vee 
(\exprSem{\expr}{\req} = \true\: \wedge\: \algSem{\alg{\all}, (\policy^{+})^k}{\req} = \langle \permit\ \ \foS_1 \rangle \\
\quad \wedge \oblSem{{\ob_{p}^{*}}}{\req} = \err)\\
\vee  
(\exprSem{\expr}{\req} = \true\: \wedge\: \algSem{\alg{\all}, (\policy^{+})^k}{\req} = \langle \deny\ \ \foS_1 \rangle \\
\quad \wedge \oblSem{{\ob_{d}^{*}}}{\req} = \err)
\end{array}
$$
where the last step exploits the fact that $\effect \in \{\permit, \deny\}$.

On the other hand, by the clause~(\ref{cstr:pol}), we have that
$$
\begin{array}{@{}l@{}}
\translSymbol_{P}\{\!|
\{ \alg{\all}\ \x{target:} \, \expr\ \x{policies:} \, (\policy^{+})^k \, \\
\qquad\qquad\qquad\qquad\qquad\ \ {\x{\oblp:} \, \ob_{p}^{*}\ \ \x{\obld:} \, \ob_{d}^{*}} \, \}
|\!\}
\proj{i}
= \\[.1cm]
\ \ \lnot\ (\isBool{\translExpr{\expr}}\ \vee\ \isBot{\translExpr{\expr}})\\
\ \ \vee\ (\translExpr{\expr} \wedge \ \translAlg{\algSyntax,(\policy^+)^k}\proj{i}) \\
\ \ \vee\ (\translExpr{\expr} \wedge\ \translAlg{\algSyntax,(\policy^+)^k}\proj{p} \wedge \, \lnot\ \translObl{{\ob_{p}^{*}}} \,) \\
\ \ \vee\ (\translExpr{\expr} \wedge\ \translAlg{\algSyntax,(\policy^+)^k} \proj{d} \wedge \, \lnot\ \translObl{{\ob_{d}^{*}}} \,)
\end{array}
$$
The thesis then follows by Lemmas~\ref{lemma:expr},~\ref{lemma:obl} and~\ref{lemma:alg}, due to the induction hypothesis and the definition of $\denSemF{C}$.

\end{description}

\end{description}
\vspace*{-.5cm}
\end{proof}


\acrodef{ABAC}{Attribute Based Access Control}
\acrodef{RBAC}{Role-Based Access Control}
\acrodef{PBAC}{Policy-Based Access Control}
\acrodef{XACML}{eXtensible Access Control Markup Language}
\acrodef{FACPL}{Formal Access Control Policy Language}
\acrodef{PEP}{Policy Enforcement Point}
\acrodef{PDP}{Policy Decision Point}
\acrodef{PAS}{Policy Authorisation System}
\acrodef{PR}{Policy Repository}

\end{document}